\documentclass[10pt]{article}
\usepackage{fullpage}
\usepackage[bookmarks]{hyperref}
\usepackage{amssymb}
\usepackage{amsmath}
\usepackage{amsthm}
\usepackage[abs]{overpic}
\usepackage{tcolorbox}
\usepackage{graphicx,color,colordvi}
\usepackage{bbm}
\usepackage{cleveref}
\usepackage{stmaryrd}
\usepackage[utf8]{inputenc}
\usepackage{fullpage}
\usepackage[blocks]{authblk}
\usepackage{dsfont}
\usepackage{mathtools}
\usepackage{authblk}

\usepackage{centernot}
\usepackage{pgfplots}

\def\Re{{\operatorname{Re}}}

\def\openone{\leavevmode\hbox{\small1\kern-3.8pt\normalsize1}}

\def\RR{\mathbb{R}}

\def\NN{\mathbb{N}}

\def\11{\mathbb{I}}

\def\sln{\succeq_{\operatorname{l.n.}}}

\newtheorem{definition}{Definition}[section]
\newtheorem{proposition}[definition]{Proposition}
\newtheorem{lemma}[definition]{Lemma}

\newtheorem{theorem}[definition]{Theorem}
\newtheorem{corollary}[definition]{Corollary}

\newcommand{\tr}{\mathop{\rm Tr}\nolimits}

\usepackage{pst-node}
\usepackage{tikz-cd} 

\newcommand{\bra}[1]{\langle#1|}
\newcommand{\ket}[1]{|#1\rangle}

\newcommand{\cA}{{\cal A}}

\newcommand{\cD}{{\cal D}}

\newcommand{\cF}{{\mathcal{F}}}

\newcommand{\cN}{{\cal N}}

\newcommand{\cI}{{\cal I}}

\newcommand{\cL}{{\cal L}}

\newcommand{\cM}{{\mathcal{M}}}

\def\d{\mathrm{d}}

\def\sln{\succeq_{\operatorname{l.n.}}}

\usepackage{graphicx}
\usepackage{setspace}
\usepackage{verbatim}
\usepackage{subfig}

\numberwithin{equation}{section}

\DeclareRobustCommand\openone{\leavevmode\hbox{\small1\normalsize\kern-.33em1}}

\newcommand{\id}{{\rm{id}}}
\newcommand{\be}{\begin{equation}}
	\newcommand{\ee}{\end{equation}}
\newcommand{\bea}{\begin{eqnarray}}
	\newcommand{\eea}{\end{eqnarray}}
\newcommand{\beas}{\begin{eqnarray*}}
	\newcommand{\eeas}{\end{eqnarray*}}

\DeclareFontFamily{U}{mathx}{\hyphenchar\font45}
\DeclareFontShape{U}{mathx}{m}{n}{<-> mathx10}{}
\DeclareSymbolFont{mathx}{U}{mathx}{m}{n}
\DeclareMathAccent{\widebar}{0}{mathx}{"73}
\newcommand{\Renyi}{R{\'e}nyi~}
\newcommand{\Hell}{H}

\renewcommand{\d}{\textnormal{d}}
\newcommand{\cX}{{\cal X}}

\renewcommand{\Re}{D}
\newcommand{\Id}{{\mathds{1}}}
\DeclareMathAccent{\widehat}{0}{mathx}{"70}
\DeclareMathAccent{\widecheck}{0}{mathx}{"71}

\title{Quantum R\'enyi and $f$-divergences from integral representations}

\author[1,2]{Christoph Hirche\thanks{\href{mailto:christoph.hirche@gmail.com}{christoph.hirche@gmail.com}}}
\affil[1]{Zentrum Mathematik, Technical University of Munich, 85748 Garching, Germany}
\affil[2]{Centre for Quantum Technologies, National University of Singapore, Singapore}
\author[2,3]{Marco Tomamichel}
\affil[3]{Department of Electrical and Computer Engineering, National University of Singapore, Singapore}

\begin{document}

\maketitle

\begin{abstract}
Smooth Csisz\'ar $f$-divergences can be expressed as integrals over so-called hockey stick divergences. This motivates a natural quantum generalization in terms of quantum Hockey stick divergences, which we explore here. Using this recipe, the Kullback-Leibler divergence generalises to the Umegaki relative entropy, in the integral form recently found by Frenkel. 
We find that the R\'enyi divergences defined via our new quantum $f$-divergences are not additive in general, but that their regularisations surprisingly yield the Petz R\'enyi divergence for $\alpha < 1$ and the sandwiched R\'enyi divergence for $\alpha > 1$, unifying these two important families of quantum R\'enyi divergences. Moreover, we find that the contraction coefficients for the new quantum $f$ divergences collapse for all $f$ that are operator convex, mimicking the classical behaviour and resolving some long-standing conjectures by Lesniewski and Ruskai. We derive various inequalities, including new reverse Pinsker inequalites with applications in differential privacy and explore various other applications of the new divergences.
\end{abstract}

\section{Introduction}

Divergences are measures of dissimilarity between two statistical or quantum mechanical models and play a central role in information theory~\cite{csiszar2011information,liese2006divergences}. One reason for this is that most information measures with operational significance, like entropy and mutual information, can be expressed in terms of an underlying divergence~\cite{tomamichel2015quantum}. Divergences also appear directly in many mathematical arguments where their prominence derives from the fact that they are monotones under data processing. Probably the most prominent example of such a quantity is the relative entropy, with many applications including optimal rates in hypothesis testing or information transmission~\cite{hiai1991proper}. However, plenty of generalizations and other divergences like \Renyi divergences play an important role in information theory~\cite{mosonyi2015quantum}. As such, the study of a wide range of divergences and their quantum generalizations is of fundamental interest in quantum information theory.

In order to prove certain mathematical properties of divergences and relations amongst them, it often helps to find a convenient representation that allows one to look at the quantity from a new angle. For example, it has proven valuable to give integral representations of divergences. A large class of divergences, the Csisz\'ar $f$-divergences~\cite{csiszar63}, are parameterized by a twice differentiable convex function $f$ on non-negative reals with $f(1) = 0$ and defined as $D_f(P\|Q) = \sum_{x \in \cX} Q(x)\, f\left( P(x) / Q(x) \right)$.
Sason and Verd\'u recently gave an alternative representation for $f$-divergences as follows~\cite[Proposition 3]{sason2016f}. For two discrete probability mass functions $P$ and $Q$ with equal support on a set $\mathcal{X}$, we have
\begin{align}
    D_f(P\|Q) 
    &= \int_1^\infty  f''(\gamma) E_\gamma(P\|Q) + \gamma^{-3} f''(\gamma^{-1})E_\gamma(Q\|P) \, \d\gamma \,,
    \label{eq:integralformula}
\end{align}
where $E_\gamma$ is the so-called hockey-stick divergence, $E_{\gamma}(P\|Q) = \sum_{x \in \cX} \max \{  P(x) - \gamma Q(x) , 0 \}$, defined for all $\gamma \geq 1$~\cite{polyanskiy2010channel}. The above shows that every $f$-divergence can be written as an integral over hockey-stick divergences, establishing these as a fundamental quantity when investigating $f$-divergences and their properties. Consequently, as we will see, many properties of $f$-divergences can be derived directly from properties of the hockey-stick divergence in this representation.

Important examples of $f$-divergences are the Kullback-Leibler divergence, denoted by $D$, for $f: x \mapsto  x\log x$ and the Hellinger divergence, denoted by $\Hell_\alpha$, for $f: x \mapsto \frac{x^\alpha-1}{\alpha-1}$, respectively. Namely,
\begin{align}
    D(P\|Q) &=\int_1^\infty \frac{1}{\gamma} E_\gamma(P\|Q) + \frac{1}{\gamma^2} E_\gamma(Q\|P) \,\d\gamma, \\
    \Hell_\alpha(P\|Q) &= \alpha\int_1^\infty \gamma^{\alpha-2} E_\gamma(P\|Q) + \gamma^{-\alpha-1} E_\gamma(Q\|P) \,\d\gamma. \label{eq:Hellinger}
\end{align}
The latter is closely related to \Renyi divergences~\cite{renyi61}, denoted by $D_{\alpha}$, via
\begin{align}
    D_\alpha(P\|Q) = \frac{1}{\alpha-1} \log \big(1+(\alpha-1)\Hell_\alpha(P\|Q) \big) \label{eq:renyitransform}
\end{align}
An important special case of the Hellinger divergence is the $\chi^2$-divergence, which coincides with $H_2$. We refer to~\cite{{sason2016f}} for more examples and properties of classical $f$-divergences.

In this work we investigate a straight-forward quantum generalization of the $f$-divergence, denoted $D_f(\rho\|\sigma)$, using the expression in Equation~\eqref{eq:integralformula}, where we simply substitute the classical hockey-stick divergence with its quantum counterpart given by $E_{\gamma}(\rho\|\sigma) = \tr ( \rho - \gamma \sigma )_+$.\footnote{Here, $\tr (\cdot )_+$ and $\tr (\cdot)_-$ denote the sum over positive and negative eigenvalues of a Hermitian matrix, respectively.} The Hellinger distances $H_{\alpha}(\rho\|\sigma)$ and the corresponding \Renyi divergence $D_{\alpha}(\rho\|\sigma)$ are then defined via Equations~\eqref{eq:Hellinger} and~\eqref{eq:renyitransform}. These quantities are certainly well-defined when the supports of $\rho$ and $\sigma$ coincide, an assumption that we are willing to make in this introduction. We review properties of the quantum hockey-stick divergence and formally introduce our quantum generalization of the $f$-divergence for general states in Section~\ref{sec:fdivergence}. The quantities clearly obey many of the properties that we would expect from a quantum $f$-divergence. For example, joint convexity and the data-processing inequality (DPI) follow directly from the corresponding properties of the quantum hockey-Stick divergence. Clearly it also reduces to the classical $f$-divergence for commuting $\rho$ and $\sigma$. However, given the multitude of possible generalizations of classical $f$-divergences (see, e.g.,~\cite{petz1985quasi,petz1986quasi,petz1998contraction,matsumoto2018new,wilde2018optimized} for some prominent examples) what makes this one special? We will attempt to convince you that it deserves our attention because it possesses some remarkable properties. 

This investigation is inspired by recent investigations of integral formulae for the Umegaki relative entropy by Frenkel~\cite{frenkel2022integral} and Jen\v{c}ov\'a~\cite{jenvcova2023recoverability}. In particular, we show that the integral formula in~\cite[Theorem 6]{frenkel2022integral}, namely,
\begin{align}
    D(\rho\|\sigma) = \tr \big( \rho (\log \rho - \log \sigma) \big) = \int_{-\infty}^{\infty} \frac{1}{|t|(t-1)^2}\tr \big( (1-t)\rho+t\sigma \big)_{\!-}\, \d t
\end{align}
coincides with the new quantum $f$-divergence for $f: x \mapsto x \log x$. That is, we find a new integral representation for the Umegaki relative entropy:
\begin{align}
    D(\rho\|\sigma) = \int_1^\infty \frac{1}{\gamma} E_\gamma(\rho\|\sigma) + \frac{1}{\gamma^2}E_\gamma(\sigma\|\rho)  \, \d\gamma \,.
\end{align}
One of the questions this immediately raises is whether our generalization also agrees with known quantum $f$-divergences for other functions. We answer this partially in the affirmative by showing that the $\chi^2$-divergence emerging from the new definition has a form that was previously investigated, see e.g.~\cite[Section 4]{petz1998contraction} and~\cite[Example 1]{lesniewski1999monotone}, i.e., we find 
\begin{align}  
    \chi^2(\rho\|\sigma) \equiv H_2(\rho\|\sigma) 
    = \int_0^\infty \ \tr \big( \rho (\sigma + s\Id)^{-1} \rho (\sigma + s\Id)^{-1} \big)  \d s - 1 \,.
\end{align}
We can see already from this result that the corresponding \Renyi divergences, defined via Equation~\ref{eq:renyitransform}, do not generally concur with any of the known families of quantum \Renyi divergences. Moreover, to our distress, the new definitions are not generally additive for tensor product states.
However, one of our main results, discussed in Section~\ref{sec:renyi}, is that operationally meaningful quantum \Renyi divergences naturally emerge when we regularize this quantity. That is, we establish that
\begin{align}
     \lim_{n\to\infty} \frac{1}{n} D_\alpha \big(\rho^{\otimes n} \big\|\sigma^{\otimes n} \big) = \begin{cases}
         \widebar D_\alpha(\rho\|\sigma) &0<\alpha<1 \\
         \widetilde D_\alpha(\rho\|\sigma) &1<\alpha<\infty
         \end{cases}, \label{Eq:DI-closed-formula-intro}
\end{align}
where $\widebar D_\alpha$ is the Petz \Renyi divergence~\cite{petz1985quasi} and $\widetilde D_\alpha$ is the sandwiched \Renyi divergence~\cite{muller2013quantum,wilde2014strong}. Interestingly, the above equality mirrors the range of $\alpha$ for which these divergences most commonly find applications, in particular, in the description of error and strong converse exponents in asymmetric binary quantum hypothesis testing~\cite{mosonyi2015quantum}. There, perhaps surprisingly, the corresponding optimal expressions are given using the two different \Renyi divergences depending on the error regime. This is the first instance where both R\'enyi divergences emerge naturally, via regularization, from the same smooth parent quantity.\footnote{It has previously been suggested  that the right hand side of Equation~\eqref{Eq:DI-closed-formula-intro} is the operationally most meaningful generalization of the classical \Renyi divergence~\cite{mosonyi2015quantum}. It is worth noting, however, that this patched up quantity is not twice continuously differentiable~\cite{lintomamichel14}. More importantly, recent results, e.g.\ in~\cite{liyao23}, indicate that the sandwiched R\'enyi divergence has operational meaning for $\alpha < 1$ as well.}

For our second main application, we turn to contraction coefficients. Since contraction coefficients are intrinsically linked to the strong data-processing properties of their underlying divergence, they play an important role in many areas of quantum information processing. In general, for any $f$-divergence the strong data processing inequality (SDPI) constant is defined as
\begin{align}
     \eta_f(\cA; \sigma)  =  \sup_{\rho} \frac{D_f(\cA(\rho) \| \cA(\sigma))}{D_f(\rho\|\sigma)} \,. 
\end{align} 
and the corresponding contraction coefficient as $\eta_f(\cA) = \sup_\sigma \eta_f(\cA;\sigma)$. 
Important examples include the SDPI constant and contraction coefficient for the relative entropy, denoted by $\eta_{\Re}$, and for the $\chi^2$-divergence, denoted by $\eta_{\chi^2}$. Additionally we consider the contraction coefficient for the trace distance, denoted by $\eta_{\tr}$.  Naturally, it is of interest to find meaningful bounds between different contraction coefficients. In fact, in the classical setting such bounds are fairly well understood and perhaps surprisingly it is known that for a large class of functions $f$, namely operator convex functions, the corresponding contraction coefficient is indeed the same. The quantum equivalents of these questions have been intensely studied as well, however the structure emerging appears to be significantly more complicated then in the classical setting and many relationships remain unknown~\cite{temme2010chi,hiai2016contraction}.  

In Section~\ref{Sec:ContractionCoef} we will discuss relationships between contraction coefficients based on our integral $f$-divergences. Notably, we will show that most of the structure present in the classical case transfers also to the quantum setting. In particular, for $f$ with $0<f''(1)<\infty$, we have
\begin{align}
     \eta_{x^2}(\cA) \leq \eta_f(\cA) \leq \eta_{\tr}(\cA) \,, 
 \end{align}
which includes the useful inequality $\eta_{\Re}(\cA) \leq \eta_{\tr}(\cA)$ and hence resolves an important special case of~\cite[Conjecture 4.8]{lesniewski1999monotone}. Additionally, we also show that, analogous to the classical setting, for operator convex functions the contraction coefficients collapse. Important examples of this include the relative entropy and the Hellinger divergence for $0\leq\alpha\leq 2$. In particular, for operator convex $f$, we find
\begin{align}
    \eta_f(\cA) = \eta_{x^2}(\cA) = \eta_{\Re}(\cA) \,. \label{Eq:contraction-summary-Intro}
\end{align}
This can again be compared to~\cite[Conjecture 4.8]{lesniewski1999monotone}, where an analogous result was conjectured for the Petz $f$-divergence. At the end of the section we discuss the implications of our findings on partial orders between quantum channels. In particular, we give equivalences between the well known less noisy order and different criteria based on the new $f$-divergences.

Finally, in Section~\ref{sec:apps} we discuss further properties of the $f$-divergences that go beyond the questions discussed above, aiming to exploit the integral representation. Here we give the main results for the Umegaki relative entropy. We can exploit the representation to give reverse Pinsker-type inequalities, which take the form
\begin{align}
     D(\rho\|\sigma)  &\leq \left( \frac{e^{D_{\max}(\rho\|\sigma)}D_{\max}(\rho\|\sigma)}{e^{D_{\max}(\rho\|\sigma)}-1} + \frac{D_{\max}(\sigma\|\rho)}{1-e^{D_{\max}(\sigma\|\rho)}} \right)  \|\rho-\sigma\|_{\tr} \\ 
     &\leq \Xi(\rho\|\sigma)  \|\rho-\sigma\|_{\tr},
\end{align}
where $\Xi(\rho,\sigma) := \max\{D_{\max}(\rho\|\sigma), D_{\max}(\sigma\|\rho)\}$ is the Thompson metric, $D_{\max}(\rho\|\sigma)$ is the $\max$-relative entropy, and $\|\rho-\sigma\|_{\tr} := \frac12 \|\rho-\sigma\|_1$ is the trace distance. The first inequality is in fact an equality under a particular linearity condition that is satisfied for commuting qubit states.
More generally, our technique yields relatively simple upper bounds on $f$-divergences in terms of the trace distance. We will also see that these bounds have direct applications when discussing quantum differential privacy~\cite{zhou2017differential,hirche2022quantum}. As a second example, we give continuity bounds on the divergences of the form
\begin{align}
     D(\rho\|\sigma) - D(\tau\|\sigma) &\leq  \Omega(\rho,\sigma) \|\rho-\tau\|_{\tr} \,, \label{Eq:D-continuity-2-intro}
\end{align} 
where $\Omega(\rho,\sigma) := D_{\max}(\rho\|\sigma) + D_{\max}(\sigma\|\rho)$ is the Hilbert projective metric.
Note that these bounds, in contrast to most bounds found in the literature (see, e.g., \cite{bluhm2023continuity} and references therein), do not explicitly depend on the smallest eigenvalue of $\sigma$, but instead on the Hilbert projective metric. This can be advantageous especially when $\rho$ and $\sigma$ have eigenvalues that are suitably aligned.

The hockey-stick divergence is also closely related to the concept of differrential privacy~\cite{hirche2022quantum}. As our next application we leverage this connection to give bounds on divergences under differrentially private noise with an application to hypothesis testing. In particular we give a converse bound for a Stein's Lemma under differential privacy, generalizing a classical result in~\cite{asoodeh2021local}. Finally, we give upper bounds on amortized channel divergences for different $f$-divergences in terms of the diamond norm. 

\emph{Note added:} During the completion of this work we became aware of a recent update to~\cite{frenkel2022integral}, proposing a \Renyi divergence generalizing the integral representation in~\cite[Theorem 6]{frenkel2022integral}, see~\cite[Section 9]{frenkel2022integral}. Following the same proof as in Corollary~\ref{Cor:D-integral} one can show that their definition is equal to ours. The definition was not  further explored in~\cite{frenkel2022integral}.

\section{Quantum $f$ divergences from hockey-stick divergence}
\label{sec:fdivergence}

\subsection{Notation}

Throughout this work we consider only finite dimensional Hilbert spaces and linear operators acting on them. Amongst them, positive operators are those that have only non-negative eigenvalues and quantum states are positive operators with unit trace. We use standard notation in the field of quantum information theory. Quantum states are denoted by lower-case greek letters $\rho$, $\sigma$, $\tau$, etc. For a general hermitian operator $X$ we denote by $X_+$ and $X_-$ the positive and negative part of the eigen-decomposition of $X$, such that $X=X_+ - X_-$. We use $(\cdot)^{-1}$ to denote the generalized (Penrose) inverse. We write $\rho\ll\sigma$ if the support of $\rho$ is contained in the support of $\sigma$ and we write $\rho\ll\gg\sigma$ if the supports of $\rho$ and $\sigma$ coincide. By $\lambda_{\min}(\cdot)$ we denote the smallest and by $\lambda_{\max}(\cdot)$ the largest eigenvalue. 
Quantum channels are completely positive trace preserving maps from linear operators to linear operators. We use $\log(x)$ to denote the natural logarithm. 

\subsection{Quantum hockey-stick divergence}

The quantum generalization of the hockey-stick divergence first explicitly appeared in~\cite{sharma2012strong}. The $\gamma$-hockey-stick divergence is usually defined for any two quantum states $\rho$ and $\sigma$ and $\gamma\geq 1$ as 
$E_\gamma(\rho\|\sigma) := \tr(\rho-\gamma\sigma)_+$. Our definition is an extension to positive operators $A,B$ and all $\gamma\geq0$, 
\begin{align}\label{Eq:Def-HockeyStick}
    E_\gamma(A\|B) := \tr(A-\gamma B)_+ - (\tr(A-\gamma B))_+
\end{align}
Similar extension to $0\leq\gamma\leq 1$ have been considered in the classical setting, e.g. in~\cite{Zamanlooy2023}. 
Besides the application in strong converses for which the divergence was used in~\cite{sharma2012strong}, it has also recently been proven useful in the context of quantum differential privacy in~\cite{hirche2022quantum}. As the hockey-stick divergence is central to this work, we will now list its most important properties, which were mostly proven in~\cite{sharma2012strong, hirche2022quantum}. Here we also generalize some of the results to positive operators and the extended range of $\gamma$. 

\begin{lemma}\label{Lem:HS-Properties}
Let $A, B, C$ be positive semidefinite operators and $\rho, \sigma, \tau$ quantum states. The following properties hold. 
\begin{enumerate}
    
    \item (Variational formula) For $\gamma\geq0$, we have
    \begin{align}
        E_\gamma(A\|B) &= \frac12\|A-\gamma B\|_1 +\frac12|\tr(A-\gamma B)| \\
        &= \sup_{0\leq M\leq \Id} \tr M(A-\gamma B)-(\tr(A-\gamma B))_+. 
    \end{align}    
    
    \item (Orthogonality) For $\gamma>0$, we have $E_{\gamma}(A\|B) = \tr A - (\tr(A-\gamma B))_+$ if and only if $A \perp B$, i.e., if $A$ and $B$ are orthogonal.
    
    \item (Positivity) We have for all $\gamma\geq 0$ that $E_\gamma(A\| B) \geq 0$. 
    \item (Monotonicity in $\gamma$) The function $\gamma \mapsto E_{\gamma}(A\|B)$ is monotonically decreasing for $\gamma \geq 1$
    with $E_1(\rho\|\sigma) = \frac12 \| \rho - \sigma\|_1$. 
    
    \item (Relation to max-divergence) For $\gamma\geq1$ we have, 
    \begin{align}\label{equ:dmax-hockey}
      D_{\max}(A\|B)\leq\log\gamma \quad\iff\quad E_\gamma(A\|B)=0, 
    \end{align}
    where $D_{\max}(\rho\|\sigma) := \inf\{\lambda : \rho\leq e^{\lambda}\sigma\}$.
    
    \item (DPI) For any positive trace preserving map $\cA$ and $\gamma \geq 0$, we have
    \begin{align}
        E_\gamma(\cA(A)\|\cA(B)) \leq E_\gamma(A\|B). 
    \end{align}

    \item{(Exchange of arguments} For $\gamma\geq 0$, we have
    \begin{align}
    E_\gamma(A\|B) = \gamma E_{\frac1\gamma}(B\|A).\label{Eq:HS-exchange}
    \end{align}

    \item{(Stability)} For $\gamma\geq 0$, we have
    \begin{align}  E_\gamma(A\otimes C\| B\otimes C) = \tr(C) E_\gamma(A\|B).
    \end{align}
    
    \item{(Strong convexity)} Let $\gamma_1,\gamma_2\geq 1$ and consider states $\rho = \sum_x p(x)\rho_x$ and  $\sigma = \sum_x q(x)\sigma_x$ for some probability mass functions $p$ and $q$. Then, we have
    \begin{align}
    E_{\gamma_1\gamma_2}(\rho\|\sigma) \leq \sum_x p(x) E_{\gamma_1}(\rho_x\|\sigma_x) + \gamma_1  E_{\gamma_2}(p\| q). 
    \end{align}
    This in particular also implies convexity and joint convexity.

     \item (Triangle inequality)   For $\gamma_1, \gamma_2 \geq 1$, we have
    \begin{align}
    E_{\gamma_1\gamma_2}(\rho\|\sigma) \leq E_{\gamma_1}(\rho\|\tau) + \gamma_1 E_{\gamma_2}(\tau\|\sigma). \label{Eq:Triangle}
    \end{align}

    \item{(Subadditivity)} For $\gamma_1,\gamma_2\geq 1$, we have
    \begin{align}\label{equ:subadditivity_gamma}
    E_{\gamma_1\gamma_2}(\rho_1\otimes\rho_2\|\sigma_1\otimes\sigma_2)
    &\leq E_{\gamma_1}(\rho_1\|\sigma_1) + \gamma_1 E_{\gamma_2}(\rho_2\|\sigma_2), \\
    E_{\gamma_1\gamma_2}(\rho_1\otimes\rho_2\|\sigma_1\otimes\sigma_2) 
    &\leq E_{\gamma_1}(\rho_2\|\sigma_2) + \gamma_1 E_{\gamma_2}(\rho_1\|\sigma_1) .
    \end{align}

    \item{(Trace distance bound)} For $\gamma\geq 1$, we have
    \begin{align}
    E_\gamma(\rho\|\sigma) + E_\gamma(\rho\|\tau) \geq \frac{\gamma}{2}\|\tau - \sigma\|_1 +(1-\gamma).
    \end{align}

\end{enumerate}
\end{lemma}
\begin{proof}
1. The first equality follows~\cite[Section II]{sharma2012strong} using $2(x-y)_+=|x-y|+(x-y)_+$. The second follows the proof in~\cite[Lemma II.1]{hirche2022quantum}. 
2. Follows from the definition and the assumption that $A,B$ are positive. 
3. This holds because $\tr(A-B)_+\geq (\tr(A-B))_+$. 
4. This follows from the definition and~\cite[Section II]{sharma2012strong}. 
5. This is by defintion, see also~\cite[Lemma II.6]{hirche2022quantum}. 
6. This is an extension of the proof of~\cite[Lemma 4]{sharma2012strong}, which we give here for completeness. Let $A-\gamma B=X_+-X_-$ for $X_+,X_-\geq0$ and define the projector $P_+=\{\cA(A)\geq \gamma\cA(B)\}$. First note that $\tr(\cA(A)-\gamma \cA(B))=\tr(A-\gamma B)$, because $\cA$ is trace preserving. Then we only need to show, 
\begin{align}
    &\tr(A-\gamma B)_+ = \tr X_+ = \tr\cA(X_+) \geq \tr P_+\cA(X_+) \\
    &\geq \tr P_+\cA(X_+-X_-) = \tr P_+\cA(A-\gamma B) = \tr(\cA(A)-\gamma \cA(B))_+
\end{align}
where all steps are either by definition or follow from the properties of the channel. This completes the proof. 
7. We give a brief proof that follows~\cite[Proposition II.5]{hirche2022quantum}), 
\begin{align}
    E_\gamma(A\|B) &= \sup_{0\leq M\leq \Id} \tr M(A-\gamma B)-(\tr(A-\gamma B))_+ \\
    &= \sup_{0\leq M\leq \Id} \tr (\Id-M)(A-\gamma B)-(\tr(A-\gamma B))_+ \\
    &= \tr(A-\gamma B) + \sup_{0\leq M\leq \Id} \tr M(\gamma B-A)-(\tr(A-\gamma B))_+ \\
    &=  \gamma\sup_{0\leq M\leq \Id} \tr M(B-\frac{1}{\gamma}A)  -\gamma(\tr(B-\frac1{\gamma} A))_+ \\
    &= \gamma E_{\frac1\gamma}(B\|A). 
\end{align}
8. This was proven in~\cite[Proposition II.5]{hirche2022quantum} for $\gamma\geq1$ and extends to all $\gamma\geq0$ because we showed above that also data processing holds for that range. 9.-12. These were all proven as given in~\cite[Proposition II.5]{hirche2022quantum}). 
\end{proof}
With these properties in mind we are all set up for the following sections.

\subsection{Integral representations of the Umegaki quantum relative entropy}\label{Sec:IntRep-RelEnt}

The quantum relative entropy is defined as, 
\begin{align}
    D(\rho\|\sigma) = \begin{cases}
        \tr\rho(\log\rho-\log\sigma) &\rho\ll\sigma \\
        \infty &\text{otherwise}
    \end{cases}. 
\end{align}
Recently, Frenkel~\cite{frenkel2022integral} gave the following integral representation of the relative entropy.  
\begin{theorem}[Theorem 6 in~\cite{frenkel2022integral}]
Let $\rho,\sigma$ be quantum states. Then
\begin{align}
    D(\rho\|\sigma) = \int_{-\infty}^{\infty} \frac{\d t}{|t|(t-1)^2}\tr((1-t)\rho+t\sigma)_-. 
\end{align}
\end{theorem}
This representation was then used in~\cite{jenvcova2023recoverability} to investigate the recoverability of quantum channels in the context of hypothesis testing. To that end~\cite[Corollary 1]{jenvcova2023recoverability} gave an alternative expression of the above integral representation. Here we take a similar route and give a slightly different new expression in terms of the quantum hockey-stick divergence. This generalizes the integral representation for the classical relative entropy given in~\cite{sason2016f} to the quantum setting. 
\begin{corollary}\label{Cor:D-integral}
Let $\rho,\sigma$ be quantum states. Then 
\begin{align}
    D(\rho\|\sigma) = \int_1^\infty \left(\frac{1}{s} E_s(\rho\|\sigma) + \frac{1}{s^2}E_s(\sigma\|\rho)\right) \d s. 
\end{align}
\end{corollary} 
\begin{proof}
    We start off similarly as~\cite{frenkel2022integral,jenvcova2023recoverability} by noting that whenever $0\leq t\leq 1$ then $\tr((1-t)\rho+t\sigma)_-=0$ and we can split the integral into two parts, integrating over $t\leq 0$ and $t\geq 1$. For the latter, we can entirely follow the argument in~\cite{jenvcova2023recoverability} where it was shown that
    \begin{align}
        \int_1^\infty \frac{\d t}{t(t-1)^2}\tr((1-t)\rho+t\sigma)_- = \int_1^\infty \frac{\d s}{s}\tr(\rho-s\sigma)_+ = \int_1^\infty \frac{\d s}{s} E_s(\rho\|\sigma), 
    \end{align}
    where~\cite{jenvcova2023recoverability} essentially substituted $s=\frac{t}{t-1}$ which shows the first equality. The second equality then follows from the definition of the hockey-stick divergence. 

    For the integral over $t\leq 0$, we take a slightly different route than~\cite{jenvcova2023recoverability}.
    \begin{align}
        \int_{-\infty}^0 \frac{\d t}{-t(t-1)^2}\tr((1-t)\rho+t\sigma)_- &=
        \int_{-\infty}^0 \frac{\d t}{(t-1)^2}\tr(\frac{(t-1)}{t}\rho-\sigma)_- \\
        &= \int_{1}^\infty \frac{\d s}{s^2}\tr(s\rho-\sigma)_- \\
        &= \int_{1}^\infty \frac{\d s}{s^2}\tr(\sigma-s\rho)_+ \\
        &= \int_{1}^\infty \frac{\d s}{s^2}E_s(\sigma\|\rho),
    \end{align}
    where the first equality follows because $-t\geq0$, the second is substituting $s=\frac{t-1}{t}$, the third uses $\tr A_-=\tr(-A)_+$. Putting both integrals back together concludes the proof. 
\end{proof}


\subsection{A new quantum $f$-divergence via integral representation} 

As generalization of the classical $f$-divergence, quantum $f$-divergences define a wide class of divergences including many commonly used divergences as special cases. However, due to the non-commutativity of quantum theory, there are many possible ways of defining quantum $f$-divergences that all reduce to the classical quantity. We will not discuss the plethora of definitions available in the literature in detail, but instead refer the reader to~\cite{petz1985quasi,petz1986quasi,petz1998contraction,matsumoto2018new,wilde2018optimized} for further reading. Instead, we will only discuss two quantum generalizations of the $f$-divergence that are known to be extremal  in the sense that they are smallest or largest amongst all quantum generalization that satisfy the DPI (see~\cite{gour2020optimal} for a general theory on minimal and maximal quantum extensions of information measures).

The minimal quantum $f$-divergence, also known as the measured $f$-divergence, is defined via
\begin{align}
     \widecheck{D}_{f}(\rho\|\sigma) := \sup_M D_f( P_{M,\rho} \| P_{M,\sigma})
    \label{eq:Df_measured}
\end{align}
where the supremum is over all quantum measurements and $P_{M,\rho}$ is the pdf induced by this measurement on the state $\rho$ using Born's rule. On the other hand, the maximal $f$-divergence~\cite{matsumoto2018new} is defined as
\begin{align}
     \widehat{D}_f(\rho\|\sigma) = \inf\{ D_f(p\|q) \,:\, \Phi(p)=\rho, \Phi(q)=\sigma \}, 
\end{align}
where the infimum is over all probability distributions $p, q$ and classical-quantum channels $\Phi$. It is easy to verify that any quantum generalization of $D_f$ must lie between $\widecheck{D}_{f}$ and $\widehat{D}_f$.

Motivated by our earlier discussion, we now make the following definition for a new quantum $f$-divergence. 
\begin{definition}
We denote by $\cF$ the set of functions $f:(0,\infty)\rightarrow\RR$ that are convex and twice differentiable with $f(1)=0$.
Let $f \in \cF$. Then, for any quantum states $\rho$ and $\sigma$, we define
\begin{align}
    D_f(\rho\|\sigma) := \int_1^\infty f''(\gamma) E_\gamma(\rho\|\sigma) + \gamma^{-3} f''(\gamma^{-1})E_\gamma(\sigma\|\rho)\, \d \gamma ,
\end{align}
whenever the integral is finite and $D_f(\rho\|\sigma) := +\infty$ otherwise.
\end{definition}
This is a valid quantum generalization of the classical $f$-divergence in the sense that for commuting states it reduces to the classical $f$-divergence.
We will see that it has many desirable properties and start with some fundamental observations collected in the following proposition. In the main text we only consider normalized quantum states and refer to Appendix~\ref{App:Scaling} for a discussion regarding general positive semi-definite hermitian operators. First note that one can give a simpler expression using Equation~\eqref{Eq:HS-exchange}.
\begin{lemma}
Let $f \in \cF$ and $\rho, \sigma$ two quantum states, then
 \begin{align}
    D_f(\rho\|\sigma) = \int_0^\infty f''(\gamma) E_\gamma(\rho\|\sigma)\, \d \gamma ,
\end{align}   
\end{lemma}
\begin{proof}
    Observe, 
    \begin{align}
        \int_1^\infty  \gamma^{-3} f''(\gamma^{-1})E_\gamma(\sigma\|\rho)\, \d \gamma 
        &= \int_1^\infty  \gamma^{-2} f''(\gamma^{-1})E_{\frac1\gamma}(\rho\|\sigma)\, \d \gamma \\
         &= \int_0^1  f''(\widehat\gamma)E_{\widehat\gamma}(\rho\|\sigma)\, \d \widehat\gamma, 
    \end{align} 
    where the first equality is by Equation~\eqref{Eq:HS-exchange} and the second by substituting $\widehat\gamma=\gamma^{-1}$. Giving this back into the definition of the $f$-divergence completes the proof. 
\end{proof}
However, it will often prove convenient to use the longer expression, because the Hockey-Stick divergence has particularly nice properties for $\gamma\geq1$. We continue by listing some additional properties of our $f$-divergence. 
\begin{proposition}\label{Prop:Df-properties}
    Let $f \in \cF$ and $\rho, \sigma$ two quantum states. The following properties hold:
    \begin{enumerate}
        
        \item (Trivial function) If $f(x) = b (x - 1)$ for some constant $b$ then $D_f(\rho\|\sigma) = 0$. 

        \item (Linear combinations) Define $h(x) = \alpha f(x) + \beta g(x)$ for $\alpha,\beta \geq 0$ and a function $g \in \cF$. Then,
        \begin{align}
                D_h(\rho\|\sigma) = \alpha D_f(\rho\|\sigma) + \beta D_g(\rho\|\sigma).
        \end{align}
        
        \item (Exchange of arguments) Define the function $g(t)=t f(\frac1t)$, then
            \begin{align}
                D_f(\rho\|\sigma) = D_g(\sigma\|\rho) .
            \end{align}
            The function $g$ is sometimes referred to as the $\star$-conjugate of $f$, see~\cite{liese1987convex}. Moreover, if $f$ itself is such that $f(t)=t f(\frac1t)$, then
            \begin{align}
                D_f(\rho\|\sigma) = D_f(\sigma\|\rho) =\int_1^\infty f''(\gamma) \left( E_\gamma(\rho\|\sigma) + E_\gamma(\sigma\|\rho)\right) \d \gamma \,. 
            \end{align}
         We say, $f$ is $\star$-self conjugate.   
        \item (Faithfullness) We have
        \begin{align}
            D_f(\rho\|\sigma) \geq 0,
        \end{align}
        with equality if $f''(x) = 0$ for all $x\in(e^{-D_{\max}(\sigma\|\rho)},  e^{D_{\max}(\rho\|\sigma)})$. In particular, equality always holds when $\rho = \sigma$.

        \item (Upper bound) We have
        \begin{align}
            D_f(\rho\|\sigma) \leq f(x) - x f'(x) \big|_{x = e^{-D_{\max}(\sigma\|\rho)}} + f' (x) \big|_{x = e^{D_{\max}(\rho\|\sigma)}} ,
        \end{align}
        where, in case $D_{\max}$ is not finite, the respective limits of the function at $0$ and $\infty$ are to be used. Furthermore, equality holds if and only if $\rho \perp \sigma$ (i.e., when $\rho$ and $\sigma$ are orthogonal). 
        
        \item (Data-processing inequality) $D_f(\rho\|\sigma)$ obeys the DPI, i.e., for any quantum channel $\cN$, we have
        \begin{align}
            D_f\big(\cN(\rho) \big\| \cN(\sigma)\big) \leq D_f(\rho\|\sigma).
        \end{align}
        
        \item (Joint convexity) $D_f$ is jointly convex, i.e., for states $\rho=\sum_x p(x) \rho_x$ and $\sigma=\sum_x p(x)\sigma_x$, we have
        \begin{align}
            D_f(\rho\|\sigma) \leq \sum_x p(x) D_f(\rho_x\|\sigma_x) \,.
        \end{align}

        \item (Relation to minimal and maximal quantum $f$-divergence) We have
        \begin{align}
            \widecheck{D}_{f}(\rho\|\sigma) \leq D_f(\rho\|\sigma) \leq \widehat{D}_f(\rho\|\sigma).
        \end{align}
    \end{enumerate}
\end{proposition}

\begin{proof}
    We establish the items separately.
    
    (1) In this case $f''(x) = 0$ and, thus, $D_f(\rho\|\sigma) = 0$.
    
    (2) This follows by the linearity of the derivatives and the integral.
    
    (3) This follows from the fact that
    \begin{align}
        g''(t) = \frac{1}{t^3}f''\left(\frac1t\right) \qquad \textrm{and} \quad
        \frac{1}{t^3}g''\left(\frac{1}{t}\right) = f''(t) \,.
    \end{align}
  
   (4) The fact that $D_f(\rho\|\sigma) \geq 0$ follows from the convexity of $f$, which ensures that $f''(x) \geq 0$ everywhere. To establish the equality condition we recall that the hockey-stick divergence satisfies
   \begin{align}
        E_\gamma(\rho\|\sigma) > 0 \implies \gamma \leq e^{D_{\max}(\rho\|\sigma)}, 
        \quad \textnormal{and} \quad 
        E_\gamma(\sigma\|\rho) > 0 \implies \gamma^{-1} \geq e^{-D_{\max}(\sigma\|\rho)} \,.
   \end{align}
   Therefore, as long as $f''(\gamma)$ vanishes on any such open interval the integral will be zero. The special case $\rho = \sigma$ can be seen more directly, since then $E_\gamma(\rho\|\sigma) = E_\gamma(\sigma\|\rho) = 0$ and, thus, $D_f(\rho\|\sigma) = 0$.
    
    (5) Recall that the hockey-stick divergence is bounded by $E_\gamma(\rho\|\sigma) \leq 1$ with equality if and only if $\rho \perp \sigma$. Moreover, it is zero when $\gamma \geq D_{\max}(\rho\|\sigma)$. Now let $\mu = e^{D_{\max}(\rho\|\sigma)}$ and $\lambda = e^{D_{\max}(\sigma\|\rho)}$.
    It follows that 
    \begin{align}
        D_f(\rho\|\sigma) &\leq \int_1^\mu f''(\gamma) \d\gamma + \int_1^\lambda \gamma^{-3} f''(\gamma^{-1})\, \d\gamma 
        = \int_1^\mu f''(\gamma) \,\d\gamma + \int_{\frac{1}{\lambda}}^1 \mu f''(\mu) \,\d\mu \\
        &=  f'(\gamma) \Big|_{1}^{\mu} + \mu f'(\mu) - f(\mu) \Big|_{\frac{1}{\lambda}}^1 
        = f(x) - x f'(x) \big|_{x = \frac{1}{\lambda}} + f' (x) \big|_{x = \mu}, 
    \end{align}
    where we performed a change of variable and partial integration and used that $f(1)=0$ for our functions.
    
    (6) and (7) These properties follow directly from DPI and joint convexity of the hockey-stick divergence, respectively, as well as the fact that $D_f(\rho\|\sigma)$ is a linear combination with positive coefficients of hockey-stick divergences $E_{\gamma}(\rho\|\sigma)$ and $E_{\gamma}(\sigma\|\rho)$.
    
    (8) This follows from the fact that the new definition coincides with the classical $f$-divergence for commuting states and the defining property of the minimal and maximal quantum extension~\cite{matsumoto2018new,gour2020optimal}.
\end{proof}

This gives us the tools to explore further properties of our $f$-divergences in the following sections.

\subsection{Quantum DeGroot statistical information} 

A classical integral representation in~\cite{liese2006divergences}, which serves as the starting point of the proof of~\cite[Proposition 3]{sason2016f}, was originally given in terms of the deGroot statistical information~\cite{degroot1962uncertainty}. We construct a quantum equivalent of this function based on symmetric hypothesis testing and show that it can also be used to give an integral representation of our quantum $f$-divergence. 

As previously described, in quantum state discrimination we attempt to minimize the error in determining a given quantum state. 
Here we are interested in the symmetric setting, where the goal is to minimize the error probability
\begin{align}
    p_e^*(p,\rho,\sigma) &= \inf_M p_e(p,\rho,\sigma, M) \\
    &=\inf_M p \tr M_1\rho + (1-p)\tr M_0\sigma \\
    & =\inf_M p \tr (1-M_0)\rho + (1-p)\tr M_0\sigma \\
    & =\inf_M p - p\tr M_0 \rho + (1-p)\tr M_0\sigma \sigma \\
    & =\inf_M p - p\tr \left[M_0 \left(\rho -\frac{1-p}{p} \sigma\right)\right] \\
    & = p - p E_\frac{1-p}{p}(\rho\|\sigma). 
\end{align}
Similarly, we could have rewritten the above as 
\begin{align}
    p_e^*(p,\rho,\sigma) = (1-p) - (1-p) E_\frac{p}{1-p}(\sigma\|\rho). 
\end{align}
This gives two equivalent expressions for the minimum error probability. We remark that alternatively this equivalence follows also from the symmetry of the hockey-stick divergence shown in~\cite[Proposition II.5]{hirche2022quantum} that is $E_\gamma(\rho\|\sigma) = \gamma E_{\frac1\gamma}(\sigma\|\rho) + (1-\gamma)$. We can now write the symmetric error probability as
\begin{align}
    B_p(\rho\|\sigma) := p^*_e(p,\rho,\sigma) = \begin{cases} 
    p - p E_\frac{1-p}{p}(\rho\|\sigma) & 0\leq p \leq \frac12 \\
    (1-p) - (1-p) E_\frac{p}{1-p}(\sigma\|\rho) & \frac12 \leq p \leq 1 \end{cases}, 
\end{align}
expressing it solely in terms of hockey-stick divergences $E_\gamma$ with $\gamma\geq 1$.  
Further, consider the case when $\rho=\sigma$, meaning the setting where the states do not give any additional information and the best guess is to assume the hypothesis with the higher prior, 
\begin{align}
    B_p := B_p(\rho\|\rho) =  \begin{cases} 
    p  & 0\leq p \leq \frac12 \\
    (1-p)  & \frac12 \leq p \leq 1 \end{cases} = \min\{p,1-p\}. 
\end{align}
Analogous to the classical setting we now define the quantum deGroot statistical information as the difference of minimum error between selecting the most likely a-priori hypothesis and when the most likely a posteriori hypothesis is selected, 
\begin{align}
    \cI_p(\rho\|\sigma) := B_p - B_p(\rho\|\sigma) = \begin{cases} 
     p E_\frac{1-p}{p}(\rho\|\sigma) & 0\leq p \leq \frac12 \\
    (1-p) E_\frac{p}{1-p}(\sigma\|\rho) & \frac12 \leq p \leq 1 \end{cases}.
\end{align}
Note as a special case that
\begin{align}
    \cI_\frac12(\rho\|\sigma) = \frac14 \|\rho-\sigma\|_1. 
\end{align}
Given this new operationally motivated quantity we can state the following integral representation of the relative entropy. 
\begin{corollary}
Let $f \in \cF$ and $\rho,\sigma$ be quantum states. Then 
\begin{align}
    D_f(\rho\|\sigma) = \int_0^1 \frac{1}{p^3} f''\left(\frac{1-p}{p}\right) \cI_p(\rho\|\sigma) \, \d p. \label{Eq:deGroot-rep}
\end{align}
\end{corollary}
\begin{proof}
    Starting from the integral representation above one can derive that given in~\ref{Cor:D-integral} using the same steps as done in the classical case in the proof of~\cite[Proposition 3]{sason2016f}. We only briefly give the main steps here. By definition of the quantum deGroot statistical information, we get from Equation~\eqref{Eq:deGroot-rep}, 
    \begin{align}
        D_f(\rho\|\sigma) = \int_0^{\frac12} \frac{1}{p^2} f''\left(\frac{1-p}{p}\right) E_\frac{1-p}{p}(\rho\|\sigma) \, \d p\, + \int_{\frac12}^1 \frac{1-p}{p^3} f''\left(\frac{1-p}{p}\right) E_\frac{p}{1-p}(\sigma\|\rho) \, \d p.
    \end{align}
    The claim then follows by substituting $\gamma=\frac{1-p}{p}$ and $\gamma=\frac{p}{1-p}$ in the first and second integral respectively, which gives the representation in Equation~\ref{Cor:D-integral}. 
\end{proof}
This gives a quantum version of the integral representation derived in~\cite{liese2006divergences}. For the case of the relative entropy, Equation~\eqref{Eq:deGroot-rep} specializes to,
\begin{align}
    D(\rho\|\sigma) = \int_0^1 \frac{1}{p^2(1-p)} \cI_p(\rho\|\sigma) \, \d p. 
\end{align}

\subsection{The $\chi^2$ divergence}

In this section we collect useful relationships between different $f$-divergences. First, we show that every $f$-divergence is connected to the $\chi^2$ divergence locally. 
\begin{theorem}\label{Thm:H2-is-limit-of-f}
    Let $f \in \cF$ and $\rho \ll\gg \sigma$. We have
    \begin{align}
        \frac{\partial^2}{\partial \lambda^2} D_f(\lambda\rho+(1-\lambda)\sigma\|\sigma) \Big|_{\lambda = 0} = f''(1) \,\chi^2(\rho\|\sigma). 
    \end{align}
    where $\chi^2(\rho\|\sigma) \equiv H_2(\rho\|\sigma)$ is the Hellinger divergence of order $2$.
\end{theorem}
As a direct consequence of this we can find alternative expressions for $\chi_2(\rho\|\sigma)$ and $D_2(\rho\|\sigma)$. For this, consider $f: t \mapsto t \log t$ such that $f''(1) = 1$ and, thus,
\begin{align}
    \chi^2(\rho\|\sigma) = \frac{\partial^2}{\partial \lambda^2} D(\lambda\rho+(1-\lambda)\sigma\|\sigma) \Big|_{\lambda = 0}
\end{align}
This second derivative was to the best of our knowledge first computed in~\cite{petz1998contraction}, and yields
\begin{align}
    \chi^2(\rho\|\sigma) &= \int_0^\infty \mathrm{d}s \ \tr \left( (\sigma + s \Id)^{-1} (\rho - \sigma) (\sigma + s \Id)^{-1} (\rho - \sigma) \right) \\
    &= \int_0^\infty \mathrm{d}s \ \tr \left( \rho (\sigma + s \Id)^{-1} \rho (\sigma + s \Id)^{-1} \right) - 1 \, \qquad\qquad\qquad \textrm{and} \\
    D_2(\rho\|\sigma) &= \log \int_0^\infty \mathrm{d}s \ \tr \left( \rho (\sigma + s \Id)^{-1} \rho (\sigma + s \Id)^{-1} \right) .
\end{align}

\begin{proof}[Proof of Theorem~\ref{Thm:H2-is-limit-of-f}]
   Since the $f$-divergence vanishes at $\lambda = 0$ and has a minimum there (implying that also the first derivative in $\lambda$ vanishes), the desired relation can be rewritten as
    \begin{align}
        \lim_{\lambda\rightarrow 0} \frac{2}{\lambda^2} D_f(\lambda\rho+(1-\lambda)\sigma\|\sigma) = f''(1) \,\chi^2(\rho\|\sigma). 
    \end{align}

    To prove this statement let us first rewrite the hockey-stick divergences for convex combinations of the above form. We start with the following observation, 
    \begin{align}
        E_\gamma(\lambda\rho+(1-\lambda)\sigma \| \sigma) &= \tr(\lambda\rho+(1-\lambda)\sigma - \gamma\sigma)_+ \\
        &= \lambda \tr\left(\rho - \frac{\gamma+\lambda-1}{\lambda}\sigma\right)_+ \\
        &= \lambda E_{\frac{\gamma+\lambda-1}{\lambda}}(\rho\|\sigma). \label{Eq:HS-identities-cc-1} 
    \end{align}
    For the second observation we assume $\gamma\leq\frac{1}{1-\lambda}$, 
    \begin{align}
        E_\gamma(\sigma\| \lambda\rho+(1-\lambda)\sigma) 
        &= \tr(\sigma-\gamma\lambda\rho-\gamma (1-\lambda)\sigma)_+ \\
        &= (1-\gamma(1-\lambda)) \tr\left(\sigma - \frac{\gamma\lambda}{1-\gamma(1-\lambda)}\rho\right)_+ \\
        &= (1-\gamma(1-\lambda)) E_{\frac{\gamma\lambda}{1-\gamma(1-\lambda)}}(\sigma\|\rho), \label{Eq:HS-identities-cc-2} 
    \end{align}
    where the second equality holds because $\gamma\leq\frac{1}{1-\lambda}$ ensures that $(1-\gamma(1-\lambda))\geq 0$. Note that if $\gamma\geq\frac{1}{1-\lambda}$ it can easily be checked that $E_\gamma(\sigma\| \lambda\rho+(1-\lambda)\sigma)=0$. 
    Using the above we can now simplify the integral representation of the $f$-divergences. We split the task into two steps. For the first, choose $\hat\gamma=\frac{\gamma+\lambda-1}{\lambda}$ and we get
    \begin{align}
        \int_1^\infty f''(\gamma) E_\gamma(\lambda\rho+(1-\lambda)\sigma \| \sigma) \,\d\gamma 
        &= \int_1^\infty f''(\gamma) \lambda E_{\hat\gamma}(\sigma\|\rho) \,\d\gamma \\
        &=\int_1^\infty \lambda^2 f''(\lambda(\hat\gamma-1)+1) E_{\hat\gamma}(\sigma\|\rho) \,\d\hat\gamma. 
    \end{align}
    For the second, choose $\hat\gamma=\frac{\gamma\lambda}{1-\gamma(1-\lambda)}$ and get
    \begin{align}
        \int_1^\infty \frac{1}{\gamma^3} f''\left(\frac{1}{\gamma}\right) E_\gamma(\sigma\| \lambda\rho+(1-\lambda)\sigma) \,\d\gamma 
        &= \int_1^{\frac{1}{1-\lambda}} \frac{1}{\gamma^3} f''\left(\frac{1}{\gamma}\right) (1-\gamma(1-\lambda)) E_{\hat\gamma}(\sigma\|\rho) \,\d\gamma \\ 
        &= \int_1^{\infty} \frac{\lambda^2}{\hat\gamma^3} f''\left(\frac{\lambda+\hat\gamma(1-\lambda)}{\hat\gamma}\right)  E_{\hat\gamma}(\sigma\|\rho) \,\d\hat\gamma, 
    \end{align}
    where in the first equality we limited the integral by using the vanishing hockey-stick divergence for $\gamma\geq\frac{1}{1-\lambda}$. Interestingly, this gets somewhat reversed by our chosen substitution. 
    
    Taking both integrals together, we see that
    \begin{align}
        \frac{2}{\lambda^2} D_f(\lambda\rho+(1-\lambda)\sigma\|\sigma) = 2 \int_1^\infty  f''(\lambda(\hat\gamma-1)+1) E_{\hat\gamma}(\sigma\|\rho)  +  \frac{1}{\hat\gamma^3} f''\left(\frac{\lambda+\hat\gamma(1-\lambda)}{\hat\gamma}\right)  E_{\hat\gamma}(\sigma\|\rho)  \,\d\hat\gamma. 
    \end{align}
    Taking the limit $\lambda\rightarrow 0$ we finally get
    \begin{align}
        \lim_{\lambda\rightarrow 0} \frac{2}{\lambda^2} D_f(\lambda\rho+(1-\lambda)\sigma\|\sigma) &= 2 \int_1^\infty f''(1) E_{\hat\gamma}(\sigma\|\rho)  +  \frac{1}{\hat\gamma^3} f''\left(1\right)  E_{\hat\gamma}(\sigma\|\rho) \,\d\hat\gamma \\
        &= f''(1) \Hell_2(\rho\|\sigma), 
    \end{align}
    which concludes the proof. 
    \end{proof}

\subsection{Generalized skew-divergences}

Skew-divergences~\cite{lee1999measures,audenaert2014quantum} are another way of generalizing a given family of divergences by introducing an additional parameter, say $\lambda$, that interpolates between the given states, $\lambda\rho+(1-\lambda)\sigma$. These divergences often have interesting properties, e.g. they are more easily bounded from above. Here, we will see that the simple structure of the Hockey-Stick divergence lends itself well to investigating skew-divergences and give some properties that will find applications later on. 

Given a function $f \in \cF$ and $0\leq\lambda,\mu\leq 1$, define the family of functions
\begin{align}
    F_{\lambda,\mu} : x \mapsto (1-\lambda+\lambda x) \left[ \mu f\left( \frac{1}{1-\lambda+\lambda x}\right) + (1-\mu) f\left(\frac{x}{1-\lambda+\lambda x}\right) \right]. \label{Eq:def-F-lambda-mu}
\end{align}
Clearly $F_{\lambda,\mu}(1) = 0$ and the second derivative evaluates to
\begin{align}
    F_{\lambda,\mu}''(x) = \frac{\mu \lambda^2 f''(\frac{1}{1-\lambda+\lambda x}) + (1-\mu)(1-\lambda)^2 f''(\frac{x}{1-\lambda+\lambda x})}{(1-\lambda+\lambda x)^3} ,
\end{align}
ensuring that the function is convex and, therefore, $F_{\lambda,\mu} \in \cF$.

\begin{proposition}\label{Prop:D_F_lambda-equality}
    Fix $f \in \cF$. Let $\rho,\sigma$ be quantum states and define $F_{\lambda,\mu}$ as in Equation~\eqref{Eq:def-F-lambda-mu} above. Then, we have 
    \begin{align}
        D_{F_{\lambda,\mu}}(\rho\|\sigma) = (1-\mu) D_f \big(\rho \big\|\lambda\rho+(1-\lambda)\sigma\big) + \mu D_f\big(\sigma \big\|\lambda\rho+(1-\lambda)\sigma \big) .
    \end{align}
\end{proposition}

\begin{proof}
Recall the following identities from Equation~\eqref{Eq:HS-identities-cc-1},
    \begin{align}
        E_\gamma(\lambda\rho+(1-\lambda)\sigma \| \sigma) &= \lambda E_{\frac{\gamma+\lambda-1}{\lambda}}(\rho\|\sigma), \label{Eq:HS-identities-cc-1r}
    \end{align}
    and, from Equation~\eqref{Eq:HS-identities-cc-2}, for $\gamma\leq\frac{1}{1-\lambda}$, 
    \begin{align}
        E_\gamma(\sigma\| \lambda\rho+(1-\lambda)\sigma) 
        &= (1-\gamma(1-\lambda)) E_{\frac{\gamma\lambda}{1-\gamma(1-\lambda)}}(\sigma\|\rho).  \label{Eq:HS-identities-cc-2r}
    \end{align}
    Furthermore one can similarly check that 
    \begin{align}
        E_\gamma(\lambda\rho+(1-\lambda)\sigma \| \rho) &= (1-\lambda) E_{\frac{\gamma-\lambda}{1-\lambda}}(\sigma\|\rho),\label{Eq:HS-identities-cc-3}
    \end{align}
    and for $\gamma\leq\frac{1}{\lambda}$, 
    \begin{align}
        E_\gamma(\rho\| \lambda\rho+(1-\lambda)\sigma) 
        &= (1-\gamma\lambda) E_{\frac{\gamma(1-\lambda)}{1-\gamma\lambda}}(\rho\|\sigma). \label{Eq:HS-identities-cc-4}       
    \end{align}
Moreover, we write $F_{\lambda}''(x) = \mu M_\lambda(x) + (1-\mu) N_\lambda(x)$ using
\begin{align} 
    M_\lambda(x) :=\frac{\lambda^2 f''(\frac{1}{1-\lambda+\lambda x})}{(1-\lambda+\lambda x)^3} \quad\text{and}\quad N_\lambda(x) :=\frac{(1-\lambda)^2 f''(\frac{x}{1-\lambda+\lambda x})}{(1-\lambda+\lambda x)^3} \,.
\end{align}

We can now rewrite the individual $f$-divergences on the right hand side of our statement,
\begin{align}
    D_f \big(\rho \big\|\lambda\rho+(1-\lambda)\sigma\big) = \int_1^\infty f''(\gamma) E_\gamma(\rho\|\lambda\rho+(1-\lambda)\sigma) +  f''\left(\frac1{\gamma}\right)\frac{1}{\gamma^3} E_\gamma(\lambda\rho+(1-\lambda)\sigma\|\rho) \,\d\gamma \,.
\end{align} 
and evaluate the two summands in the integral separately. The first integral simplifies to
\begin{align}
    \int_1^\infty f''(\gamma) E_\gamma(\rho\|\lambda\rho+(1-\lambda)\sigma) \,\d\gamma
    &= \int_1^{\frac{1}{\lambda}} f''(\gamma) (1-\gamma\lambda) E_{\frac{\gamma(1-\lambda)}{1-\gamma\lambda}}(\rho\|\sigma) \,\d \gamma \\
    &= \int_1^\infty f''\left(\frac{\bar\gamma}{1-\lambda(1-\bar\gamma)}\right) \frac{(1-\lambda)^2}{(1-\lambda(1-\bar\gamma))^3} E_{\bar\gamma}(\rho\|\sigma) \,\d\bar\gamma \\
    &= \int_1^\infty N_\lambda(\bar\gamma) E_{\bar\gamma}(\rho\|\sigma) \,\d\bar\gamma, 
\end{align}
    where the first equality is by Equation~\eqref{Eq:HS-identities-cc-4}, the second by substituting $\bar\gamma=\frac{\gamma(1-\lambda)}{1-\gamma\lambda}$ and the third by definition. Proceeding similarly, we get
\begin{align}
     \int_1^\infty f''\left(\frac1{\gamma}\right)\frac{1}{\gamma^3} E_\gamma(\lambda\rho+(1-\lambda)\sigma\|\rho) \,\d\gamma 
    &=  \int_1^\infty f''\left(\frac1{\gamma}\right) \frac{(1-\lambda)}{\gamma^3} E_{\frac{\gamma-\lambda}{1-\lambda}}(\sigma\|\rho) \,\d\gamma \\
    &= \int_1^\infty f''\left(\frac{\bar\gamma^{-1}}{(1-\lambda)+\lambda\bar\gamma^{-1}}\right)\frac{(1-\lambda)^2}{((1-\lambda)\bar\gamma +\lambda)^3} E_{\bar\gamma}(\sigma\|\rho) \,\d\bar\gamma \\
    &= \int_1^\infty \frac{1}{\bar\gamma^3} N_\lambda\left(\frac{1}{\bar\gamma}\right) E_{\bar\gamma}(\sigma\|\rho) \,\d\bar\gamma, 
\end{align}
where the first inequality is from Equation~\eqref{Eq:HS-identities-cc-3} and the second by substituting $\bar\gamma=\frac{\gamma-\lambda}{1-\lambda}$. 
Next, we have
\begin{align}
     \int_1^\infty f''(\gamma) E_\gamma(\sigma\| \lambda\rho+(1-\lambda)\sigma) \,\d\gamma 
    &= \int_1^{\frac{1}{1-\lambda}} f''(\gamma) (1-\gamma(1-\lambda)) E_{\frac{\gamma\lambda}{1-\gamma(1-\lambda)}}(\sigma\|\rho) \,\d\gamma \\
    &= \int_1^\infty f''\left(\frac1{1-\lambda+\lambda\bar\gamma^{-1}}\right) \frac{\lambda^2}{((1-\lambda)\bar\gamma +\lambda)^3} E_{\bar\gamma}(\sigma\|\rho) \,\d\bar\gamma \\
    &=  \int_1^\infty \frac{1}{\bar\gamma^3} M_\lambda\left(\frac{1}{\bar\gamma}\right) E_{\bar\gamma}(\sigma\|\rho) \,\d\bar\gamma, 
\end{align}
where the first equality follows from Equation~\eqref{Eq:HS-identities-cc-2r} and the second by substituting $\bar\gamma=\frac{\gamma\lambda}{1-\gamma(1-\lambda)}$. 
Finally, we have, 
\begin{align}
     \int_1^\infty f''\left(\frac1{\gamma}\right)\frac{1}{\gamma^3} E_\gamma(\lambda\rho+(1-\lambda)\sigma \| \sigma) \,\d\gamma 
    &=  \int_1^\infty f''\left(\frac1{\gamma}\right) \frac{\lambda}{\gamma^3}  E_{\frac{\gamma+\lambda-1}{\lambda}}(\rho\|\sigma) \,\d\gamma \\
    &=  \int_1^\infty f''\left(\frac1{1-\lambda+\lambda\bar\gamma}\right) \frac{\lambda^2}{(1-\lambda(1-\bar\gamma))^3} E_{\bar\gamma}(\rho\|\sigma) \,\d\bar\gamma \\
    &= \int_1^\infty M_\lambda(\bar\gamma) E_{\bar\gamma}(\rho\|\sigma) \,\d\bar\gamma, 
\end{align}
where the first equality is from Equation~\eqref{Eq:HS-identities-cc-1r} and the second is by substituting $\bar\gamma=\frac{\gamma+\lambda-1}{\lambda}$. Putting all the above integrals together, we find
\begin{align}
    &(1-\mu) D_f(\rho\|\lambda\rho+(1-\lambda)\sigma) + \mu D_f(\sigma\|\lambda\rho+(1-\lambda)\sigma) \nonumber\\
    &\qquad= \int_1^\infty F_{\lambda,\mu}''(\bar\gamma) E_{\bar\gamma}(\rho\|\sigma)  + \frac{1}{\bar\gamma^3} F_{\lambda,\mu}''\left(\frac{1}{\bar\gamma}\right) E_{\bar\gamma}(\sigma\|\rho)  \,\d\bar\gamma \\
    &\qquad= D_{F_{\lambda,\mu}}(\rho\|\sigma), 
\end{align}
concluding the proof.
\end{proof}

This very general result produces a number of useful special cases. For example, consider the function $F_{\lambda}(x) := F_{\lambda,\lambda}(x)$, which
interpolates between the cases
\begin{align}
    F_0(x) = f(x)  \quad &\implies\quad  D_{F_0}(\rho\|\sigma) = D_{f}(\rho\|\sigma)    \qquad \textnormal{and} \\
    F_1(x) = xf(x^{-1})  \quad&\implies\quad  D_{F_1}(\rho\|\sigma)= D_{f}(\sigma\|\rho).    
\end{align}
A related function is $\hat F_{\lambda}(x) := F_{\lambda,1-\lambda}(x)$.
Note that $\hat F_\lambda$ is such that $\hat F_0=f(1)=0$ and $\hat F_1=xf(1)=0$ and, therefore, $D_{\hat F_0}(\rho\|\sigma)=D_{\hat F_1}(\rho\|\sigma)=0$. 
Finally, note that we have
\begin{align}
    D_{F_{\lambda,0}}(\rho\|\sigma) =  D_f(\rho\|\lambda\rho+(1-\lambda)\sigma)  \qquad \textnormal{and} \qquad
    D_{F_{\lambda,1}}(\rho\|\sigma) =  D_f(\sigma\|\lambda\rho+(1-\lambda)\sigma),
\end{align}
which generalizes the concept of skew-divergences~\cite{lee1999measures,audenaert2014quantum} to our $f$-divergence. 

We have seen earlier that our integral representation reduces to the usual quantum relative entropy for $f(x)=x\log x$. As a special case of the previous proposition we can generalize that result and give an integral representation that interpolates between $D(\rho\|\sigma)$ and $D(\sigma\|\rho)$ and also includes the Jensen-Shannon divergence as half-way point, i.e.\ the quantities
\begin{align}
    (1-\mu) D(\rho\|\lambda\rho+(1-\lambda)\sigma) + \mu D(\sigma\|\lambda\rho+(1-\lambda)\sigma)
\end{align}
are themselves $f$-divergences, and include the Jensen-Shannon divergence
\begin{align}
JS(\rho\|\sigma) := \frac12 D\left(\rho\middle\|\frac{\rho+\sigma}{2}\right) + \frac12 D\left(\sigma\middle\|\frac{\rho+\sigma}{2}\right) .
\end{align}

Finally, we discuss one more special case that will be useful in the following sections. Define
\begin{align}
    LC_\lambda(\rho\|\sigma) := D_{g_\lambda}(\rho\|\sigma),  \label{Eq:LeCam-def}
\end{align}
where $\lambda \in [0, 1]$ and $g_\lambda(x)=\lambda(1-\lambda)\frac{(x-1)^2}{\lambda x+(1-\lambda)}$, in analogy to the classical Le Cam divergence~\cite{le2012asymptotic,gyorfi2001class}. 
\begin{corollary}\label{Cor:LC-H2-eq}
    Let $\lambda \in [0,1]$ and $\rho,\sigma$ be quantum states. We have, 
    \begin{align}
        LC_\lambda(\rho\|\sigma) = \lambda \chi^2\big(\rho \big\|\lambda\rho+(1-\lambda)\sigma\big) + (1-\lambda) \chi^2\big(\sigma \big\|\lambda\rho+(1-\lambda)\sigma\big) \,. 
    \end{align}
\end{corollary}

This provides us with a wider class of functions for which the new integral $f$-divergence equals an expression given by known divergences.


\section{Relationship with quantum \Renyi divergences}
\label{sec:renyi}

There are plenty of quantum generalizations of \Renyi divergences and it is a natural question whether we can find an integral representation analogous to the classical case. Some of the most common quantum \Renyi divergences are those by Petz~\cite{petz1985quasi,petz1986quasi}, as well as the sandwiched~\cite{muller2013quantum,wilde2014strong} and geometric~\cite{matsumoto2018new} quantum \Renyi divergence. Those are respectively defined as
\begin{align}
    \widebar D_\alpha(\rho\|\sigma) &:= \frac{1}{\alpha -1} \log \tr \left( \rho^{\alpha}\sigma^{1-\alpha} \right) , \qquad &&\textnormal{for}\ \alpha \in (0, 1) \cup (1, 2] , \\
    \widetilde D_\alpha(\rho\|\sigma) &= \frac{1}{\alpha -1} \log \tr \left( \left(\sigma^{\frac{1-\alpha}{2\alpha}}\rho\sigma^{\frac{1-\alpha}{2\alpha}}\right)^{\alpha} \right), \qquad &&\textnormal{for}\ \alpha \in \Big[\frac12, 1\Big) \cup (1, \infty) , \\
    \widehat D_\alpha(\rho\|\sigma) &= \frac{1}{\alpha -1} \log \tr \left( \sigma^{\frac12}(\sigma^{-\frac12}\rho\sigma^{-\frac12})^\alpha \sigma^{\frac12} \right) , \qquad &&\textnormal{for}\ \alpha \in (0, 1) \cup (1, 2] ,
\end{align}
whenever either $\rho \ll \sigma$ or $\alpha < 1$ and $\rho \not\perp \sigma$ and are set to $+\infty$ otherwise. We can also define $\widecheck{D}_{\alpha}$ and $\widehat{D}_{\alpha}$ as the minimal and maximal quantum extensions of $D_{\alpha}$.

Our previous success in generalizing the classical integral representation approach to the quantum relative entropy motivates us to make the following definitions, analogous to the classical case in Equation~\eqref{eq:renyitransform} via the Hellinger divergence. For $\alpha \in (0, 1) \cup (1, \infty)$ and $\rho \ll \sigma$ or $0 < \alpha< 1\wedge \rho\not\perp\sigma$, we define\footnote{We take note of an abuse of notation concerning $D_{\alpha}$ and $D_f$. In the following, Greek letters in the subscript denote \Renyi divergences whereas lower-case Latin letters denote $f$-divergences.}

\begin{align}
    \Hell_\alpha(\rho\|\sigma) =&\ \alpha\int_1^\infty \gamma^{\alpha-2} E_\gamma(\rho\|\sigma) + \gamma^{-\alpha-1} E_\gamma(\sigma\|\rho) \,\d\gamma, \\
    D_\alpha(\rho\|\sigma) :=&\ \frac{1}{\alpha-1}\log \left( 1+(\alpha-1)\Hell_{\alpha}(\rho\|\sigma) \right) .
\end{align}
Otherwise, we set $D_\alpha(\rho\|\sigma) = +\infty$.

One can quickly check numerically that $D_\alpha(\rho\|\sigma)$ does not coincide with any of the previously defined standard quantum extensions of the \Renyi divergence (see Figure~\ref{Fig:ExamplesDa}). Nonetheless, we can take note of some basic properties of this definition.
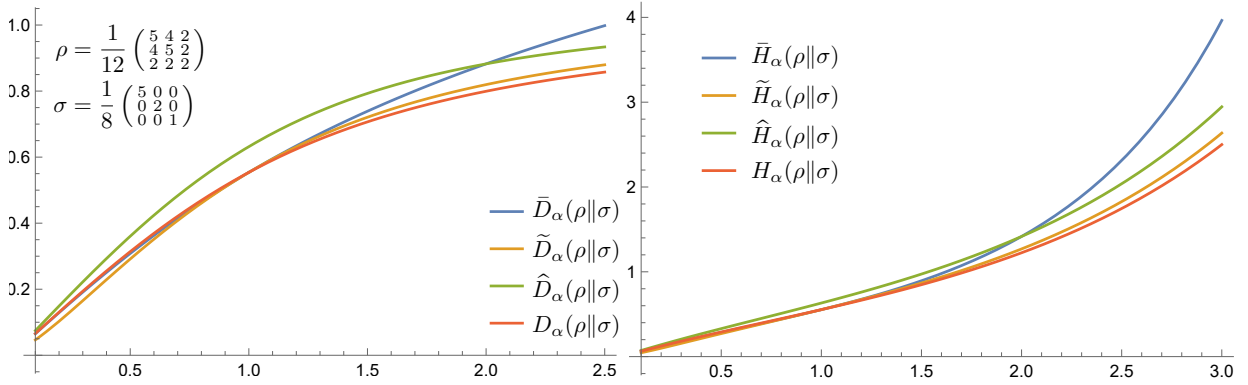
\begin{figure}[t]
\centering
\begin{minipage}{.49\textwidth}
\scalebox{0.88}{\begin{tikzpicture}
    \node[anchor=south west,inner sep=0] (image) at (0,0) {\includegraphics[width=1.136\textwidth]{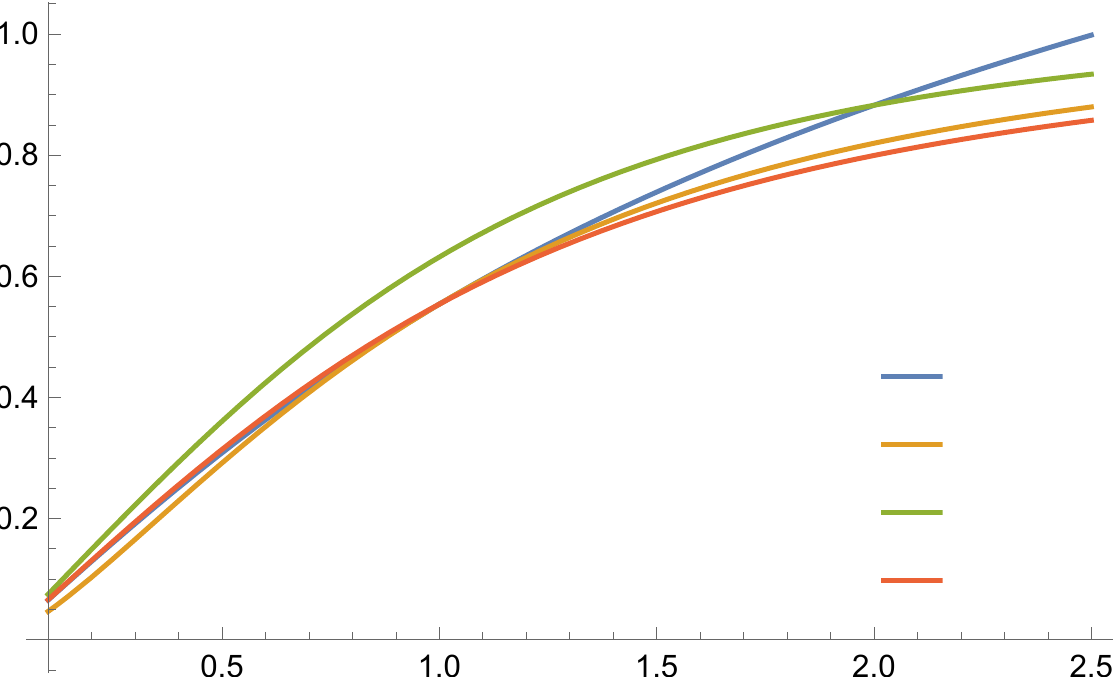}};
    \begin{scope}[x={(image.south east)},y={(image.north west)}]
        \node[] at (0.932,0.29){$\begin{aligned}
        \widebar D_\alpha(\rho\|\sigma) \\
        \widetilde D_\alpha(\rho\|\sigma) \\
        \widehat D_\alpha(\rho\|\sigma) \\
        D_\alpha(\rho\|\sigma) 
        \end{aligned}$};
        \node[] at (0.2,0.81){$\begin{aligned}
        \rho&=\frac1{12}\left( \begin{smallmatrix} 5 & 4 & 2 \\ 4 & 5 & 2 \\ 2 & 2 & 2 \end{smallmatrix}\right) \\
        \sigma&=\frac1{8}\left( \begin{smallmatrix} 5 & 0 & 0 \\ 0 & 2 & 0 \\ 0 & 0 & 1 \end{smallmatrix}\right) 
        \end{aligned}$};
    \end{scope}
\end{tikzpicture}}
\end{minipage}
\begin{minipage}{.49\textwidth}
\scalebox{0.88}{\begin{tikzpicture}
    \node[anchor=south west,inner sep=0] (image) at (0,0) {\includegraphics[width=1.136\textwidth]{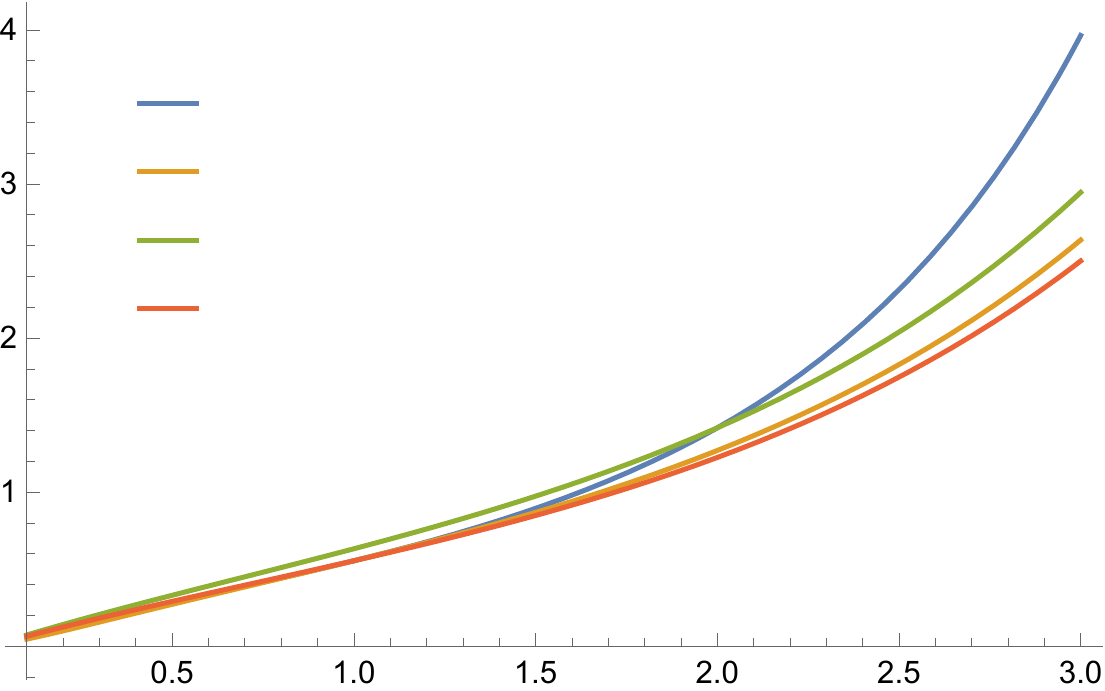}};
    \begin{scope}[x={(image.south east)},y={(image.north west)}]
        \node[] at (0.278,0.7){$\begin{aligned}
        \widebar \Hell_\alpha(\rho\|\sigma) \\
        \widetilde \Hell_\alpha(\rho\|\sigma) \\
        \widehat \Hell_\alpha(\rho\|\sigma) \\
        \Hell_\alpha(\rho\|\sigma) 
        \end{aligned}$};
    \end{scope}
\end{tikzpicture}}
\end{minipage}
\caption{\label{Fig:ExamplesDa} Plot of the different \Renyi divergences (left) and Hellinger divergences (right) defined in the main text over $\alpha$ for the states specified above. As expected $\widebar D_\alpha$, $\widetilde D_\alpha$ and $D_\alpha$ appear equal for $\alpha=1$, however for other values of $\alpha$ the quantity $D_\alpha$ is distinctly different from the others.}
\end{figure}

\begin{proposition}\label{Prop:Da-properties}
    Let $\alpha\in(0,1)\cup(1,\infty)$ and $\rho, \sigma$ two quantum states. The following properties hold:
    \begin{enumerate}

        \item (Support condition) If $\rho \ll \sigma$ or $0 < \alpha< 1\wedge \rho\not\perp\sigma$ then $D_{\alpha}(\rho\|\sigma)$ is finite.
        
        \item (Data-processing inequality) $D_\alpha(\rho\|\sigma)$ obeys the DPI, i.e., for any quantum channel $\cN$, we have
        \begin{align}
            D_\alpha\big(\cN(\rho) \big\| \cN(\sigma)\big) \leq D_\alpha(\rho\|\sigma).
        \end{align}
                  
        \item (Faithfullness) We have
        \begin{align}
            D_\alpha(\rho\|\sigma) \geq 0
        \end{align}
        with equality if and only if $\rho = \sigma$.

        \item (Relation to minimal and maximal quantum \Renyi divergence) We have
        \begin{align}
            \widecheck{D}_\alpha(\rho\|\sigma) \leq D_\alpha(\rho\|\sigma) \leq \widehat{D}_\alpha(\rho\|\sigma)
        \end{align}

    \end{enumerate}
\end{proposition}
\begin{proof}
    We establish the items separately.
    
    (1) This can be checked from the integral form of $\Hell_\alpha$. First note that the second integral is always bounded by $1$, even for orthogonal states. The first integral can easily be checked to be finite whenever $D_{\max}(\rho\|\sigma)$ is bounded, i.e. when $\rho \ll \sigma$. Moreover, if $0<\alpha<1$ then the first integral is never larger than $\frac{\alpha}{1-\alpha}$ with equality only if $\rho\not\perp\sigma$. This covers all the cases of the statement. 
    
     (2) follow directly from the DPI of the underlying $f$-divergence.

     (3) Since $\gamma > 0$ we have $H_{\alpha}(\rho\|\sigma) \geq 0$ with equality if and only if $\rho = \sigma$, which implies the desired property.

    (4) This follows by definition of the minimal and maximal quantum extension.

\end{proof}

One of the properties that we sorely miss is that the quanity is not additive under tensor products, that is, in general we have
\begin{align}
    D_{\alpha}(\rho_1 \otimes \rho_2 \| \sigma_1 \otimes \sigma_2) \neq D_{\alpha}(\rho_1 \| \sigma_1) + D_{\alpha}(\rho_2 \| \sigma_2) \,.
\end{align}
This can be seen as follows. It is well known~\cite{muller2013quantum,tomamichel2015quantum}, that the sandwiched \Renyi divergence is the smallest \Renyi divergence that satisfies data processing and additivity. For our divergence, we have shown data processing above, nevertheless, as evident in Figure~\ref{Fig:ExamplesDa}, the quantity can be smaller than the sandwiched \Renyi divergence and hence can not be additive. 

However, as we will see, this changes when we include a regularization, and define
\begin{align}
    D_\alpha^{\mathrm{reg}}(\rho\|\sigma) := \lim_{n\rightarrow\infty} \frac1n D_\alpha(\rho^{\otimes n}\|\sigma^{\otimes n}). 
\end{align}
In the following we will often abbreviate $\rho_n=\rho^{\otimes n}$ and $\sigma_n=\sigma^{\otimes n}$.

The following theorem is the main result of this section. 
\begin{theorem}\label{thm:Renyi-main}
Let $\rho, \sigma$ be quantum states and $\alpha\in(0,1)\cup(1,\infty)$. Then,
\begin{align}
    D_\alpha^{\mathrm{reg}}(\rho\|\sigma) = \begin{cases}
        \widebar D_\alpha(\rho\|\sigma) &0<\alpha<1 \\
        \widetilde D_\alpha(\rho\|\sigma) &1<\alpha<\infty
        \end{cases}. \label{Eq:DI-closed-formula}
\end{align}
\end{theorem}

The proof involves four cases with very different techniques and we will present it in Subsection~\ref{sec-proof}. Before that we note a couple of properties of $D_\alpha^{\mathrm{reg}}$ that follow easily from the above representation. 
First, clearly $D_\alpha^{\mathrm{reg}}$ obeys data processing on its full range $\alpha\in(0,1)\cup(1,\infty)$. 
Secondly, it can naturally be extended to $\alpha\in\{0,1,\infty\}$ by taking the appropriate limit, 
\begin{align}
    \lim_{\alpha\rightarrow 0} D_\alpha^{\mathrm{reg}}(\rho\|\sigma) &= \widebar D_0(\rho\|\sigma) = -\log\Pi_\rho\sigma, \\
    \lim_{\alpha\rightarrow 1} D_\alpha^{\mathrm{reg}}(\rho\|\sigma) &= D(\rho\|\sigma), \\
    \lim_{\alpha\rightarrow \infty} D_\alpha^{\mathrm{reg}}(\rho\|\sigma) &= D_{\max}(\rho\|\sigma) = \inf\{ \lambda : \rho \leq \exp(\lambda)\sigma \}, 
\end{align}
where $\Pi_\rho$ is the projector onto the support of $\rho$. 

Finally, we remark that $\widebar D_\alpha$ and $\widetilde D_\alpha$ have operational interpretations for exactly those ranges of $\alpha$ in which they appear in the right hand side of Equation~\eqref{Eq:DI-closed-formula} in terms of binary hypothesis testing, more precisely in the error exponent and the strong converse exponent respectively. With that motivation  it has previously been suggested that the right hand of~\eqref{Eq:DI-closed-formula} could be the standard definition of the quantum \Renyi divergence (see~\cite{mosonyi2015quantum}).

\subsection{Proof of Theorem~\ref{thm:Renyi-main}}
\label{sec-proof}

In the following we present the proof of Theorem~\ref{thm:Renyi-main} split into four cases. 
Throughout we will only consider the case where $\rho \ll \sigma$ or $0<\alpha< 1\wedge \rho\not\perp\sigma$, as otherwise both sides are infinite and the statement holds trivially. 

\subsubsection{Proof of upper bound for $\alpha>1$}

First, we will need some properties of the hockey-stick divergence. Recall that 
\begin{align}
    E_\gamma(\rho_n\|\sigma_n) &= \sup_M \tr M(\rho_n-\gamma\sigma_n) \\
    &= \tr P_{\gamma,n} (\rho_n - \gamma\sigma_n),  
\end{align}
where $P_{\gamma,n} = \{\rho_n-\gamma\sigma_n\}_+$ is the projector onto the positive part of $\rho_n-\gamma\sigma_n$. This is strikingly similar to the Neyman-Pearson test appearing e.g. in the proof of the strong converse exponent for binary quantum state discrimination, where $\gamma=e^{an}$ is specified, see~\cite{mosonyi2015quantum}. Now, from the positivity of the hockey-stick divergence we easily see that
\begin{align}
    \tr P_{\gamma,n} \rho_n \geq \gamma\tr P_{\gamma,n} \sigma_n. \label{Eq:proj-bound}
\end{align}
With this in place, we can get the following lemma whose first part is an adaptation of the proof of~\cite[Lemma IV.3]{mosonyi2015quantum}. 
\begin{lemma}
For $\alpha\geq 1$ and any $n\in\NN$, we have
\begin{align}
    \tr P_{\gamma,n} \rho_n \leq \gamma^{1-\alpha} \widetilde Q_\alpha(\rho_n\|\sigma_n). 
\end{align}
and as a consequence for $\gamma\geq 1$, 
\begin{align}
    E_\gamma(\rho_n\|\sigma_n) &\leq \gamma^{1-\alpha} \widetilde Q_\alpha(\rho_n\|\sigma_n), \\
    E_\gamma(\sigma_n\|\rho_n) &\leq 1 -\gamma + \gamma^{\alpha} \widetilde Q_\alpha(\rho_n\|\sigma_n).
\end{align}
\end{lemma}
\begin{proof}
    Observe the following, 
    \begin{align}
        \tr P_{\gamma,n} \rho_n &=  \left(\tr P_{\gamma,n}\rho_n\right)^\alpha \left(\tr P_{\gamma,n}\rho_n\right)^{1-\alpha}  \\
        &\leq \gamma^{1-\alpha} \left(\tr P_{\gamma,n}\rho_n\right)^\alpha \left(\tr P_{\gamma,n}\sigma_n\right)^{1-\alpha} \\
        &\leq \gamma^{1-\alpha} \left[ \left(\tr P_{\gamma,n}\rho_n\right)^\alpha \left(\tr P_{\gamma,n}\sigma_n\right)^{1-\alpha} +  \left(\tr (\Id-P_{\gamma,n})\rho_n\right)^\alpha \left(\tr (\Id-P_{\gamma,n})\sigma_n\right)^{1-\alpha} \right] \\
        &\leq \gamma^{1-\alpha} \widetilde Q_\alpha(\rho_n\|\sigma_n), 
    \end{align}
    where the first inequality follows from Equation~\ref{Eq:proj-bound}, the second by adding a positive term and the third by data processing. With this we get
    \begin{align}
        E_\gamma(\rho_n\|\sigma_n) &= \tr P_{\gamma,n} \rho_n - \gamma\tr P_{\gamma,n} \sigma_n \\
        &\leq \gamma^{1-\alpha} \widetilde Q_\alpha(\rho_n\|\sigma_n),
    \end{align}
    and 
    \begin{align}
        E_\gamma(\sigma_n\|\rho_n) &= \gamma E_{\frac{1}{\gamma}}(\rho_n\|\sigma_n)  \\
        &= \gamma \tr P_{\frac{1}{\gamma},n} \rho_n - \tr P_{\frac{1}{\gamma},n} \sigma_n - \gamma(1-\frac1\gamma)_+ \\
        &\leq 1 -\gamma + \gamma^{\alpha} \widetilde Q_\alpha(\rho_n\|\sigma_n),
    \end{align}
    where the first equality is~\cite[Proposition II.5]{hirche2022quantum}.
\end{proof}

We will use these bounds now to prove the desired upper bound. 
For any $\alpha>1$ and any $0<\epsilon<\alpha-1$, fix $\alpha'=\alpha+\epsilon>1$ and $\alpha''=\alpha-\epsilon>1$. We get
\begin{align}
    \Hell_\alpha(\rho\|\sigma) &= \alpha\int_1^\infty \gamma^{\alpha-2} E_\gamma(\rho\|\sigma) + \gamma^{-\alpha-1} E_\gamma(\sigma\|\rho) \d\gamma \\
    &\leq \alpha\int_1^\infty \gamma^{\alpha-2} \gamma^{1-\alpha'} \widetilde Q_{\alpha'}(\rho\|\sigma) + \gamma^{-\alpha-1} \left(1 -\gamma + \gamma^{\alpha''} \widetilde Q_{\alpha''}(\rho\|\sigma)\right)  \d\gamma \\ 
    &= \alpha \left(\widetilde Q_{\alpha'}(\rho\|\sigma) + \widetilde Q_{\alpha''}(\rho\|\sigma) \right) \int_1^\infty \gamma^{-1-\epsilon}  \d\gamma + \alpha\int_1^\infty \left(\gamma^{-\alpha-1} -\gamma^{-\alpha}\right) \d\gamma \\
    &= \frac{\alpha}{\epsilon} \left(\widetilde Q_{\alpha'}(\rho\|\sigma) + \widetilde Q_{\alpha''}(\rho\|\sigma) \right) - \frac{1}{\alpha-1}.
\end{align}
This implies
\begin{align}
    D_\alpha(\rho\|\sigma) &= \frac{1}{\alpha-1}\log[ 1+(\alpha-1)\Hell_{\alpha}(\rho\|\sigma)] \\
    &\leq \frac{1}{\alpha-1}\log\left[\frac{(\alpha^2-\alpha)}{\epsilon} \left(\widetilde Q_{\alpha'}(\rho\|\sigma) + \widetilde Q_{\alpha''}(\rho\|\sigma) \right)\right]
\end{align}
and further 
\begin{align}
    D_\alpha^{\mathrm{reg}}(\rho\|\sigma) &= \lim_{n\rightarrow\infty} \frac1n D_\alpha(\rho_n\|\sigma_n)  \\
    &\leq \lim_{n\rightarrow\infty} \frac1n \frac{1}{\alpha-1}\log\left[\frac12\left( \widetilde Q_{\alpha'}(\rho_n\|\sigma_n) + \widetilde Q_{\alpha''}(\rho_n\|\sigma_n) \right)\right] + \lim_{n\rightarrow\infty} \frac1n \frac{1}{\alpha-1}\log\left[\frac{2(\alpha^2-\alpha)}{\epsilon} \right] \\
    &= \lim_{n\rightarrow\infty} \frac1n \frac{1}{\alpha-1}\log\left[ \frac12\left(\widetilde Q_{\alpha'}(\rho\|\sigma)^n + \widetilde Q_{\alpha''}(\rho\|\sigma)^n \right) \right] \\
    &\leq \frac{1}{\alpha-1} \log \max \left\{ \widetilde Q_{\alpha'}(\rho\|\sigma), \widetilde Q_{\alpha''}(\rho\|\sigma)  \right\} . 
\end{align}
Finally, as we can chose $\epsilon$ arbitrarily small, continuity of $\alpha \mapsto \widetilde{Q}_{\alpha}(\rho\|\sigma)$ implies that $D_\alpha^{\mathrm{reg}}(\rho\|\sigma) \leq \widetilde{D}_\alpha(\rho\|\sigma)$.

\subsubsection{Proof of upper bound for $0<\alpha<1$}

This part of the proof draws its ideas from symmetric hypothesis testing, in particular the proof of the converse of Chernoffs bound in~\cite{nussbaum2009chernoff}. 
Recall that the optimal probability of error considered in this scenario can be written as
\begin{align}
    p_e^*(p,\rho,\sigma) &= \inf_M p_e(p,\rho,\sigma, M) \\
    &=\inf_M p \tr M_1\rho + (1-p)\tr M_0\sigma \\
    & = p - p E_\frac{1-p}{p}(\rho\|\sigma), 
\end{align}
or equivalently, 
\begin{align}
    p_e^*(p,\rho,\sigma) = (1-p) - (1-p) E_\frac{p}{1-p}(\sigma\|\rho). 
\end{align}
For the next step, consider decompositions of the states, 
\begin{align}
    \rho = \sum_x r_x |v_x\rangle\langle v_x| \quad \text{and}\quad \sigma = \sum_y s_y |u_y\rangle\langle u_y|, 
\end{align} 
with $\{|v_x\rangle\}_{x=1\dots d}$ and $\{|u_y\rangle\}_{x=1\dots d}$ being two orthonormal bases. Based on these we define the classical Nussbaum-Szko{\l}a distributions as
\begin{align}
    P_{\rho,\sigma}(x,y) = r_x |\langle v_x|u_y\rangle|^2 \quad \text{and}\quad Q_{\rho,\sigma}(x,y) = s_x |\langle v_x|u_y\rangle|^2. 
\end{align}
The main observation in~\cite{nussbaum2009chernoff} is that using these particular distributions one has 
\begin{align} 
    p_e^*(p,\rho,\sigma) \geq \frac12 p_e^*(p,P_{\rho,\sigma},Q_{\rho,\sigma}), \label{Eq:NS-prob-error}
\end{align}
where the right hand side is the optimal symmetric error distinguishing between $P_{\rho,\sigma}$ and $Q_{\rho,\sigma}$, see also~\cite[Proposition 2]{audenaert2008asymptotic}. As a result we can get the following bounds on the hockey-stick divergence as a corollary.
\begin{corollary}
    For $\gamma\geq 1$ and $P_{\rho,\sigma}$ and $Q_{\rho,\sigma}$ as above, one has
    \begin{align}
        E_\gamma(\rho\|\sigma) &\leq \frac{1}2  + \frac{1}2 E_\gamma(P_{\rho,\sigma}\|Q_{\rho,\sigma}), 
    \end{align}
\end{corollary}
\begin{proof}
    We start by noting that Equation~\ref{Eq:NS-prob-error} is equivalent to
    \begin{align}
        p - p E_\frac{1-p}{p}(\rho\|\sigma) &\geq \frac{p}2  - \frac{p}2 E_\frac{1-p}{p}(P_{\rho,\sigma}\|Q_{\rho,\sigma}) \\
      \iff\quad   E_\frac{1-p}{p}(\rho\|\sigma) &\leq \frac{1}2  + \frac{1}2 E_\frac{1-p}{p}(P_{\rho,\sigma}\|Q_{\rho,\sigma}) 
    \end{align}
    which is equivalent to the desired statement for $\gamma=\frac{1-p}{p}$. 
\end{proof}
Using this, we find the following sequence of inequalities:
\begin{align}
    \Hell_\alpha(\rho\|\sigma) &= \alpha\int_1^\infty \gamma^{\alpha-2} E_\gamma(\rho\|\sigma) + \gamma^{-\alpha-1} E_\gamma(\sigma\|\rho) \, \d\gamma \\
    &\leq \alpha\int_1^\infty \gamma^{\alpha-2} \left(\frac{1}2  + \frac{1}2 E_\gamma(P_{\rho,\sigma}\|Q_{\rho,\sigma})\right) + \gamma^{-\alpha-1} \left(\frac{1}2  + \frac{1}2 E_\gamma(Q_{\rho,\sigma}\|P_{\rho,\sigma})\right)  \d\gamma\\
    &= \frac{\alpha}{2}\int_1^\infty (\gamma^{\alpha-2} + \gamma^{-\alpha-1} ) \d\gamma + \frac{\alpha}{2}\int_1^\infty \gamma^{\alpha-2}  E_\gamma(P_{\rho,\sigma}\|Q_{\rho,\sigma}) + \gamma^{-\alpha-1}  E_\gamma(Q_{\rho,\sigma}\|P_{\rho,\sigma}) \, \d\gamma \\
    &= \frac{1}{2(1-\alpha)} + \frac{1}{2} \Hell_\alpha(P_{\rho,\sigma}\|Q_{\rho,\sigma}), \\
    &= \frac{1}{2(1-\alpha)} + \frac{1}{2} \widebar\Hell_\alpha(\rho\|\sigma),
\end{align}
where the last equality follows from the definition of the probabilities~\cite{nussbaum2009chernoff}, see also~\cite[Proposition 1]{audenaert2008asymptotic}. 
This implies
\begin{align}
    D_\alpha(\rho\|\sigma) = \frac{1}{\alpha-1}\log \left( 1+(\alpha-1)\Hell_{\alpha}(\rho\|\sigma) \right) \leq \widebar D_\alpha(\rho\|\sigma) + \frac{1}{1-\alpha}\log 2 
\end{align}
and, thus, $D_\alpha^{\mathrm{reg}}(\rho\|\sigma) \leq \widebar D_\alpha(\rho\|\sigma)$.

\subsubsection{Proof of lower bound for $\alpha>1$}

In this part of the proof we make use of the following observation~\cite{mosonyi2015quantum}, see also~\cite{tomamichel2015quantum}: 
\begin{align}
    \widetilde D_\alpha(\rho\|\sigma) = \lim_{n\rightarrow\infty} \frac{1}{n} \widecheck{D}_{\alpha}(\rho_n\|\sigma_n), 
\end{align}
where, using the notion in Equation~\ref{eq:Df_measured}, we define
\begin{align}
    \widecheck{D}_{\alpha}(\rho\|\sigma) = \sup_M D_\alpha( P_{M,\rho} \| P_{M,\sigma}) \label{Eq:meas-RRE}
\end{align}
as the measured \Renyi divergence. 
For ease of presentation we will use $P \equiv P_{M,\rho}$ and $Q \equiv P_{M,\sigma}$ for a measurement $M$ that achieves the supremum in the above. By data processing, we have 
\begin{align}
    \Hell_\alpha(\rho\|\sigma) &= \alpha\int_1^\infty (\gamma^{\alpha-2} E_\gamma(\rho\|\sigma) + \gamma^{-\alpha-1} E_\gamma(\sigma\|\rho)) \,\d\gamma \\
    &\geq \alpha\int_1^\infty (\gamma^{\alpha-2} E_\gamma(P\|Q) + \gamma^{-\alpha-1} E_\gamma(Q\|P)) \,\d\gamma  =\Hell_\alpha(P\|Q). 
\end{align}
This implies that $D_\alpha(\rho\|\sigma) \geq \widecheck{D}_\alpha(\rho\|\sigma)$ and, regularizing both sides, we conclude that $D_\alpha^{\mathrm{reg}}(\rho\|\sigma) \geq \widetilde D_\alpha(\rho\|\sigma)$.

\subsubsection{Proof of lower bound for $0<\alpha<1$}

Again, we will need some new bounds on the hockey-stick divergence. The following statement was proven in~\cite{audenaert2007discriminating}, see also~\cite{audenaert2008asymptotic}. For all $0\leq s\leq 1$ and positive semidefinite operators $A$ and $B$ one has
\begin{align}
    \tr[A+B-|A-B|]/2 \leq \tr[A^sB^{1-s}]. \label{Eq:Audenaert}
\end{align}
We use this as the main ingredient to prove the following corollary. 
\begin{corollary}
    For any $\gamma\geq1$ and $0\leq s\leq 1$, we have
    \begin{align}
        E_\gamma(\rho\|\sigma) &\geq 1-\gamma^{1-s}\widebar Q_s(\rho\|\sigma).
    \end{align}
\end{corollary}
\begin{proof}
   This follows directly from Equation~\eqref{Eq:Audenaert} by choosing $A=\rho$ and $B=\gamma\sigma$ and recalling that
   \begin{align}
        E_\gamma(\rho\|\sigma) = \frac12 \|\rho-\gamma\sigma\|_1 +\frac12(1-\gamma). 
   \end{align}
\end{proof}
We use this for the following bound, 
\begin{align}
    \Hell_\alpha(\rho\|\sigma) &= \alpha\int_1^\infty (\gamma^{\alpha-2} E_\gamma(\rho\|\sigma) + \gamma^{-\alpha-1} E_\gamma(\sigma\|\rho)) \d\gamma \\
    &\geq \alpha\int_1^\infty \gamma^{\alpha-2} \left(1-\gamma^{1-s}\widebar Q_s(\rho\|\sigma)\right) + \gamma^{-\alpha-1} \left(1-\gamma^{1-s'}\widebar Q_{s'}(\sigma\|\rho)\right) \d\gamma \\
    &= -\alpha\int_1^\infty \gamma^{\alpha-1-s}\widebar Q_s(\rho\|\sigma) + \gamma^{-\alpha-s'}\widebar Q_{s'}(\sigma\|\rho) \d\gamma + \alpha\int_1^\infty \gamma^{\alpha-2} + \gamma^{-\alpha-1} \, \d\gamma \\
    &= -\alpha\int_1^\infty \gamma^{\alpha-1-s}\widebar Q_s(\rho\|\sigma) + \gamma^{-\alpha-s'}\widebar Q_{s'}(\sigma\|\rho) \, \d\gamma + \frac{1}{1-\alpha}.
\end{align}
Now, fix an $\epsilon$ such that $0<\epsilon<\min\{\alpha,1-\alpha\}$, define $\alpha':=\alpha+\epsilon$ and $\alpha'':=\alpha-\epsilon$ and set $s=\alpha'$ and $s'=1-\alpha''$. This leads to
\begin{align}
    \Hell_\alpha(\rho\|\sigma) &\geq -\alpha\left(\widebar Q_{\alpha'}(\rho\|\sigma) + \widebar Q_{\alpha''}(\rho\|\sigma)\right) \left(\int_1^\infty \gamma^{-1-\epsilon} \d\gamma\right) + \frac{1}{1-\alpha} \\
    &= -\frac{\alpha}{\epsilon}\left(\widebar Q_{\alpha'}(\rho\|\sigma) + \widebar Q_{\alpha''}(\rho\|\sigma)\right) + \frac{1}{1-\alpha}.
\end{align}
This implies
\begin{align}
    D_\alpha(\rho\|\sigma) &= \frac{1}{\alpha-1}\log \left( 1+(\alpha-1)\Hell_{\alpha}(\rho\|\sigma) \right) \\
    &\geq \frac{1}{\alpha-1}\log \left( \frac{\alpha(1- \alpha)}{\epsilon}\left(\widebar Q_{\alpha'}(\rho\|\sigma) + \widebar Q_{\alpha''}(\rho\|\sigma)\right) \right)
\end{align}
and, therefore, 
\begin{align}
    D_\alpha^{\mathrm{reg}}(\rho\|\sigma) &\geq \lim_{n\rightarrow\infty} \frac1n \frac{1}{\alpha-1}\log\left( \frac{1}{2}\left(\widebar Q_{\alpha'}(\rho_n\|\sigma_n) + \widebar Q_{\alpha''}(\rho_n\|\sigma_n)\right)\right) + \lim_{n\rightarrow\infty} \frac1n \frac{1}{\alpha-1}\log\left(-\frac{2\alpha(\alpha-1)}{\epsilon}\right) \\
    &= \lim_{n\rightarrow\infty} \frac1n \frac{1}{\alpha-1}\log\left(\frac{1}{2}\left(\widebar Q_{\alpha'}(\rho\|\sigma)^n + \widebar Q_{\alpha''}(\rho\|\sigma)^n\right)\right) \\
    &\geq \frac{1}{\alpha-1} \log \min \left\{ \widebar Q_{\alpha'}(\rho\|\sigma) , \widebar Q_{\alpha''}(\rho\|\sigma) \right\} \,.
\end{align}
Finally, using the continuity of $\alpha \mapsto \widebar Q_{\alpha}(\rho\|\sigma)$ and using the fact that $\epsilon$ can be chosen arbitrarily small, we get
$D_\alpha^{\mathrm{reg}}(\rho\|\sigma) \geq \widebar D_\alpha(\rho_n\|\sigma_n)$, concluding the proof.

\subsection{Multiplicative continuity bounds in $\alpha$}

We briefly show an additional property of $\Hell_{\alpha}$ by giving a multiplicative continuity bound of the following form.
\begin{proposition}\label{Prop:ab-continuity}
    Let $\rho \ll\gg \sigma$ be quantum states and $\alpha, \beta \in [0,\infty)$. Then,
    \begin{align}
        \kappa(\alpha,\beta) \Hell_\beta(\rho\|\sigma) \leq \Hell_\alpha(\rho\|\sigma) \leq \kappa(\beta,\alpha)^{-1} \Hell_\beta(\rho\|\sigma), 
    \end{align}
    where
    \begin{align}
        \kappa(\alpha,\beta) = \begin{cases}
        2^{(\beta-\alpha)D_{\max}(\sigma\|\rho)} & \alpha\geq\beta \\
         2^{(\alpha-\beta)D_{\max}(\rho\|\sigma)} & \alpha\leq\beta
        \end{cases} .
    \end{align}
\end{proposition}
\begin{proof}
    The second inequality follows directly from the first. 
    We note that to give multiplicative bounds between say $\Hell_\alpha(\rho\|\sigma)$ and $\Hell_\beta(\rho\|\sigma)$ it suffices to find $\kappa$ such that 
    \begin{align}
        \kappa \gamma^{\beta-2} &\leq \gamma^{\alpha-2} \qquad &&\forall \gamma\in [1,2^{D_{\max}(\rho\|\sigma)}] \\
        \kappa \gamma^{-\beta-1} &\leq \gamma^{-\alpha-1} \qquad &&\forall \gamma\in [1,2^{D_{\max}(\sigma\|\rho)}]. 
    \end{align}
    The above claim follows by examining the different cases. 
\end{proof}
Continuity bounds on $D_{\alpha}$ close to $\alpha = 1$ follow immediately from this, but they are not of multiplicative form. Nonetheless, we find the following bound.
\begin{corollary}
    Let $\rho \ll\gg \sigma$ be quantum states.
    For $\alpha\geq1$, we have $D_\alpha(\rho\|\sigma) \leq 2^{(\alpha-1)D_{\max}(\rho\|\sigma)} D(\rho\|\sigma)$ and for $\alpha \leq 1$, we have $D_\alpha(\rho\|\sigma) \geq 2^{(\alpha-1)D_{\max}(\rho\|\sigma)} D(\rho\|\sigma)$.
\end{corollary}
\begin{proof}
    The proof follows by choosing $\beta=1$ in Proposition~\ref{Prop:ab-continuity} and the bound $\log(1+x)\leq x$ for $x\geq -1$. 
\end{proof}


\section{Contraction coefficients}\label{Sec:ContractionCoef}

Contraction coefficients give a strong version of data processing inequalities by quantifying the decrease in information.
First, define the contraction coefficients for our $f$-divergence as
\begin{align}
\eta_f(\cA,\sigma) &:= \sup_{\rho} \frac{D_f(\cA(\rho) \| \cA(\sigma))}{D_f(\rho\|\sigma)},\\
\eta_f(\cA) &:= \sup_\sigma \eta_f(\cA,\sigma) =  \sup_{\rho,\sigma} \frac{D_f(\cA(\rho) \| \cA(\sigma))}{D_f(\rho\|\sigma)}. 
\end{align}
Clearly we have $0\leq \eta_f(\cA,\sigma)\leq \eta_f(\cA) \leq 1$. By definition the contraction coefficients give us the best channel dependent constants such that
\begin{alignat}{2}
    D(\cA(\rho) \| \cA(\sigma)) &\leq \eta_f(\cA,\sigma) D(\rho\|\sigma) \qquad &&\forall\rho \\ 
    D(\cA(\rho) \| \cA(\sigma)) &\leq \eta_f(\cA) D(\rho\|\sigma) &&\forall\rho,\sigma.
\end{alignat}
An important special case is that of the relative entropy. We define for a quantum channel $\cA$,  
\begin{align}
    \eta_\Re(\cA) := \eta_{x\log x}(\cA) = \sup_{\rho,\sigma} \frac{D(\cA(\rho) \| \cA(\sigma))}{D(\rho\|\sigma)}. 
\end{align}

Besides the relative entropy we are also interested in the contraction of the hockey-stick divergence and define
\begin{align}
    \eta_\gamma(\cA) = \sup_{\rho,\sigma} \frac{E_\gamma(\cA(\rho) \| \cA(\sigma))}{E_\gamma(\rho\|\sigma)}.
\end{align}
While contraction coefficients can often be hard to compute, it was recently shown in~\cite{hirche2022quantum} that $\eta_\gamma(\cA)$ has a comparably simple form, 
\begin{align}
    \eta_\gamma(\cA) = \sup_{\ket{\Psi}\perp\ket{\Phi}} E_\gamma(\cA(\ket{\Psi}\bra{\Psi}) \| \cA(\ket{\Phi}\bra{\Phi})),
\end{align}
where the optimization is over two orthogonal pure states. An important special case of this is
\begin{align}
    \eta_{\tr}(\cA) := \eta_1(\cA),  
\end{align}
which is also often called the Dobrushin coefficient. For a broader overview over the topic of contraction coefficients we refer to~\cite{hirche2022contraction,hiai2016contraction}.

\subsection{Bounds between contraction coefficients} 

Comparing contraction coefficients for different divergences can be an important tool when determining properties of strong data-processing inequalities. Providing such comparisons will be the focus of this section. 
First, we give a general upper bound on the contraction coefficient of any $f$-divergence in terms of that of the trace distance. 
\begin{lemma}\label{Lem:ConBoundTrace}
    For any quantum channel $\cA$, we have
         \begin{align}
             \eta_f(\cA) \leq \eta_{\tr}(\cA), 
          \end{align}
including, as a special case, 
\begin{align}
    \eta_{\Re}(\cA) \leq \eta_{\tr}(\cA). 
\end{align}
\end{lemma}
\begin{proof}
By definition of the $f$-divergence we have, 
\begin{align}
    D_f(\cA(\rho)\|\cA(\sigma)) &=\int_1^\infty (f''(\gamma) E_\gamma(\cA(\rho)\|\cA(\sigma)) + \gamma^{-3} f''(\gamma^{-1})E_\gamma(\cA(\sigma)\|\cA(\rho))) \d\gamma \\
    &\leq \int_1^\infty (f''(\gamma) \eta_\gamma(\cA) E_\gamma(\rho\|\sigma) + \gamma^{-3} f''(\gamma^{-1}) \eta_\gamma(\cA) E_\gamma(\sigma\|\rho)) \d\gamma \\
    &\leq \eta_{\tr}(\cA) \int_1^\infty (f''(\gamma)  E_\gamma(\rho\|\sigma) + \gamma^{-3} f''(\gamma^{-1})  E_\gamma(\sigma\|\rho)) \d\gamma \\
    &= \eta_{\tr}(\cA)\, D_f(\rho\|\sigma),
\end{align}
where the first inequality is by definition of $\eta_\gamma$, the second by its monotonicity in $\gamma$ and the final equality is again by definition. 
\end{proof}
This is an often useful bound because of how much easier it is to compute the contraction coefficient for the trace distance. This result was previously conjectured in~\cite{ruskai1994beyond}, however with a different definition for the $f$-divergence. However, this bound might be loose. Classically it is known that for a wide class of $f$-divergences their contraction coefficients are the same~\cite{cohen1993relative,cohen1998comparisons,raginsky2016strong}. Whether the same holds true in the quantum setting is a long standing open problem. In the following we will give several bounds between contraction coefficients, ultimately showing that the classical results translate to the quantum setting when considering the new $f$-divergence.  
\begin{corollary}\label{Cor:x2-lower-bound}
    For any twice differentiable $f\in\cF$ with $0<f''(1)<\infty$, we have
    \begin{align}
        \eta_{x^2}(\cA,\sigma) &\leq \eta_f(\cA,\sigma), \\
        \eta_{x^2}(\cA) &\leq \eta_f(\cA),
    \end{align}
    where $\eta_{x^2}$ is the contraction coefficient for the $\Hell_2$ divergence. 
\end{corollary}
\begin{proof}
    By definition we have,
    \begin{align}
        D_f(\cA(\lambda\rho+(1-\lambda)\sigma)\|\cA(\sigma)) \leq \eta_f(\cA,\sigma) D_f(\lambda\rho+(1-\lambda)\sigma\|\sigma).
    \end{align}
    Multiplying by $\frac2\lambda$ and taking the limit $\lambda\rightarrow 0$ gives the first claim by using Theorem~\ref{Thm:H2-is-limit-of-f}. The second follows by taking the supremum over all $\sigma$. 
\end{proof}

We are also interested in corresponding upper bounds. While for some $f$ the contraction coefficients in the previous corollary might be quite different, we can make useful statements when we further restrict the set of functions to operator convex functions. In this case we can choose between a large number of integral representations for the operator convex function itself. In particular, here we consider the following representation proven in~\cite[Theorem 8.1]{hiai2011quantum} specialized to the case of $f(1)=0$, 
\begin{align}
    f(x) = a(x-1) + b(x^2-1) + \int_0^\infty \left( \frac{x}{1+t} - \frac{x}{x+t} \right) d\mu(t), 
\end{align}
for some real $a$, some $b\geq 0$ and a non-negative measure $\mu$ on $(0,\infty)$ with
\begin{align}
    \int_0^\infty \frac{d\mu(t)}{(1+t)^2} < \infty. 
\end{align}
The constants are uniquely determined by 
\begin{align}
    a+b=-f(0), \qquad b=\lim_{x\rightarrow\infty} \frac{f(x)}{x^2}. 
\end{align}
Using this representation we can get the following alternative representation of $f$-divergences. 
\begin{corollary} \label{Cor:LC-representation}
For any operator convex $f\in\cF$, we have
    \begin{align}
        D_f(\rho\|\sigma) = b\Hell_2(\rho\|\sigma) + \int_0^1 \lambda\, LC_\lambda(\rho\|\sigma) d\nu(\lambda),
    \end{align} 
with $b$ given as above and some non-negative measure $\nu$.
\end{corollary}
\begin{proof}
    Noting that polynomials of degree one or two do not contribute to the $f$-divergence, see Proposition~\ref{Prop:Df-properties}, we get that for every operator convex $f$ we can equivalently use 
    \begin{align}
        \tilde f(x) = bx^2 - \int_0^\infty \left( \frac{x}{x+t} \right) d\mu(t).
    \end{align}
    This results in
    \begin{align}
        D_f(\rho\|\sigma) &= D_{\tilde f}(\rho\|\sigma) \\
        &= b\Hell_2(\rho\|\sigma) + \int_0^\infty D_{\frac{-x}{x+t}}(\rho\|\sigma) d\mu(t) \\
        &= b\Hell_2(\rho\|\sigma) + \int_0^1 D_{\frac{-\lambda x}{\lambda x+(1-\lambda)}}(\rho\|\sigma) d\nu(\lambda)  \\
        &= b\Hell_2(\rho\|\sigma) + \int_0^1 \lambda\, LC_\lambda(\rho\|\sigma) d\nu(\lambda),
    \end{align} 
    where the third equality follows by substituting $\lambda=\frac{1}{t+1}$ and choosing $d\nu(\lambda) = \frac{d\mu(t)}{(t+1)^2}$. The forth inequality holds because $\frac{-x}{\lambda x+(1-\lambda)}$ and $g_\lambda(x)$ as defined for Equation~\eqref{Eq:LeCam-def} can be checked to have the same second derivative. 
\end{proof}
This new representation can be used to prove the following bounds on contraction coefficients. 
\begin{corollary}\label{Cor:x2-LC-upper-bound}
    For any operator convex $f\in\cF$, we have
    \begin{align}
        \eta_f(\cA,\sigma) &\leq \max\left\{\eta_{x^2}(\cA,\sigma),\sup_{0\leq\lambda\leq 1} \eta_{LC_\lambda}(\cA,\sigma)\right\}, \label{Eq:x2-LC-upper-bound-1}\\
        \eta_f(\cA) &\leq \max\left\{\eta_{x^2}(\cA),\sup_{0\leq\lambda\leq 1} \eta_{LC_\lambda}(\cA)\right\}.
    \end{align}
\end{corollary}
\begin{proof}
    By the above results we get
    \begin{align}
    D_f(\cA(\rho)\|\cA(\sigma)) &= b\Hell_2(\cA(\rho)\|\cA(\sigma)) + \int_0^1 \lambda\, LC_\lambda(\cA(\rho)\|\cA(\sigma)) d\nu(t) \\
    &\leq \eta_{x^2}(\cA,\sigma) \, b\, \Hell_2(\rho\|\sigma) + \int_0^1 \lambda\,\eta_{LC_\lambda}(\cA,\sigma)\, LC_\lambda(\rho\|\sigma) d\nu(t) \\
    &\leq \max\left\{\eta_{x^2}(\cA,\sigma),\sup_{0\leq\lambda\leq 1} \eta_{LC_\lambda}(\cA,\sigma)\right\} D_f(\rho\|\sigma),
    \end{align}
    which proves the first statement. The second follows by optimizing over $\sigma$. 
\end{proof}

\begin{corollary}\label{Cor:CC-LC-x2}
    We have that, 
    \begin{align}
        \eta_{LC_\lambda}(\cA) \leq \eta_{x^2}(\cA), 
    \end{align}
    and as a consequence, for every operator convex $f\in\cF$, we have
    \begin{align}
        \eta_f(\cA) &= \eta_{x^2}(\cA). \label{Eq:f-x2-contraction-eq}
    \end{align}
\end{corollary}
\begin{proof}
    We start by proving the first statement. Observe, 
    \begin{align}
        LC_\lambda(\cA(\rho)\|\cA(\sigma)) &= \lambda \Hell_2(\cA(\rho)\|\lambda\cA(\rho)+(1-\lambda)\cA(\sigma)) + (1-\lambda)\Hell_2(\cA(\sigma)\|\lambda\cA(\rho)+(1-\lambda)\cA(\sigma)) \\
        &\leq \lambda\, \eta_{x^2}(\cA)\, \Hell_2(\rho\|\lambda\rho+(1-\lambda)\sigma) + (1-\lambda)\, \eta_{x^2}(\cA) \, \Hell_2(\sigma\|\lambda\rho+(1-\lambda)\sigma) \\
        &= \eta_{x^2}(\cA)\, LC_\lambda(\rho\|\sigma),
    \end{align}
    where both equalities follow from Proposition~\ref{Cor:LC-H2-eq}. This establishes 
        \begin{align}
        \eta_{LC_\lambda}(\cA) \leq \eta_{x^2}(\cA), 
    \end{align}
    for all $0\leq\lambda\leq 1$ and as a consequence
    \begin{align}
        \sup_{0\leq\lambda\leq 1} \eta_{LC_\lambda}(\cA) \leq \eta_{x^2}(\cA).  
    \end{align}
    The second statement then follows by combining this with Corollary~\ref{Cor:x2-lower-bound} and Corollary~\ref{Cor:x2-LC-upper-bound}. 
\end{proof}
The classical special case of Equation~\eqref{Eq:f-x2-contraction-eq} was previously proven in~\cite{choi1994equivalence}. The classical version of the more general input dependent result in Equation~\eqref{Eq:x2-LC-upper-bound-1} was recently proven in~\cite{raginsky2016strong}. Finally, by noting that $g(x)=x\log(x)$ is operator convex, we can summarize that by Corollary~\ref{Cor:CC-LC-x2} and Lemma~\ref{Lem:ConBoundTrace} we have for every operator convex function $f$ that
    \begin{align}
        \eta_f(\cA) = \eta_{x^2}(\cA) = \eta_{\Re}(\cA) \leq \eta_{\tr}(\cA). \label{Eq:contraction-summary}
    \end{align}
Note that the equalities above do not hold in general for $\eta_f(\cA,\sigma)$ even when $f$ is operator convex. Next, we will discuss the depolarizing channel as an explicit example of this. 

\subsection{Contraction for Depolarizing Channels}

Consider the generalized depolarizing channel
\begin{align}
    \cD_{p,\sigma}(\rho) = (1-p) \rho + p \sigma. 
\end{align}
The speed at which this channel converges to its fixed point as measured by different divergences has been the subject of numerous publications~\cite{kastoryano2013quantum,muller2016relative,muller2018sandwiched}. This question is closely related to the strong data processing inequality for the corresponding divergence. For the new $f$-divergence, we have the following observation. 
\begin{proposition}\label{Prop:Depol-f}
    For $f\in\cF$ and quantum states $\rho,\sigma$, we have 
    \begin{align}
        D_f\!\left(\cD_{p,\sigma}(\rho)\|\sigma\right) = D_F(\rho\|\sigma),
    \end{align}
    with 
    \begin{align}
        F(x)=f((1-p)x+p). 
    \end{align}
\end{proposition}
\begin{proof}
    The proof works along the same lines to that Proposition~\ref{Prop:D_F_lambda-equality}, using the explicit form of the generalized depolarizing channel.  Recall from Equations~\eqref{Eq:HS-identities-cc-1r} and~\eqref{Eq:HS-identities-cc-2r},
    \begin{align}
        E_\gamma(\lambda\rho+(1-\lambda)\sigma \| \sigma) &= \lambda E_{\frac{\gamma+\lambda-1}{\lambda}}(\rho\|\sigma), \label{Eq:HS-identities-cc-1rr}
    \end{align}
   and for $\gamma\leq\frac{1}{1-\lambda}$, 
    \begin{align}
        E_\gamma(\sigma\| \lambda\rho+(1-\lambda)\sigma) 
        &= (1-\gamma(1-\lambda)) E_{\frac{\gamma\lambda}{1-\gamma(1-\lambda)}}(\sigma\|\rho).  \label{Eq:HS-identities-cc-2rr}
    \end{align}
   This implies
    \begin{align}
        E_\gamma(\cD_{p,\sigma}(\rho) \| \sigma) &= (1-p) E_{\frac{\gamma-p}{1-p}}(\rho\|\sigma), \label{Eq:HS-identities-cc-Dp}
        \end{align}
        and for $\gamma\leq\frac{1}{p}$,
        \begin{align}
        E_\gamma(\sigma\| \cD_{p,\sigma}(\rho)) 
        &= (1-\gamma p) E_{\frac{\gamma(1-p)}{1-\gamma p}}(\sigma\|\rho).\label{Eq:HS-identities-cc-Dp}
    \end{align}
    From here we get, 
    \begin{align}
        \int_1^\infty f''(\gamma) E_\gamma(\cD_{p,\sigma}(\rho) \| \sigma) \d\gamma= \int_1^\infty f''((1-p)\hat\gamma+p) (1-p)^2 E_{\hat\gamma}(\rho\| \sigma) d\hat\gamma,
    \end{align}
    by substituting $\hat\gamma=\frac{\gamma-p}{1-p}$, and
    \begin{align}
        \int_1^\infty \frac{f''(\gamma^{-1})}{\gamma^3} E_\gamma(\sigma \| \cD_{p,\sigma}(\rho)) \d\gamma &= \int_1^{\frac{1}{p}} \frac{f''(\gamma^{-1})}{\gamma^3} E_\gamma(\sigma \| D_{p,\sigma}(\rho)) \d\gamma \\
        &= \int_1^\infty \frac{f''\left((1-p)\hat\gamma^{-1}+p\right)}{\hat\gamma^3} (1-p)^2 E_{\hat\gamma}(\sigma\| \rho) d\hat\gamma,
    \end{align}
    where the second equality is substituting $\hat\gamma=\frac{\gamma(1-p)}{1-\gamma p}$. Our initial statement then follows by noting that $F''(x)=(1-p)^2 f''((1-p)x+p)$ which concludes the proof. 
\end{proof}

In particular, with the above result it can easily be checked that 
\begin{align}
    \Hell_2\!\left(\cD_{p,\sigma}(\rho)\|\sigma\right) = (1-p)^2 \Hell_2(\rho\|\sigma)
\end{align}
and as a result
\begin{align}
    \eta_{x^2}(\cD_{p,\sigma},\sigma) = (1-p)^2.
\end{align}
By the previous discussion it follows immediately that for every twice differentiable $f\in\cF$ with  $0<f''(1)<\infty$, we have
\begin{align}
    (1-p)^2 \leq \eta_{f}(\cD_{p,\sigma},\sigma) \leq \eta_{\tr}(\cD_{p,\sigma}) = (1-p). 
\end{align}
Note that both the lower and upper bound are independent of $\sigma$, while in general the SDPI constants will depend on it. These bounds can also not be improved while maintaining this property, however better $\sigma$ dependent bounds can be derived via Proposition~\ref{Prop:Depol-f}.
To see the optimality of the above bounds we recall the main result of~\cite{muller2016relative}, which proved that for the relative entropy
\begin{align}
    \eta_{\Re}(\cD_{p,\sigma},\sigma) = (1-p)^{2\alpha(\sigma)}
\end{align}
with $\alpha(\sigma)$ such that $\alpha(\tau_2)=1$ and $\lim_{d\rightarrow\infty}\alpha(\tau_d)=\frac12$ where $\tau_d$ is the maximally mixed state in dimension $d$. 

This discussion also exemplifies that while for all operator convex functions $f$,
\begin{align}
    \eta_{f}(\cD_{p,\sigma})= \eta_{x^2}(\cD_{p,\sigma}),
\end{align}
as discussed previously, the same does not hold for fixed $\sigma$ with the above giving an example of an operator convex $f$ for which 
\begin{align}
    \eta_{f}(\cD_{p,\sigma},\sigma) > \eta_{x^2}(\cD_{p,\sigma},\sigma).
\end{align}

\subsection{Contraction for \Renyi divergences}

Although, as discussed earlier, \Renyi divergences are not directly $f$-divergences, their contraction coefficients are still interesting and we shall try to compare them to the previously discussed ones. Define, 
\begin{align}
\eta_\alpha(\cA,\sigma) &:= \sup_{\rho} \frac{D_\alpha(\cA(\rho) \| \cA(\sigma))}{D_\alpha(\rho\|\sigma)}. 
\end{align}
We can relate this to the Hellinger contraction coefficient as follows. 
\begin{proposition}
    For a given quantum channel $\cA$, quantum state $\sigma$ and $0<\alpha<1$, we have, 
    \begin{align}
        \eta_\alpha(\cA,\sigma) \leq \eta_{\Hell_\alpha}(\cA,\sigma).
    \end{align}
\end{proposition}
\begin{proof}
    Choose $g(x)=(\alpha-1)^{-1}\log(1+(\alpha-1)x)$ and notice that it is monotone in $x$ and convex for $0<\alpha<1$ and therefore $g(\lambda x)\leq \lambda f(x)$ for $0\leq\lambda\leq 1$. 
    By definition we have, 
    \begin{align}
        D_\alpha(\cA(\rho) \| \cA(\sigma)) &= g(\Hell_\alpha(\cA(\rho) \| \cA(\sigma))) \\
        &\leq g\left(\eta_{\Hell_\alpha}(\cA,\sigma) \Hell_\alpha(\rho \| \sigma)\right) \\
        &\leq \eta_{\Hell_\alpha}(\cA,\sigma)\, g\left( \Hell_\alpha(\rho \| \sigma)\right) \\
        &= \eta_{\Hell_\alpha}(\cA,\sigma) D_\alpha(\rho \| \sigma).
    \end{align}
\end{proof}
Note that as a consequence, we also have for $0\leq\alpha\leq 1$,
\begin{align}
    \widebar D_\alpha(\cA(\rho) \| \cA(\sigma)) \leq \sup_n\left(\eta_{\Hell_\alpha}(\cA^{\otimes n},\sigma^{\otimes n})\right) \widebar D_\alpha(\rho \| \sigma),
\end{align}
for the Petz \Renyi divergence. 
It would be interesting to investigate whether this gives a non-trivial bound, beyond standard data processing. Note that is known that for example $\eta_{\Re}(\cD_{p,\tau_2}^{\otimes n},\tau_2^{\otimes n}) = (1-p)^2$ for all $n\geq1$~\cite{kastoryano2013quantum,muller2016relative}, which might give hope that such a result is possible.

\subsection{Less noisy channels}

Partial orders between channels have a wide range of applications~\cite{watanabe2012private,hirche2022bounding}. One of the most common orders is the \textit{less noisy} partial order. We call a channel $\cM$ less noisy than a channel $\cN$, denoted $\cM \sln \cN$, if
\begin{align}
    I(U:B)_{\cM(\rho)} \geq I(U:B')_{\cN(\rho)}\qquad\forall \rho_{UA}, \label{Eq:less-noisy-MI}
\end{align}
where $I(A:B)_{\rho_{AB}}=D(\rho_{AB}\|\rho_A\otimes \rho_B)$ is the quantum mutual information and $\rho_{UB}=(\id_U\otimes\cM)(\rho_{UA})$, $\rho_{UB'}=(\id_U\otimes\cN)(\rho_{UA})$ with the $\rho_{UA}$ being classical quantum states of the form 
\begin{align}
    \rho_{UA} = \sum_u p(u) |u\rangle\langle u| \otimes \rho^u_A. 
\end{align}
It was recently shown in~\cite{hirche2022contraction} that the condition in Equation~\eqref{Eq:less-noisy-MI} is equivalent to a condition in terms of the relative entropy, 
\begin{align}
    D(\cM(\rho)\|\cM(\sigma)) \geq D(\cN(\rho)\|\cN(\sigma))\qquad\forall\rho_A,\sigma_A. \label{Eq:less-noisy-D}
\end{align}
Classically, however, it is even known that this equivalence remains true if the relative entropy is replaced by any $f$-divergence with an operator convex non-linear $f$, see~\cite[Theorem 12]{makur2015linear}. 
We have now the tools to show that the same holds in the quantum setting for the new $f$-divergence. 
\begin{proposition}
    For any non-linear operator convex function $f\in\cF$ and quantum channels $\cM$ and $\cN$, the following two statements are equivalent:
    \begin{enumerate} 
    \item $\cM \sln \cN$ 
    \item $ D_f(\cM(\rho)\|\cM(\sigma)) \geq D_f(\cN(\rho)\|\cN(\sigma))\qquad\forall\rho_A,\sigma_A.   $ \label{Item:less-noisy-fD}
    \end{enumerate}
\end{proposition}
\begin{proof}
    Since $x\log(x)$ is operator convex it suffices to show the equivalence of Statement~\eqref{Item:less-noisy-fD} for different non-linear operator convex $f$. 
    Fix an applicable function $f$ and choose $f_2(x)=x^2-1$. We find that if~\eqref{Item:less-noisy-fD} holds for $f$ then in particular
    \begin{align}
        D_f(\cM(\lambda\rho+(1-\lambda)\sigma) \|\cM(\sigma)) \geq D_f(\cN(\lambda\rho+(1-\lambda)\sigma)\|\cN(\sigma))
    \end{align} 
    and by Theorem~\ref{Thm:H2-is-limit-of-f} it also holds for $f_2$. This is essentially the same argument as that in the proof of Corollary~\ref{Cor:x2-lower-bound}. Conversely, if~\eqref{Item:less-noisy-fD} holds for $f_2$, then by combining Corollary~\ref{Cor:LC-representation} and Corollary~\ref{Cor:LC-H2-eq} we have 
        \begin{align}
        D_f(\rho\|\sigma) = b\chi^2(\rho\|\sigma) + \int_0^1 \lambda\,\left( \lambda \chi^2(\rho\|\lambda\rho+(1-\lambda)\sigma) + (1-\lambda)\chi^2(\sigma\|\lambda\rho+(1-\lambda)\sigma)\right) \d\nu(t)
    \end{align} 
    and~\eqref{Item:less-noisy-fD} also holds for $f$. Again, this is essentially the argument that lead to Corollaries~\ref{Cor:x2-LC-upper-bound} and~\ref{Cor:CC-LC-x2}. Since this holds for any applicable $f$, the statement of the proposition follows. 
\end{proof}
This result can be compared to the case of the Petz $f$-divergence where the different partial orders are not known to be equivalent for different $f$, see~\cite[Proposition 6.7]{hirche2022contraction}. The above has implications e.g. for quantum capacities. Following the argument in~\cite{watanabe2012private,hirche2022bounding} we have, 
\begin{align}
    \chi^2(\cN^c(\rho)\|\cN^c(\sigma)) \geq \chi^2(\cN(\rho)\|\cN(\sigma)) \quad\Rightarrow\quad P^{(1)}(\cN) = 0, 
\end{align}
where $\cN^c$ is the complementary channel of $\cN$ and $P^{(1)}(\cN)$ is the private information of $\cN$. The same holds true for $n$ copies of the channel, implying a new condition for channels with zero private capacity. 

Finally, we will show that also the underlying equivalence of the above discussion between Equations~\eqref{Eq:less-noisy-MI} and~\eqref{Eq:less-noisy-D} generalizes to our $f$-divergences. This was shown for the Petz and maximal $f$-diveregences in~\cite[Proposition 6.2]{hirche2022contraction} and is a generalization of the classical case shown in~\cite{raginsky2016strong}. 
First we define the $f$-mutual information based on our $f$-divergence,
\begin{align}
    I_f(A:B)_{\rho_{AB}} := D_f(\rho_{AB}\|\rho_A\otimes \rho_B).
\end{align}
We now show the following. 
\begin{proposition}\label{Prop:If-equiv-Df}
    Let $\cM$ and $\cN$ be quantum channels, $\eta\geq 0$ and $\sigma_A$ a fixed quantum state with full rank. For given $f\in\cF$  the following are equivalent:
    \begin{itemize}
        \item[(i)] For all classical-quantum states $\rho_{UA}$ with marginal $\rho_A=\sigma_A$, where $U$ is an arbitrary classical system, 
        \begin{align}
            \eta I_f(U:B_1)_{\cM(\rho_{UA})} \geq I_f(U:B_2)_{\cN(\rho_{UA})}.
        \end{align}
        \item[(ii)] For any state $\rho_A$, 
        \begin{align}
            \eta D_f(\cM(\rho)\|\cM(\sigma)) \geq D_f(\cN(\rho)\|\cN(\sigma)).
        \end{align}
    \end{itemize}
\end{proposition}
\begin{proof}
    The proof follows along the same lines as that of~\cite[Proposition 6.2]{hirche2022contraction}.  
    Note that the hockey-stick divergence has the following direct-sum property,
    \begin{align}
        E_\gamma(\rho_{UA}\|\rho_U\otimes\rho_A) = \sum_u p(u) E_\gamma(\rho^u_A \| \rho_A). 
    \end{align}
    This carries over to the new $f$-divergence, resulting in
    \begin{align}
        I_f(U:B_1)_{\cM(\rho_{UA})} = \sum_u p(u) D_f(\cM(\rho^u)\|\cM(\sigma)). 
    \end{align}
    It therefore follows directly that (ii) implies (i). We now show the opposite direction. For an arbitrary $\rho_A$ fix $\epsilon>0$ such that $\sigma-\epsilon\rho\geq 0$. For any $0\leq\lambda\leq\epsilon$ define the binary random variable $U$ with $p(0)=\lambda$ and $p(1)=1-\lambda$ and conditional states $\rho_A^0=\rho_A$ and $\rho_A^1=(1-\lambda)^{-1}(\sigma_A-\lambda\rho_A)$. Clearly we have $\tr_U \rho_{UA}=\sigma_A$. Define,
    \begin{align}
        \varphi(\lambda) &:= \eta I_f(U:B_1)_{\cM(\rho_{UA})} - I_f(U:B_2)_{\cN(\rho_{UA})} \\
        &= \eta \lambda D_f(\cM(\rho^0)\|\cM(\sigma)) +\eta(1-\lambda) D_f(\cM(\rho^1)\|\cM(\sigma)) \nonumber\\
        &\quad- \lambda D_f(\cN(\rho^0)\|\cN(\sigma)) -(1-\lambda) D_f(\cN(\rho^1)\|\cN(\sigma)). 
    \end{align}
    It can easily be checked that $\varphi(0)=0$. Now, (i) implies $\varphi(\lambda)\geq 0$ and hence $\varphi'(0)\geq 0$. We compute 
    \begin{align}
        \varphi'(0) &= \eta D_f(\cM(\rho^0)\|\cM(\sigma)) - D_f(\cN(\rho^0)\|\cN(\sigma)) \nonumber\\
        &\quad + \frac{\partial}{\partial \lambda} \left(\eta D_f(\cM(\rho^1)\|\cM(\sigma)) - D_f(\cN(\rho^1)\|\cN(\sigma)) \right)\Big|_{\lambda = 0}. 
    \end{align}
    However, we have already argued in the proof of Theorem~\ref{Thm:H2-is-limit-of-f} that the remaining derivative vanishes and therefore (i) does indeed imply (ii). 
\end{proof}
Combining our results in the section, we have that for any two operator convex functions $f$ and $g$, we have 
\begin{align}
    &\cM \sln \cN \\
    &\iff D_f(\cM(\rho)\|\cM(\sigma)) \geq D_f(\cN(\rho)\|\cN(\sigma))\qquad\forall\rho_A,\sigma_A \\
    &\iff I_g(U:B_1)_{\cM(\rho_{UA})} \geq I_g(U:B_2)_{\cN(\rho_{UA})} \qquad\forall\rho_{UA}, 
\end{align}
giving a wide range of expressions for the less noisy partial order. 

Finally, returning to the main topic of this chapter, Proposition~\ref{Prop:If-equiv-Df} also has implications for contraction. In particular, for any suitable $f$, we have for the SDPI constant that
\begin{align}
    \eta_f(\cN,\sigma) = \sup_{\rho_{UA}} \frac{I_f(U:B)}{I_f(U:A)}, 
\end{align}
where $\rho_{UB}=\cN(\rho_{UA})$ and the supremum is over all classical-quantum states $\rho_{UA}$ with marginal $\tr_U\rho_{UA}=\sigma$. Furthermore for any two operator convex functions $f$ and $g$, we have 
\begin{align}
    \eta_f(\cN) = \sup_{\rho_{UA}} \frac{I_g(U:B)}{I_g(U:A)}, 
\end{align}
where now the supremum is over all classical-quantum states $\rho_{UA}$ without restriction on the marginal. 

\section{Other applications}
\label{sec:apps}

In this section we will present a number of additional applications of our integral representation and its properties. 

\subsection{Reverse Pinsker-type inequalities}\label{Sec:RevPinskerIneq}
We start the applications with reverse Pinsker type inequalities. As Pinsker inequality we commonly refer to the bound
\begin{align}
    2 E_1(\rho\|\sigma)^2 \leq D(\rho\|\sigma). \label{Eq:QPinsker}
\end{align}
Recall that $E_1(\rho\|\sigma)\equiv \|\rho-\sigma\|_{\tr}$. 
This is matched by the reverse bound~\cite{audenaert2005continuity,audenaert2011continuity},
\begin{align}
    D(\rho\|\sigma) \leq \frac{2}{\lambda_{\min}(\sigma)} E_1(\rho\|\sigma)^2. \label{Eq:Q-RevPin}
\end{align}
An improvement of the above was also given in~\cite{audenaert2005continuity,audenaert2011continuity} by
\begin{align}
    D(\rho\|\sigma) \leq (\lambda_{\min}(\sigma) + E_1(\rho\|\sigma)) \log\left(1+\frac{E_1(\rho\|\sigma)}{\lambda_{\min}(\sigma)}\right) - \lambda_{\min}(\rho)\log\left(1+\frac{E_1(\rho\|\sigma)}{\lambda_{\min}(\rho)}\right). \label{Eq:Aud-RevPinsker}
\end{align}
 Note, that in general Equation~\eqref{Eq:Aud-RevPinsker} is significantly tighter than the bound in Equation~\eqref{Eq:Q-RevPin}. Reverse Pinsker-type inequalities in the classical setting are discussed for the relative entropy in~\cite{csiszar2006context,sason2015upper,sason2015reverse} and for $f$-divergences in~\cite{binette2019note}. 

We will now discuss how to give reverse Pinsker type inequalities based on the new integral representation. First, we will need the following lemma.
\begin{lemma}\label{Lem:Egamma-E1-sharp}
    The function $\gamma \mapsto E_\gamma(\rho\|\sigma)$ is convex. As a result, for $1\leq \gamma \leq e^{D_{\max}(\rho\|\sigma)}$, we have, 
    \begin{align}
        E_\gamma(\rho\|\sigma) \leq \frac{e^{D_{\max}(\rho\|\sigma)}-\gamma}{e^{D_{\max}(\rho\|\sigma)}-1} E_1(\rho\|\sigma) \leq E_1(\rho\|\sigma). \label{Eq:Egamma-E1-sharp}
    \end{align}
    If the function is also linear, the first inequality becomes an equality. 
\end{lemma}
\begin{proof}
    We first proof the convexity. Observe, 
    \begin{align}
        E_{p\gamma+(1-p)\hat\gamma}(\rho\|\sigma) &= \tr(\rho - (p\gamma+(1-p)\hat\gamma)\sigma)_+ \\
       &=  \tr(p(\rho -\gamma\sigma) +  (1-p)(\rho -\hat\gamma\sigma))_+  \\
       &\leq p\tr(\rho -\gamma\sigma)_+ +  (1-p)\tr(\rho -\hat\gamma\sigma)_+  \\
       & = p E_\gamma(\rho\|\sigma) + (1-p) E_{\hat\gamma}(\rho\|\sigma), 
    \end{align}
    where the inequality holds because $\tr (A+B)_+ \leq \tr A_+ + \tr B_+$, which, for example, can be seen from the variational representation $\tr A_+ = \max_{0\leq S\leq\Id} \tr A S$. 
    Since the function is convex, we can upper bound it for $1\leq \gamma \leq e^{D_{\max}(\rho\|\sigma)}$ by the line that goes through $(1,E_1(\rho,\sigma))$ and $(e^{D_{\max}(\rho\|\sigma)},0)$. After some calculation one finds that the upper bound is given by the first inequality in Equation~\eqref{Eq:Egamma-E1-sharp}. The second inequality in the statement then follows directly from $\gamma\geq 1$ and recovers the previously known upper bound on $E_\gamma(\rho\|\sigma)$. 
\end{proof}
We start with a general result for $f$-divergences. For the following, we recall the definition of the Thompson metric~\cite{thompson1963certain} which is given by, 
\begin{align}
    \Xi(\rho\|\sigma) := \max\{ D_{\max}(\rho\|\sigma), D_{\max}(\sigma\|\rho) \}. 
\end{align}
This quantity has recently been given an operational interpretation in the context of hypothesis testing~\cite{regula2022postselected}. 
\begin{proposition}\label{Prop:Rev-Q-Pinsker}
    For quantum states $\rho,\sigma$ and $f\in\cF$, we have, 
    \begin{align}
    D_f(\rho\|\sigma) &\leq \zeta_1(\rho,\sigma) \, E_1(\rho\|\sigma)\leq \frac{f\left(e^{\Xi(\rho\|\sigma)}\right) + e^{\Xi(\rho\|\sigma)} f\left(e^{-\Xi(\rho\|\sigma)}\right)}{e^{\Xi(\rho\|\sigma)}-1} E_1(\rho\|\sigma) \label{Eq:RevPin-f-1}, 
    \end{align}
    and 
    \begin{align}
    D_f(\rho\|\sigma) &\leq \zeta_1(\rho,\sigma) \, E_1(\rho\|\sigma)\leq \zeta_2(\rho,\sigma) \, E_1(\rho\|\sigma) , \label{Eq:RevPin-f-2}
    \end{align}
    with 
    \begin{align}
    \zeta_1(\rho,\sigma) &= \int_1^{e^{D_{\max}(\rho\|\sigma)}} \frac{e^{D_{\max}(\rho\|\sigma)}-\gamma}{e^{D_{\max}(\rho\|\sigma)}-1} f''(\gamma) \d\gamma + \int_1^{e^{D_{\max}(\sigma\|\rho)}} \frac{e^{D_{\max}(\sigma\|\rho)}-\gamma}{e^{D_{\max}(\sigma\|\rho)}-1} \gamma^{-3} f''(\gamma^{-1}) \d\gamma,  \\
    \zeta_2(\rho,\sigma) &= \int_1^{e^{D_{\max}(\rho\|\sigma)}} f''(\gamma) \d\gamma + \int_1^{e^{D_{\max}(\sigma\|\rho)}} \gamma^{-3} f''(\gamma^{-1}) \d\gamma.
    \end{align}
    The first inequality holds as an equality whenever the functions $\gamma \mapsto E_\gamma(\rho\|\sigma)$ and $\gamma \mapsto E_\gamma(\sigma\|\rho)$ are linear. 
\end{proposition}
\begin{proof}
The first inequality can directly be seen by the following chain of arguments, 
\begin{align}
    D_f(\rho\|\sigma) &=\int_1^\infty (f''(\gamma) E_\gamma(\rho\|\sigma) + \gamma^{-3} f''(\gamma^{-1})E_\gamma(\sigma\|\rho)) \d\gamma \\
    &=\int_1^{e^{D_{\max}(\rho\|\sigma)}} f''(\gamma) E_\gamma(\rho\|\sigma) \d\gamma + \int_1^{e^{D_{\max}(\sigma\|\rho)}}\gamma^{-3} f''(\gamma^{-1})E_\gamma(\sigma\|\rho) \d\gamma \\
    &\leq E_1(\rho\|\sigma) \left( \int_1^{e^{D_{\max}(\rho\|\sigma)}} \frac{e^{D_{\max}(\rho\|\sigma)}-\gamma}{e^{D_{\max}(\rho\|\sigma)}-1} f''(\gamma) \d\gamma + \int_1^{e^{D_{\max}(\sigma\|\rho)}} \frac{e^{D_{\max}(\sigma\|\rho)}-\gamma}{e^{D_{\max}(\sigma\|\rho)}-1} \gamma^{-3} f''(\gamma^{-1}) \d\gamma \right),
\end{align}   
where the inequality follows from Lemma~\ref{Lem:Egamma-E1-sharp}. Now note that both integrals are monotone in the value of $D_{\max}$ which we can therefore bound by the Thompson metric $\Xi$. This simplifies our bound to, 
\begin{align}
    D_f(\rho\|\sigma) &\leq E_1(\rho\|\sigma) \left( \int_1^{e^{\Xi(\rho\|\sigma)}} \frac{e^{\Xi(\rho\|\sigma)}-\gamma}{e^{\Xi(\rho\|\sigma)}-1} f''(\gamma)  +  \frac{e^{\Xi(\sigma\|\rho)}-\gamma}{e^{\Xi(\sigma\|\rho)}-1} \gamma^{-3} f''(\gamma^{-1}) \d\gamma \right) \\
    &= \frac{E_1(\rho\|\sigma)}{e^{\Xi(\rho\|\sigma)}-1} \left( \int_1^{e^{\Xi(\rho\|\sigma)}} \left(e^{\Xi(\rho\|\sigma)}-\gamma\right) \left( f''(\gamma)  +  \gamma^{-3} f''(\gamma^{-1})\right) \d\gamma \right) \\
    &= \frac{f\left(e^{\Xi(\rho\|\sigma)}\right) + e^{\Xi(\rho\|\sigma)} f\left(e^{-\Xi(\rho\|\sigma)}\right)}{e^{\Xi(\rho\|\sigma)}-1} E_1(\rho\|\sigma) , 
\end{align}   
where the first equality is simple rearranging and the second follows from integration by parts, 
\begin{align}
    &\int_1^{e^{\Xi(\rho\|\sigma)}} \left(e^{\Xi(\rho\|\sigma)}-\gamma\right) \left( f''(\gamma)  +  \gamma^{-3} f''(\gamma^{-1})\right) \d\gamma \\
    &= \left(e^{\Xi(\rho\|\sigma)}-1\right) f(1) +  f\left(e^{\Xi(\rho\|\sigma)}\right) + e^{\Xi(\rho\|\sigma)} f\left(e^{-\Xi(\rho\|\sigma)}\right),  
\end{align}
for which we recall, 
\begin{align}
    \frac{\d^2}{\d\gamma^2} \gamma f\left(\gamma^{-1}\right) = \gamma^{-3} f''(\gamma^{-1}). 
\end{align}
Finally, the second inequality in Equation~\eqref{Eq:RevPin-f-2} follows easily from the second inequality in Lemma~\ref{Lem:Egamma-E1-sharp}. This concludes the proof. 
\end{proof}
We note here that while the function $\gamma \mapsto E_\gamma(\rho\|\sigma)$ is generally not linear, there are some interesting cases where it is. This includes all commuting qubit states. 
Related upper bounds for the standard $f$-divergence were discussed in~\cite{vershynina2019upper} via an integral representation for operator convex functions. We also note that a similar approach to the above was discussed for classical divergences in the context of differential privacy in~\cite{Zamanlooy2023}. 
To show the usefulness of our result we will discuss some special cases, starting with a thight bound on the relative entropy. 
\begin{proposition}
For quantum states $\rho,\sigma$, we have, 
    \begin{align}
    D(\rho\|\sigma) &\leq \left( \frac{e^{D_{\max}(\rho\|\sigma)}D_{\max}(\rho\|\sigma)}{e^{D_{\max}(\rho\|\sigma)}-1} + \frac{D_{\max}(\sigma\|\rho)}{1-e^{D_{\max}(\sigma\|\rho)}} \right) E_1(\rho\|\sigma) \label{Eq:NewRevPin-0}\\
    &\leq \Xi(\rho\|\sigma)\,  E_1(\rho\|\sigma),\label{Eq:NewRevPin-Thompson} 
\end{align}
and
\begin{align}
    D(\rho\|\sigma) &\leq \left( \frac{e^{D_{\max}(\rho\|\sigma)}D_{\max}(\rho\|\sigma)}{e^{D_{\max}(\rho\|\sigma)}-1} + \frac{D_{\max}(\sigma\|\rho)}{1-e^{D_{\max}(\sigma\|\rho)}} \right) E_1(\rho\|\sigma) \label{Eq:NewRevPin-0-r}\\
    &\leq  \left( 1 +  D_{\max}(\rho\|\sigma)  - e^{-D_{\max}(\sigma\|\rho)} \right) E_1(\rho\|\sigma) \label{Eq:NewRevPin-1}\\
    &\leq \Omega(\rho\|\sigma)\,  E_1(\rho\|\sigma),\label{Eq:NewRevPin} 
\end{align}
where $\Xi(\rho\|\sigma) := \max\{D_{\max}(\rho\|\sigma) , D_{\max}(\sigma\|\rho)\}$ is the Thompson metric and $\Omega(\rho\|\sigma) := D_{\max}(\rho\|\sigma) + D_{\max}(\sigma\|\rho)$ is the Hilbert projective metric. The first inequality above becomes an equality when the functions $\gamma \mapsto E_\gamma(\rho\|\sigma)$ and $\gamma \mapsto E_\gamma(\sigma\|\rho)$ are linear.  
\end{proposition}
\begin{proof}
    The first set of inequalities follows by evaluating the bounds in Equation~\eqref{Eq:RevPin-f-1} in Proposition~\ref{Prop:Rev-Q-Pinsker}. The second set follows from evaluating Equation~\eqref{Eq:RevPin-f-2}, besides the last inequality, which  follows from $1-x\leq e^{-x}$.  
\end{proof}
We compare our bound on the relative entropy against known ones by example in Figure~\ref{Fig:ExamplesRevPin}. We observe that Equation~\eqref{Eq:NewRevPin-0} gives better bounds then previous results for our examples and, as expected, often even equals the relative entropy. The weaker bounds in Equation~\eqref{Eq:NewRevPin-1} and Equation~\eqref{Eq:NewRevPin-Thompson} are generally incomparable to the previous bound and also with each other, in the sense that for either one can find examples where they perform better than the others. 

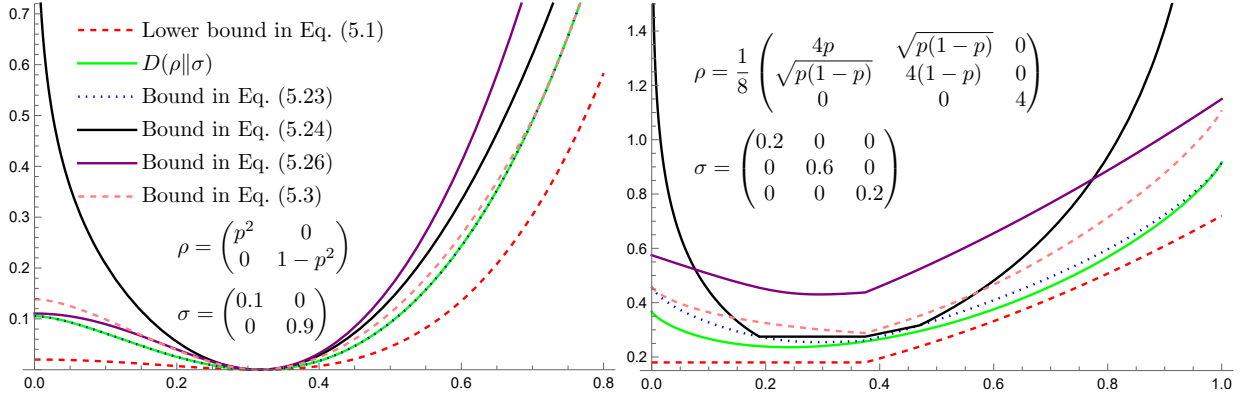
\begin{figure}[t]
\centering
\begin{minipage}{.49\textwidth}
\scalebox{0.83}{\begin{tikzpicture}
    \node[anchor=south west,inner sep=0] (image) at (0,0) {\includegraphics[width=1.205\textwidth]{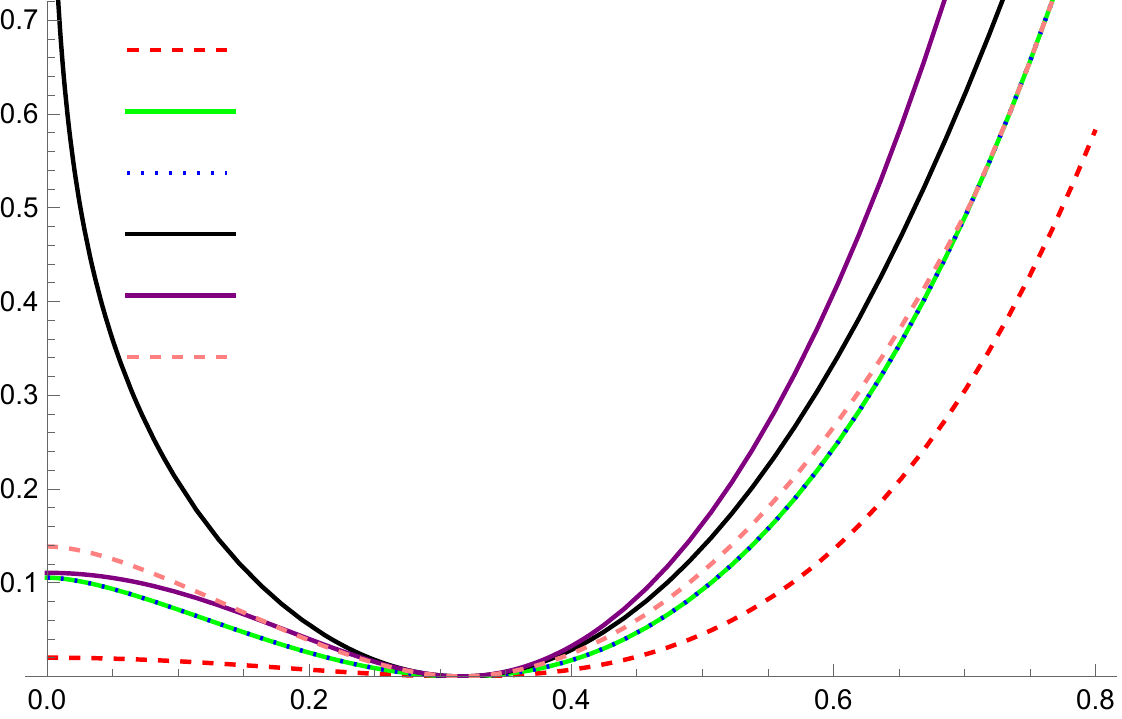}};
    \begin{scope}[x={(image.south east)},y={(image.north west)}]
        \node[] at (0.418,0.714){$\begin{aligned}
        &\text{Lower bound in Eq.~\eqref{Eq:QPinsker}} \\
        &D(\rho\|\sigma) \\
        &\text{Bound in Eq.~\eqref{Eq:NewRevPin-0}} \\
        &\text{Bound in Eq.~\eqref{Eq:NewRevPin-Thompson}} \\
        &\text{Bound in Eq.~\eqref{Eq:NewRevPin-1}} \\
        &\text{Bound in Eq.~\eqref{Eq:Aud-RevPinsker}}
        \end{aligned}$};
        \node[] at (0.42,0.287){$\begin{aligned}
        &\rho = \left( \begin{matrix} p^2 & 0 \\ 0 & 1-p^2 \end{matrix}\right) \\[1mm]
        &\sigma=\left( \begin{matrix} 0.1 & 0 \\  0 & 0.9 \end{matrix}\right) 
        \end{aligned}$};
    \end{scope}
\end{tikzpicture}}
\end{minipage}
\begin{minipage}{.49\textwidth}
\scalebox{0.83}{\begin{tikzpicture}
    \node[anchor=south west,inner sep=0] (image) at (0,0) {\includegraphics[width=1.205\textwidth]{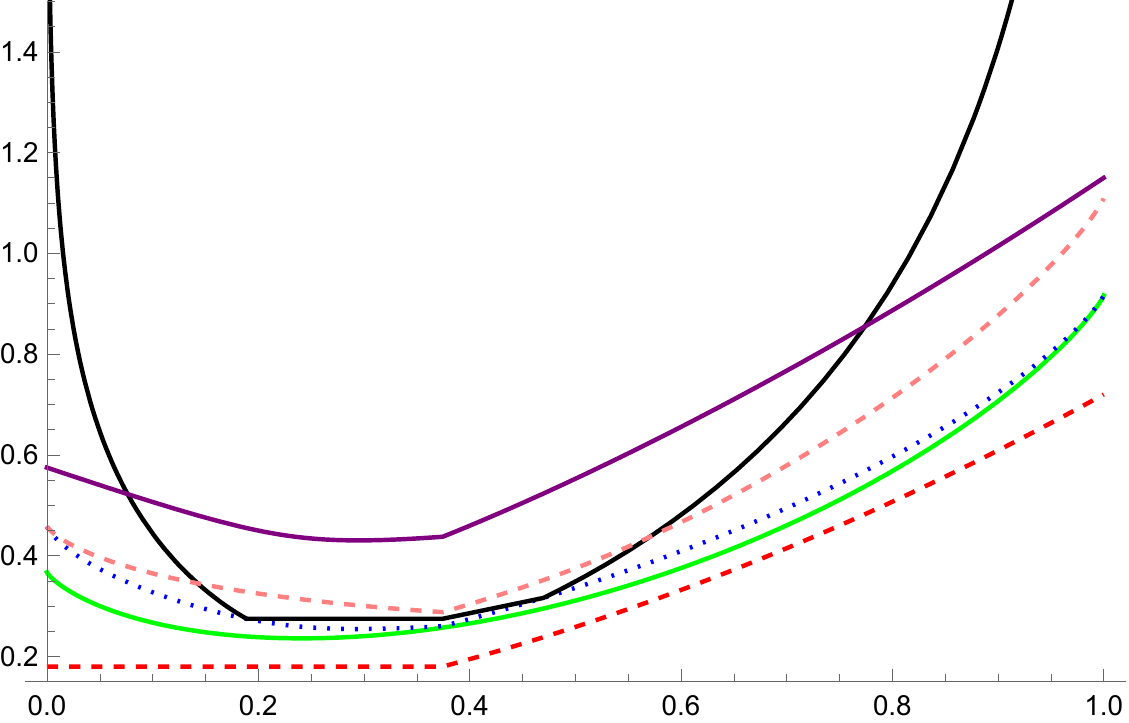}};
    \begin{scope}[x={(image.south east)},y={(image.north west)}]
        \node[] at (0.4,0.7){$\begin{aligned}
        &\rho = \frac18\left( \begin{matrix} 4p & \sqrt{p(1-p)} & 0 \\ \sqrt{p(1-p)} & 4(1-p) & 0 \\ 0 & 0 & 4 \end{matrix}\right) \\[1mm]
        &\sigma=\left( \begin{matrix} 0.2 & 0 & 0 \\  0 & 0.6 & 0 \\ 0 & 0 & 0.2 \end{matrix}\right) 
        \end{aligned}$};
    \end{scope}
\end{tikzpicture}}
\end{minipage}
\caption{\label{Fig:ExamplesRevPin} Plot of the relative entropy and several bounds on it for different states over the parameter $p$. (Left) Example for $d=2$. The states commute and hence our bound in Eq.~\eqref{Eq:NewRevPin-0} equals the relative entropy. Our bound in Eq.~\eqref{Eq:NewRevPin-1} is incomparable to the previous bound in Eq.~\eqref{Eq:Aud-RevPinsker}. (Right) Example for $d=3$. Eq.~\eqref{Eq:NewRevPin-0} gives the closest upper bound and is tight for some values of $p$.}
\end{figure}

The bounds based on the first inequality in Proposition~\ref{Prop:Rev-Q-Pinsker} are tight, but often a bit messy. In the remainder we focus on the somewhat simpler bounds based on the second inequality in Equation~\eqref{Eq:RevPin-f-2}, but we remark that one can always get tighter bounds by using the first inequality. We continue with the $f$-divergences discussed in Proposition~\ref{Prop:D_F_lambda-equality}. 
\begin{corollary}
For quantum states $\rho,\sigma$ and $F_\lambda=F_{\lambda,\lambda}$, $\hat F_\lambda=F_{\lambda,1-\lambda}$ with $F_{\lambda,\mu}$ defined in Equation~\eqref{Eq:def-F-lambda-mu}, we have, 
    \begin{align}
        D_{F_\lambda}(\rho\|\sigma) &= (1-\lambda) D(\rho\|\lambda\rho+(1-\lambda)\sigma) + \lambda D(\sigma\|\lambda\rho+(1-\lambda)\sigma) \\
        &\leq \left( \frac{1-2\lambda}{1-\lambda+\lambda e^{D_{\max}(\rho\|\sigma)}} + (1-\lambda) [ D_{\max}(\rho\|\sigma)-\log(1-\lambda + \lambda e^{D_{\max}(\rho\|\sigma)})] \right. \nonumber\\
        &\left. \quad- \frac{1-2\lambda}{\lambda + (1-\lambda)e^{D_{\max}(\sigma\|\rho)}} + \lambda [D_{\max}(\sigma\|\rho) - \log(\lambda+(1-\lambda)e^{D_{\max}(\sigma\|\rho)})] \right) E_1(\rho\|\sigma)
    \end{align}
    and 
    \begin{align}
        D_{\hat F_\lambda}(\rho\|\sigma) &= \lambda D(\rho\|\lambda\rho+(1-\lambda)\sigma) + (1-\lambda) D(\sigma\|\lambda\rho+(1-\lambda)\sigma) \\
        &\leq \left( \lambda [ D_{\max}(\rho\|\sigma) - \log(\lambda e^{D_{\max}(\rho\|\sigma)}+1-\lambda)] \right. \nonumber\\
        &\left. \quad+ (1-\lambda) [ D_{\max}(\sigma\|\rho) - \log( (1-\lambda) e^{D_{\max}(\sigma\|\rho)} +\lambda)] \right) E_1(\rho\|\sigma) \\
        &\leq h(\lambda) E_1(\rho\|\sigma), 
    \end{align} 
    where $h(\lambda)=-\lambda\log\lambda - (1-\lambda)\log(1-\lambda)$ is the binary entropy function. 
\end{corollary}
\begin{proof}
    First, we check the second derivative of $F_{\lambda}(x)$ to be
    \begin{align}
        F_{\lambda}''(x) = \frac{\lambda^3(x-1)+3\lambda^2-3\lambda+1}{x(1-\lambda+\lambda x)^2}. 
    \end{align}
    From which one can get that 
    \begin{align}
        &\int_1^{e^{D_{\max}(\rho\|\sigma)}} F_{\lambda}''(\gamma) \d\gamma \\
        &= \frac{1-2\lambda}{1-\lambda+\lambda e^{D_{\max}(\rho\|\sigma)}} + (1-\lambda) [ D_{\max}(\rho\|\sigma) -\log(1-\lambda + \lambda e^{D_{\max}(\rho\|\sigma)})] - (1-2\lambda) 
    \end{align}
    and similarly,
    \begin{align}
        &\int_1^{e^{D_{\max}(\sigma\|\rho)}} \gamma^{-3} F_{\lambda}''(\gamma^{-1}) \d\gamma \\
        &= - \frac{1-2\lambda}{\lambda + (1-\lambda)e^{D_{\max}(\sigma\|\rho)}} + \lambda [D_{\max}(\sigma\|\rho) - \log(\lambda+(1-\lambda)e^{D_{\max}(\sigma\|\rho)})] + (1-2\lambda) .
    \end{align}
In the same spirit we can also derive the second derivative of $\hat F_{\lambda}(x)$ to be
\begin{align}
    \hat F_{\lambda}''(x) = \frac{\lambda(1-\lambda)}{\lambda(x-1)x+x}. 
\end{align}
This implies
    \begin{align}
        &\int_1^{e^{D_{\max}(\rho\|\sigma)}} \hat F_{\lambda}''(\gamma) \d\gamma 
        = \lambda [ D_{\max}(\rho\|\sigma) - \log(\lambda e^{D_{\max}(\rho\|\sigma)}+1-\lambda)] 
    \end{align}
    and similarly,
    \begin{align}
        &\int_1^{e^{D_{\max}(\sigma\|\rho)}} \gamma^{-3} \hat F_{\lambda}''(\gamma^{-1}) \d\gamma 
        = (1-\lambda) [ D_{\max}(\sigma\|\rho) - \log( (1-\lambda)e^{D_{\max}(\sigma\|\rho)} +\lambda)]. 
    \end{align}
\end{proof}
This generalizes and improves a couple of quite useful results. Specializing to $F_0$ we recover the bound for the relative entropy that we stated earlier in Equation~\eqref{Eq:NewRevPin-1}, and for $F_{\frac12}$ we get,
\begin{align}
        &D_{F_{\frac12}}(\rho\|\sigma) \\
        &\quad = JS(\rho\|\sigma)  \\
        &\quad \leq {\frac12}\left(   D_{\max}(\rho\|\sigma)+D_{\max}(\sigma\|\rho) -\log \left({\frac12} + {\frac12} e^{D_{\max}(\rho\|\sigma)} \right)  - \log \left({\frac12}+{\frac12}e^{D_{\max}(\sigma\|\rho)}\right) \right) E_1(\rho\|\sigma)\\ 
        &\quad \leq \log(2) E_1(\rho\|\sigma). \label{Eq:RP-JS-2}
\end{align}
The weaker upper bound on the quantum Jensen-Shannon divergence in Equation~\eqref{Eq:RP-JS-2} was previously proven in~\cite{briet2009properties}. 

We continue by briefly giving an extension of the reverse Pinsker type inequalities above to the new integral \Renyi divergences. 
\begin{corollary}\label{Cor:cH-RevPin}
    For quantum states $\rho,\sigma$ and $\alpha \in (0, 1) \cup (1, \infty)$, we have, 
    \begin{align}
        \Hell_\alpha(\rho\|\sigma) \leq \frac{1}{\alpha-1}\left( \alpha e^{(\alpha-1) D_{\max}(\rho\|\sigma)} - (\alpha-1)e^{-\alpha D_{\max}(\sigma\|\rho)} -1 \right) E_1(\rho\|\sigma), \label{Eq:cH-RevPin}
    \end{align}
    and therefore
    \begin{align}
        D_\alpha(\rho\|\sigma) \leq \frac{1}{\alpha-1} \log\left( 1 + \left( \alpha e^{(\alpha-1) D_{\max}(\rho\|\sigma)} - (\alpha-1)e^{-\alpha D_{\max}(\sigma\|\rho)} -1 \right) E_1(\rho\|\sigma)\right). 
    \end{align}
\end{corollary}
\begin{proof}
    This follows easily from Proposition~\ref{Prop:Rev-Q-Pinsker}, by choosing $f(x)=\frac{x^\alpha-1}{\alpha-1}$, resulting in $f''(x)=\alpha x^{\alpha-2}$, and evaluating the needed integrals. 
\end{proof}
In the classical setting, reverse Pinsker inequalities for \Renyi divergences were given in~\cite{sason2015upper}. 
We end this section by noting that the above can be used to give an improvement on the well-known Fuchs-van-de-Graaf inequality, bounding the quantum fidelity $F(\rho,\sigma) = \|\sqrt{\rho}\sqrt{\sigma}\|_1$. 
\begin{corollary}
    For any quantum states $\rho,\sigma$, we have, 
    \begin{align}
        F(\rho,\sigma) \geq 1 - \frac12\left(2-e^{-\frac12 D_{\max}(\rho\|\sigma)}-e^{-\frac12 D_{\max}(\sigma\|\rho)}\right)E_1(\rho\|\sigma). \label{Eq:FvdG-improv}
    \end{align}
\end{corollary}
\begin{proof}
    It is well known that 
    \begin{align}
        -2\log F(\rho,\sigma) = \widetilde D_{\frac12}(\rho\|\sigma) = \widecheck{D}_{\frac12}(\rho\|\sigma), 
    \end{align}
    where the last quantity is the measured \Renyi divergence defined earlier in Equation~\eqref{Eq:meas-RRE}, here for $\alpha=\frac12$. Note that the second equality is particular to the $\alpha=\frac12$ case, as one needs to regularize the right hand side to get equality for general $\alpha$. As a consequence, we have
    \begin{align}
        -2\log F(\rho,\sigma) = \widecheck{D}_{\frac12}(\rho\|\sigma) \leq D_{\frac12}(\rho\|\sigma)
    \end{align}
    and therefore, 
    \begin{align}
        F(\rho,\sigma) \geq 1 - \frac12 \Hell_{\frac12}(\rho\|\sigma). 
    \end{align}  
    The claim then follows directly by applying Equation~\eqref{Eq:cH-RevPin} from Corollary~\ref{Cor:cH-RevPin} for $\alpha=\frac12$. 
\end{proof}
Note that a different improvement on the Fuchs-van-de-Graaf inequality was given in~\cite{zhang2016lower} as
\begin{align}
    F(\rho,\sigma) \geq 1 - \left(\frac{e^{\frac12 D_{\max}(\rho\|\sigma)}}{1+e^{\frac12 D_{\max}(\rho\|\sigma)}}\right)E_1(\rho\|\sigma). \label{Eq:FvdG-improv-old}
\end{align}
Here we only briefly note that it is easy to find states showing that neither improvement is generally stronger than the other. A possibly useful observation is that our bound is naturally symmetric under exchanging $\rho$ and $\sigma$, just as the fidelity and trace distance themselves. 
Our bounds also implies 
\begin{align}
    F(\rho,\sigma) = 1 - E_1(\rho\|\sigma) \quad\implies\quad \rho=\sigma \quad \lor \quad D_{\max}(\rho\|\sigma)=\infty=D_{\max}(\sigma\|\rho), 
\end{align}
which is slightly stronger than the statement made in the conclusions of ~\cite{zhang2016lower}. See also~\cite{audenaert2012comparisons} for a discussion of equality scenarios.

\subsection{Continuity in the first argument} \label{Sec:ContinuityFirst}

In this section we are interested in the continuity of divergences in their first argument, i.e. in upper bounds on the difference
\begin{align}
    D_f(\rho\|\sigma) - D_f(\tau\|\sigma). 
\end{align}
For this we first need an observation about the continuity of hockey-stick divergences. 
\begin{corollary}\label{Cor:HS-continuity}
For any quantum states $\rho,\tau,\sigma$ and $\gamma\geq 1$, we have
\begin{align}
        E_\gamma(\rho\|\sigma)-E_\gamma(\tau\|\sigma) &\leq E_1(\rho\|\tau), \\
        E_\gamma(\sigma\|\rho)-E_\gamma(\sigma\|\tau) &\leq \gamma E_1(\tau\|\rho). 
\end{align}
\end{corollary}
\begin{proof}
Recall the triangle inequality for the quantum hockey-stick divergence, 
\begin{align}
    E_{\gamma_1\gamma_2}(\rho\|\sigma) \leq E_{\gamma_1}(\rho\|\tau) + \gamma_1 E_{\gamma_2}(\tau\|\sigma). 
\end{align}
Then the first inequality follows directly by choosing $\gamma_1=1$ and $\gamma_2=\gamma$. The second follows by choosing $\gamma_1=\gamma$ and $\gamma_2=1$ and appropriately relabeling the states. 
\end{proof}
With these we can make the following general statement. 
\begin{proposition}\label{Prop:Cont-general}
For any quantum states $\rho,\tau,\sigma$ and $f\in\cF$, we have,  
    \begin{align}
        D_f(\rho\|\sigma) - D_f(\tau\|\sigma) &\leq \left( \int_1^{e^{D_{\max}(\rho\|\sigma)}} f''(s)  ds + \int_1^{e^{D_{\max}(\sigma\|\rho)}} \frac{1}{s^2}f''\left(\frac1s\right) ds \right) E_1(\rho\|\tau).
    \end{align} 
\end{proposition}
\begin{proof}
    \begin{align}
    &D_f(\rho\|\sigma) - D_f(\tau\|\sigma) \\
    &= \int_1^\infty \left( f''(s) \left(E_s(\rho\|\sigma)-E_s(\tau\|\sigma)\right) + \frac{1}{s^3}f''\left(\frac1s\right)\left(E_s(\sigma\|\rho)-E_s(\sigma\|\tau)\right)\right) ds \\ 
    &\leq \int_1^{e^{D_{\max}(\rho\|\sigma)}} f''(s) \left(E_s(\rho\|\sigma)-E_s(\tau\|\sigma)\right) ds + \int_1^{e^{D_{\max}(\sigma\|\rho)}} \frac{1}{s^3}f''\left(\frac1s\right)\left(E_s(\sigma\|\rho)-E_s(\sigma\|\tau)\right) ds \\
    &\leq \left( \int_1^{e^{D_{\max}(\rho\|\sigma)}} f''(s)  ds + \int_1^{e^{D_{\max}(\sigma\|\rho)}} \frac{1}{s^2}f''\left(\frac1s\right) ds \right) E_1(\rho\|\tau),
\end{align}
where the first inequality holds because we are only discarding negative contributions and the second by Corollary~\ref{Cor:HS-continuity}. 
\end{proof}
This we can specialize to some interesting divergences. We start with the relative entropy. 
\begin{corollary}\label{Cor:D-continuity}
For any quantum states $\rho,\tau,\sigma$, we have,  
    \begin{align}
        D(\rho\|\sigma) - D(\tau\|\sigma) &\leq \Omega(\rho\|\sigma) E_1(\rho\|\tau) \\
        &\leq  \log\left(\frac{\lambda_{\min}(\sigma)^{-1}}{\lambda_{\min}(\rho)} \right) E_1(\rho\|\tau), \label{Eq:D-continuity-2}
    \end{align}
    where $\Omega(\rho\|\sigma) := D_{\max}(\rho\|\sigma) + D_{\max}(\sigma\|\rho)$ is the Hilbert projective metric.
\end{corollary}
\begin{proof}
    Consider $f(x)=x\log(x)$, leading to $\frac{1}{s^2}f''\left(\frac1s\right)=f''(s)=\frac1s$, then
    \begin{align}
        \left( \int_1^{e^{D_{\max}(\rho\|\sigma)}} \frac1s  ds + \int_1^{e^{D_{\max}(\sigma\|\rho)}} \frac1s ds \right) = D_{\max}(\rho\|\sigma) + D_{\max}(\sigma\|\rho).
    \end{align}
    With Proposition~\ref{Prop:Cont-general} this gives the first inequality. The second follows from the usual upper bound on the $\max$-relative entropy. 
\end{proof}
These bounds can be compared to existing bounds, including those in~\cite{gour2020optimal,bluhm2023continuity}.  
Typically, continuity bounds for the relative entropy depend on the smallest eigenvalue of $\sigma$, $\lambda_{\min}(\sigma)$, as in our weaker second bound. However, while this is often necessary, one can get stronger bounds when the states in question are sufficiently similar. We note here that the Hilbert projective metric can also be expressed in the following equivalent ways, 
\begin{align}
\Omega(\rho\|\sigma) &:= D_{\max}(\rho\|\sigma) + D_{\max}(\sigma\|\rho) \\
&= \log \inf\left\{\frac{\lambda}{\mu} \,:\, \mu\sigma\leq\rho\leq\lambda\sigma \right\} \\
&= \log\sup\left\{ \frac{\tr A\sigma}{\tr B\sigma} : A,B\geq 0, \frac{\tr A\rho}{\tr B\rho}\leq 1 \right\},  
\end{align}
where the last equality is from~\cite{regula2022tight}. 
Our second, weaker, bound in Equation~\eqref{Eq:D-continuity-2} is the result of essentially ignoring the potential further similarity between the states. Before we move on to other quantities, we briefly state a consequence for the von-Neumann entropy. 
\begin{corollary}
    For quantum states $\rho,\sigma$, we have, 
    \begin{align}
        |S(\rho) - S(\sigma)| \leq \log\left(\max\left\{ \frac{\lambda_{\max}(\rho)}{\lambda_{\min}(\rho)},\frac{\lambda_{\max}(\sigma)}{\lambda_{\min}(\sigma)} \right\}\right) E_1(\rho\|\sigma). 
    \end{align}
\end{corollary}
\begin{proof}
    This follows from Corollary~\ref{Cor:D-continuity} by recalling $S(\rho)=-D(\rho\|\tau_d)$, where $\tau_d$ is the maximally mixed state of dimension $d$, and checking
    \begin{align}
        D_{\max}(\sigma\|\tau_d) &= \log(d \lambda_{\max}(\sigma)), \\
        D_{\max}(\tau_d\|\sigma) &= \log(\frac1{d \lambda_{\min}(\sigma)}), 
    \end{align}
    and similar for $\sigma$ exchanged by $\rho$. 
\end{proof}
This can be compared to the usual continuity bound for the von-Neumann entropy~\cite{audenaert2007sharp,petz2007quantum}. For either bound one can find states where that bound is tighter than the other. 

Now, similarly to our result for the relative entropy above, we can also derive continuity bounds for other divergences. We end this section with two more examples. 
\begin{corollary}
    For any quantum states $\rho,\tau,\sigma$, we have, 
    \begin{align}
    JS(\rho\|\sigma) - JS(\tau\|\sigma) &\leq \frac12 \left( D_{\max}(\rho\|\sigma) + \log\left(\frac{e^{D_{\max}(\sigma\|\rho)}+1}{e^{D_{\max}(\rho\|\sigma)}+1}\right) \right) E_1(\rho\|\tau).
\end{align} 
\end{corollary}
\begin{proof}
    This follows from Proposition~\ref{Prop:Cont-general} by considering $f(x)=\frac12(1+x)\log(\frac{2}{1+x}) + \frac12 x\log(x)$. 
\end{proof}
\begin{corollary}
    For any quantum states $\rho,\tau,\sigma$ and $\alpha \in (0, 1) \cup (1, \infty)$, we have,  
    \begin{align}
    \Hell_\alpha(\rho\|\sigma) - \Hell_\alpha(\tau\|\sigma) &\leq \frac{\alpha}{\alpha-1} \left( e^{(\alpha-1)D_{\max}(\rho\|\sigma)} - e^{(1-\alpha)D_{\max}(\sigma\|\rho)}  \right) E_1(\rho\|\tau).
\end{align} 
\end{corollary}
\begin{proof}
    This follows from Proposition~\ref{Prop:Cont-general} by considering $f(x)=\frac{x^\alpha-1}{\alpha-1}$. 
\end{proof}

\subsection{Quantum differential privacy}

Differential privacy is criteria that gives provable privacy guarantees for the outputs of noisy channels. Quantum differential privacy was first defined in~\cite{zhou2017differential} and has since been investigated in numerous publications. Here we give a definition based on an equivalence result in~\cite{hirche2022quantum}. Consider a set of quantum states $\cD$ equipped with some neighbouring relation, denoted $\rho\sim\sigma$. We say a quantum channel $\cA$ is $(\epsilon,\delta)$-differentially private, in short $(\epsilon,\delta)$-DP, if its output hockey-stick divergence is bounded as below: 
\begin{align}\label{Eq:Def-edDP}
    \cA \text{ is $(\epsilon,\delta)$-DP } \iff \,\sup_{\rho\sim \sigma} E_{e^\epsilon}(\cA(\rho)\|\cA(\sigma) \leq\delta,
\end{align} 
where the supremum is over all neighbouring quantum states in $\cD$. 
A related stronger criterion is that of local differential privacy (LDP), a quantum generalization of which was introduced in~\cite{hirche2022quantum}. In short, 
\begin{align}\label{Eq:LDP-def}
    \text{$\cA$ is $(\epsilon,\delta)$-LDP} \iff \sup_{\rho,\sigma} E_{e^\epsilon}(\cA(\rho) \| \cA(\sigma) ) \leq \delta\,.
\end{align}
Here the supremum is taken over all states. We remark that, as pointed out in~\cite{hirche2022quantum}, one could similarly define it with a supremum over all pure states, in which case all following arguments remain the same.  

In this section we briefly discuss two applications of the integral representations in our work under differentially private quantum channels. 

\subsubsection{Hypothesis testing under local differential privacy} 

For convenience let us define the set
\begin{align}
    \cL_{\epsilon,\delta} := \{ \cA \,|\, \text{$\cA$ is $(\epsilon,\delta)$-LDP} \}. 
\end{align}
We begin by giving a input dependent bound on the $f$-divergence of a locally private channel. 
\begin{corollary}
Let $f\in\cF$ and $\varphi(\epsilon,\delta)=1-e^{-\epsilon}(1-\delta)$. Then 
    \begin{align}
        \sup_{\cA\in\cL_{\epsilon,\delta}} D_f(\cA(\rho) \| \cA(\sigma)) \leq \varphi(\epsilon,\delta) D_f(\rho \| \sigma), 
    \end{align}
    where the supremum is over all quantum channels $\cA$ that are $(\epsilon,\delta)$-LDP. 
\end{corollary}
\begin{proof}
    The proof follows directly from combining our earlier Lemma~\ref{Lem:ConBoundTrace} with~\cite[Corollary V.1]{hirche2022quantum} which states that
    \begin{align}
        \eta_{\tr}(\cA) \leq \varphi(\epsilon,\delta)
    \end{align}
    for all $\cA$ that are $(\epsilon,\delta)$-LDP. 
\end{proof}
This includes in particular the relative entropy as a special case, 
    \begin{align}
        \sup_{\cA\in\cL_{\epsilon,\delta}} D(\cA(\rho) \| \cA(\sigma)) \leq \varphi(\epsilon,\delta) D(\rho \| \sigma). \label{Eq:RelEnt-LDP-bound} 
    \end{align}
In hypothesis testing we are concerned with determining which of a set of possible options is the true one. Here we will in particular look at asymmetric state discrimination. Given $n$ samples of either a state $\rho$ or a state $\sigma$ our task is to determine which one is present. Attempting to do so by the means of measuring by a POVM $\{M_i\}_{i=0,1}$ leads to two possible errors,
\begin{align}
\alpha_n = \tr (\Id-M_0)\rho^{\otimes n} \quad\text{and}\quad \beta_n=\tr M_0\sigma^{\otimes n},
\end{align}
called the Type 1 and Type 2 error, respectively. In the asymmetric setting, we are interested in the minimum Type 2 error given that the Type 1 error is bounded by a constant $0<\epsilon<1$, meaning the quantity
\begin{align}
    \beta_{n,\epsilon}(\rho,\sigma) := \min\{ \beta_n \;|\; 0\leq M_0 \leq \Id,\, \alpha_n\leq\epsilon \}. 
\end{align}
One can easily show that this quantity decays exponentially in $n$ and the optimal rate at which it does is given by the quantum Stein's Lemma~\cite{hiai1991proper,ogawa2000strong} which says that
\begin{align}
    \lim_{n\rightarrow\infty} -\frac{1}{n} \log \beta_{n,\epsilon}(\rho,\sigma) = D(\rho \| \sigma). 
\end{align}
The result in this work now allow for a converse bound on the achievable rate under local differential privacy. In particular this result is independent of the exact structure of the quantum channel, besides that it fulfils the LDP criterion. 
\begin{corollary}
    For any $(\epsilon,\delta)$-LDP quantum channel $\cA$, we have 
 \begin{align}
    \lim_{n\rightarrow\infty} -\frac{1}{n} \log \beta_{n,\epsilon}(\cA(\rho),\cA(\sigma)) \leq \varphi(\epsilon,\delta) D(\rho \| \sigma). 
\end{align}   
\end{corollary}
\begin{proof}
    This is direct consequence of combining the quantum Stein's Lemma with Equation~\eqref{Eq:RelEnt-LDP-bound}. 
\end{proof}
This implies that if states $\rho$ and $\sigma$ are subject to locally private noise, we need to scale $n$ by at least $\varphi(\epsilon,\delta)$ to reach the same error exponent as in the noiseless setting. The above generalizes the analogous classical result first shown in~\cite{asoodeh2021local}.

\subsubsection{Inequalities between divergences under $\epsilon$-differential privacy}

In this section we give some examples to illustrate how the inequalities between different divergences we gave in the previous sections can sometimes be very tight. As the setting choice we consider $\epsilon$-differentially private channels $\cA$. In this case the defining condition on the hockey-stick divergence in Equation~\eqref{Eq:Def-edDP} can be similarly expressed in terms of the $\max$-relative entropy~\cite{hirche2022quantum}, 
\begin{align}\label{Eq:Def-eDP}
    \cA \text{ is $\epsilon$-DP } &\iff \,\sup_{\rho\sim \sigma} E_{e^\epsilon}(\cA(\rho)\|\cA(\sigma)) =0 \\
    &\iff \,\sup_{\rho\sim \sigma} D_{\max}(\cA(\rho)\|\cA(\sigma)) \leq \epsilon. 
\end{align} 
In the following we will assume that the neighbouring condition is symmetric, meaning that $\rho\sim\sigma$ if and only if $\sigma\sim\rho$. This is the case for all definitions commonly used in the literature. 
We start by considering a reverse Pinsker-type inequality for the relative entropy under differentially private noise. 
\begin{corollary}\label{Cor:q-rel-pinsker-DP}
Let $\cA$ be $\epsilon$-differentially private and $\rho\sim\sigma$, then
\begin{align}
    D(\cA(\rho)\|\cA(\sigma)) \leq  \epsilon\, E_1(\cA(\rho)\|\cA(\sigma)). 
\end{align}
\end{corollary}
\begin{proof}
    The proof follows directly from Equation~\eqref{Eq:NewRevPin-Thompson}, more precisely by bounding $\Xi(\cA(\rho)\|\cA(\sigma)\leq\epsilon$.  
\end{proof}
In the same way also all the other reverse Pinsker-type inequalities for different divergences in Section~\ref{Sec:RevPinskerIneq} can be specialized to differentially private noise. 
We can also apply the same idea to the continuity bounds in Section~\ref{Sec:ContinuityFirst}, giving us results akin to the following. 
\begin{corollary}\label{Cor:D-continuity-DP}
Let $\cA$ be $\epsilon$-differentially private, $\tau$ some quantum state and $\rho\sim\sigma$, then  
    \begin{align}
        D(\cA(\rho)\|\cA(\sigma)) - D(\tau\|\cA(\sigma)) &\leq 2\epsilon E_1(\cA(\rho)\|\tau).
    \end{align} 
\end{corollary}
\begin{proof}
    This follows directly from Corollary~\ref{Cor:D-continuity}. 
\end{proof}
We leave it to the reader to apply this idea to their other favourite inequalities from the previous sections.

\subsection{Bounds on amortized channel divergences}

Channel divergences take the role of the usual divergences when we are interested in distance measures between channels. While they are defined generally for arbitrary quantum state divergences, we consider here primarily versions based on the new $f$-divergences. The most common way of defining a channel divergence is as
\begin{align}
    D_f(\cN\|\cM) := \sup_{\rho} D_f(\cN\otimes\id(\rho)\|\cM\otimes\id(\rho)),
\end{align}
where the supremum is over all bipartitie states $\rho$ with arbitrary sized reference system. Note that this is also often called a stabilized  channel divergence, because of the inclusion of an additional reference system. Analogous, one can define non-stabilized versions of this and all the following definitions. This does not substantially change the results presented here and we omit that setting for brevity.  

An alternative way of defining a channel divergence are the so-called amortized channel divergences that have recently come up in the context of channel discrimination~\cite{berta2019stein,wilde2020amortized}. For the case of the new $f$-divergence,  
\begin{align}
    D_f^A(\cN\|\cM) := \sup_{\rho,\sigma}D_f(\cN\otimes\id(\rho)\|\cM\otimes\id(\sigma))-D_f(\rho\|\sigma).
\end{align}
While for the standard channel divergence one can often give bounds on the size of the needed reference system, this is much more complicated for the amortized channel divergence. In fact, even for common divergences like the relative entropy no such dimension bound is known when amortization is used. 
One sees directly that
\begin{align}
    D_f(\cN\|\cM) \leq D_f^A(\cN\|\cM),
\end{align}
however finding upper bounds on the amortized channel divergences can often be challenging. First, note that one can easily get bounds on $D_f(\cN\|\cM)$ in terms of $E_1(\cN\|\cM)$ by extending our earlier reverse Pinsker-type inequalities to quantum channels. 
We remark that $E_1(\cN\|\cM)$ is the well known diamond norm that plays a fundamental role in quantum channel discrimination, i.e.
\begin{align}
    E_1(\cN\|\cM) = \sup_{\rho} \|\cN\otimes\id(\rho) - \cM\otimes\id(\rho)\|_{\tr} =: \|\cN-\cM\|_{\diamond}. 
\end{align} 
In the following we will see how the earlier results can be extended to also bound amortized channel divergences. First note that it was already shown in~\cite{wilde2020amortized} that amortization does not increase the trace distance, meaning, in our notation, 
\begin{align}
    E^{A}_1(\cN\|\cM) = E_1(\cN\|\cM). 
\end{align}
The first observation is the following. 
\begin{lemma}
    For quantum channels $\cN,\cM$ and $\gamma\geq 1$, we have, 
    \begin{align}
        E_\gamma(\cN(\rho)\|\cM(\sigma)) &\leq E_1(\cN(\rho)\|\cM(\rho)) + E_\gamma(\rho\|\sigma), \label{Eq:DataProcTriangle}
    \end{align}
    and as a consequence
    \begin{align}
        E^A_\gamma(\cN\|\cM) \leq E_1(\cN\|\cM).
    \end{align}
\end{lemma}
\begin{proof}
    The first statement follows from the triangle inequality for the hockey-stick divergence and data processing, 
    \begin{align}
        E_\gamma(\cN(\rho)\|\cM(\sigma)) &\leq E_1(\cN(\rho)\|\cM(\rho))+ E_\gamma(\cM(\rho)\|\cM(\sigma)) \\
        &\leq E_1(\cN(\rho)\|\cM(\rho))+ E_\gamma(\rho\|\sigma). 
    \end{align}
    The second statement follows by the definition of amortization. 
\end{proof}
We can use this to make the following general statement. 
\begin{corollary}
    For any quantum channels $\cN,\cM$ and  $f\in\cF$, we have, 
    \begin{align}
    D^{A}_f(\cN\|\cM) &\leq E_1(\cN\|\cM) \, \xi(f), 
    \end{align}
    with 
    \begin{align}
        \xi(f) = \int_1^{\infty} f''(\gamma) + \gamma^{-3} f''(\gamma^{-1}) \d\gamma
    \end{align}
\end{corollary}
\begin{proof}
    The claim is a direct consequence of Equation~\ref{Eq:DataProcTriangle} and the definition of $D^{A}_f(\cN\|\cM)$. This is similar to the proof of Proposition~\ref{Prop:Rev-Q-Pinsker}.
\end{proof}
We point out that while on the left hand side we generally do not have a bound on the needed size for the reference system, on the right hand side the optimization over states can be restricted to pure states and hence the dimension of the reference system can be bounded. 

Based on the corollary one simply needs to check for functions for which the integral above is finite. A potentially interesting example is given, for $0<\alpha< 1$, by 
\begin{align}
    \Hell^A_\alpha(\cN\|\cM) \leq \frac{1}{1-\alpha} E_1(\cN\|\cM).
\end{align}
Recall the close relationship between this divergence and the Petz \Renyi divergence. However, whether this could be helpful in finding bounds on the amortized channel Petz \Renyi divergence remains unclear. Similarly, we can use the same approach to give a bound on the armortized Jensen Shannon divergence, 
\begin{align}
    JS^A(\cN\|\cM) \leq \log{(2)}\, E_1(\cN\|\cM).
\end{align}

\section{Conclusions} 

In this work we gave a new definition of a quantum $f$-divergence, inspired by an integral representation of the classical $f$-divergence in terms of hockey-stick divergences. Perhaps surprisingly, the new $f$-divergence has a number of desirable properties. In particular, the standard Umegaki relative entropy is a special case of our definition for $f(x)=x\log x$ and the standard construction of R\'enyi divergences using $f_\alpha(x)=\frac{x^\alpha-1}{\alpha-1}$ together with regularization yields the Petz \Renyi divergence for $0\leq\alpha\leq1$ and the sandwiched \Renyi divergence for $\alpha\geq 1$. Especially the latter observation is remarkable as it unifies the two different quantum \Renyi divergences for the respective parameter ranges where they are commonly found useful. 

Additionally, the new $f$-divergences mirror the behaviour of their classical counterpart when it comes to SDPI constants and contraction coefficients in many regards. Previous definitions of $f$-divergences are not known to have many of these properties, despite significant effort to prove them. Finally, we were able to exploit the new integral representations to prove a plethora of new inequalities between different divergences that we hope are of independent interest.

Thus, at least in some ways, our initially naive way of defining a new quantum $f$-divergence has proven to be a more natural generalization of the classical quantity than those previously considered and we expect it will find additional useful applications in the future.

Nevertheless, there are many open problems left to investigate.  In particular, any results that hold for classical $f$-divergences but have so far avoided generalization to the quantum setting are a natural starting point. Two particular open problems of a more technical nature we would like to mention are the following. First, from our discussion it appears that the relative entropy plays a somewhat special role. We get a closed additive expression directly from the integral representation, while the new $D_\alpha$ is not additive for non-commuting states. Finding more functions for which the new $f$-divergence has a closed form expression would be interesting. Also a direct proof of the additivity of the relative entropy solely based on the integral representation might be enlightening. 
The second open problem is concerned with hypothesis testing. Clearly our proof techniques in relating the new $f$-divergence to the Petz and sandwiched \Renyi divergences are closely connected to the study of error and strong converse exponents in asymptotic binary quantum hypothesis testing. Could one give a proof of the error and strong converse exponents directly by using $D_\alpha$ and have the usual \Renyi divergences emerge from its regularization? This would give an operational significance to the new quantum \Renyi divergence.

\section*{Acknowledgments} 
We would like to thank Salman Beigi for helpful discussions. CH also thanks Daniel Stilck Fran\c ca for insightful exchanges. 
This project has received funding from the European Union's Horizon 2020 research and innovation programme under the Marie Sklodowska-Curie Grant Agreement No. H2020-MSCA-IF-2020-101025848. MT and CH are supported by the National Research Foundation, Singapore and A*STAR under its CQT
Bridging Grant. MT is also supported by the National Research Foundation, Singapore and A*STAR under its Quantum Engineering Programme (NRF2021-QEP2-01-P06).

\appendix

\section{Appendix: Direct proof of classical integral representation}\label{App:Direct-Proof}

As described in the main text,~\cite{sason2016f} gives a an integral representation of $f$-divergences with $f\in\cF$ based on an earlier representation in~\cite{liese2006divergences}. The latter was originally shown without the twice differentiable condition that is part of $\cF$ and has therefore a more complicated form. 
The goal is to give a simple self-contained argument for the readers convenience. Here we want to briefly show that the exact argument from~\cite{liese2006divergences} leads to the integral representation in~\cite{sason2016f} when specialized to twice differentiable functions. The special case of the formulation in~\cite{liese2006divergences} can then be found by appropriate substitution, essentially reversing the proof in~\cite{sason2016f}. 

The main ingredient is the Taylor formula for a twice continuously differentiable function $f$
\begin{align}
    f(b) = f(a) + f'(a)(b-a) + \int_a^b (b-s) f''(s) \,\d s. 
\end{align}
We write the $f$-divergence as
\begin{align}
    D_f(P,Q) &= \int q f\left(\frac{p}{q}\right) \d\mu \\
    &= \int q \int_1^\frac{p}{q} \left(\frac{p}{q}-s\right) f''(s)\, \d s\, \d\mu \\
    &= \int  \int_1^\frac{p}{q} (p-sq) f''(s)\, \d s\, \d\mu \\
    &= \int \left[ \int_1^\infty (p-sq)_+ f''(s)\, \d s + \int_0^1 (sq-p)_+ f''(s) ds \right] \, \d\mu \\
    &= \int_1^\infty E_s(P\|Q) f''(s)\, \d s + \int_0^1 s E_{\frac{1}{s}}(Q\|P) f''(s) \,\d s \\
    &= \int_1^\infty \left[ E_s(P\|Q) f''(s) + \frac{1}{s^3} E_{s}(Q\|P) f''\left(\frac{1}{s}\right) \right] \d s,
\end{align}
which is the integral representation given in~\cite{sason2016f}. In the above the main work happens in the third equality where we split by cases $p < q$ and $p>q$, noting that $(p-sq)_+=0$ for all $s\geq \frac{p}{q}$ and similarly $(sq-p)_+=0$ for all $s\leq \frac{p}{q}$. The forth equality is by definition of the hockey-stick divergence. For the final equality we substitute $t=\frac1s$ in the second integral.

\section{Appendix: Scaling of $D_f(\rho\|\sigma)$} \label{App:Scaling}

Throughout the discussion in the main text we have considered the $f$-divergence $D_f(\rho\|\sigma)$ with two quantum states as inputs. In particular, this restricts the inputs to be normalized as $\tr\rho=\tr\sigma=1$. In this section we discuss dropping this restriction and consider two, possibly non-normalized, positive semi-definite hermitian operators $A$ and $B$. In this case we propose to generalize the definition of the new $f$-divergence as follows, 
\begin{align}
    D_f(A\| B) = \tr(A-B) + \int_1^\infty \left(f''(\gamma) \widetilde E_\gamma(A\|B) + \gamma^{-3} f''(\gamma^{-1}) \widetilde E_\gamma(B\|A)\right) \d\gamma,
\end{align}
where $\widetilde E_\gamma(A\|B)=E_\gamma(A\|B)+(\tr(A-\gamma B))_+$. Note that this substitution is necessary because of how we defined the Hockey-Stick divergence in Equation~\eqref{Eq:Def-HockeyStick}. Dropping the second term there would allow us here to use $E_\gamma$ directly, but would also lead to a more complicated symmetry relation for the Hockey-Stick divergence. 
That this definition captures correctly the expected behaviour of the relative entropy is a direct consequence of the result in~\cite{frenkel2022integral}. 

A common applications is the scaling of the inputs. Meaning for positive constants $a$, $b$ we are interested in the behaviour of $D_f(a\,\rho\| b\,\sigma)$. This also serves as a sanity check for the proposed generalization. 
In the following we give a convenient alternative expression of this divergence. 
\begin{proposition}
   For $a,b>0$, $f\in\cF$ and $\rho, \sigma$ quantum states, we have
   \begin{align}
       D_f(a\,\rho\| b\,\sigma) = (a-b) + b\, D_{F_{a,b}}(\rho\|\sigma) + b \int_{\frac{b}{a}}^1 F_{a,b}''\left(\hat\gamma\right)  (1-\gamma) d\hat\gamma , 
   \end{align}
   where $F_{a,b}(x)=f(\frac{a}{b}x)$. 
\end{proposition}
\begin{proof}
First consider, 
\begin{align}
     &\int_1^\infty \left(f''(\gamma) E_\gamma(a\,\rho\|b\,\sigma) + \gamma^{-3} f''(\gamma^{-1})E_\gamma(b\,\sigma\|a\,\rho)\right) \d\gamma \\
     &= \int_1^\infty \left(f''(\gamma) a E_{\frac{b}{a}\gamma}(\rho\|\sigma) + \gamma^{-3} f''(\gamma^{-1})b E_{\frac{a}{b}\gamma}(\sigma\|\rho)\right) \d\gamma \\
     &= \int_{\frac{b}{a}}^\infty f''(\frac{a}{b}\hat\gamma) \frac{a^2}{b} E_{\hat\gamma}(\rho\|\sigma) d\hat\gamma + \int_{\frac{a}{b}}^\infty \frac{a^2}{b} \hat\gamma^{-3} f''(\frac{a}{b}\hat\gamma^{-1}) E_{\hat\gamma}(\sigma\|\rho) d\hat\gamma \\
     &=  \frac{a^2}{b} \int_{1}^\infty \left( f''(\frac{a}{b}\hat\gamma) E_{\hat\gamma}(\rho\|\sigma)  +  \hat\gamma^{-3} f''(\frac{a}{b}\hat\gamma^{-1}) E_{\hat\gamma}(\sigma\|\rho) \right) d\hat\gamma \nonumber\\
     &\qquad\quad + \int_{\frac{b}{a}}^1 f''(\frac{a}{b}\hat\gamma) \frac{a^2}{b} E_{\hat\gamma}(\rho\|\sigma) d\hat\gamma + \int_{\frac{a}{b}}^1 \frac{a^2}{b} \hat\gamma^{-3} f''(\frac{a}{b}\hat\gamma^{-1}) E_{\hat\gamma}(\sigma\|\rho) d\hat\gamma \\
     &= b D_{F_{a,b}}(\rho\|\sigma) + \int_{\frac{b}{a}}^1 f''(\frac{a}{b}\hat\gamma) \frac{a^2}{b} E_{\hat\gamma}(\rho\|\sigma) d\hat\gamma - \int_{\frac{b}{a}}^1 \frac{a^2}{b} \gamma f''(\frac{a}{b}\gamma) E_{\gamma^{-1}}(\sigma\|\rho) \d\gamma     \\
     &=  b D_{F_{a,b}}(\rho\|\sigma) 
\end{align}
where the first equality is by definition, the second because $a,b>0$, the third by substituting $\hat\gamma=\frac{b}{a}\gamma$ and $\hat\gamma=\frac{a}{b}\gamma$ in the first and second integral, respectively. The forth is by shifting the integral boarders. In the fifth equality we have defined $F_{a,b}(x)=f(\frac{a}{b}x)$ and substituted $\gamma=\hat\gamma^{-1}$ in the final integral. For the sixth equality we have used $E_\gamma(\rho\|\sigma) = \gamma E_{\frac1{\gamma}}(\sigma\|\rho)$, which was shown in Lemma~\ref{Lem:HS-Properties}, see also~\cite{hirche2022quantum}. 
Next, consider, 
\begin{align}
     &\int_1^\infty \left(f''(\gamma) (\tr(a\rho-\gamma b\sigma))_+ + \gamma^{-3} f''(\gamma^{-1})(\tr(b\sigma-\gamma a\rho))_+\right) \d\gamma \\
     &= \int_1^\infty \left(f''(\gamma) (a-\gamma b)_+ + \gamma^{-3} f''(\gamma^{-1})(b-\gamma a)_+\right) \d\gamma \\
     &= \int_1^\infty f''(\gamma) (a-\gamma b)_+ + \int_0^1 f''(\gamma)(\gamma b- a)_+ \d\gamma.
\end{align}
Now if $a\leq b$ then the first term is always zero, and if $a\geq b$ then the second term is always zero. We continue with the second case and the first one works similarly, 
\begin{align}
    \int_1^\infty f''(\gamma) (a-\gamma b)_+ = \int_1^\frac{a}{b} f''(\gamma) (a-\gamma b)
    = \frac{a^2}{b} \int_\frac{b}{a}^1 f''(\frac{a}{b}\gamma) (1-\gamma ) 
    = b \int_\frac{b}{a}^1 F_{a,b}''(\gamma) (1-\gamma ).
\end{align}
Putting all terms together gives the claimed result. 
\end{proof}
To provide some examples we give the following corollary. 
\begin{corollary}
    For $a,b>0$ and $\rho, \sigma$ quantum states, we have
    \begin{align}
        D(a\,\rho\| b\,\sigma) &= a D(\rho\|\sigma) + a\log\frac{a}{b}, \\
        \Hell_\alpha(a\,\rho\| b\,\sigma) &= \frac1{\alpha-1}[a^\alpha b^{1-\alpha}-a]+a^\alpha b^{1-\alpha} \Hell_\alpha(\rho\| \sigma). 
    \end{align}
\end{corollary}
\begin{proof}
    In the first case we have $f(x)=x\log x$ and therefore $F''_{a,b}(x)=\frac{a}{b}f''(x)$ from which the claim can easily be checked. 
    Similarly, in the second case we have $f_\alpha(x)=\frac{x^\alpha-1}{\alpha-1}$ and hence $F''_{a,b}(x)= \left(\frac{a}{b}\right)^\alpha f''(x)$ from which the claim follows after some algebra. 
\end{proof}
Defining the generalized quantities
\begin{align}
    Q_\alpha(A\| B) &= \tr[A] + (\alpha-1)\Hell(A\| B) \qquad \textnormal{and} \\
    D_\alpha(A\| B) &= \frac1{\alpha-1}\log\frac{Q_\alpha(A\| B)}{\tr A}, 
\end{align}
we get the usual scaling of $Q_\alpha(a\rho\| b\sigma) = a^\alpha b^{1-\alpha} Q_\alpha(\rho\| \sigma)$ and $D_\alpha(a\rho\| b\sigma) = D_\alpha(\rho\| \sigma) + \log\frac{a}{b}$.

\bibliographystyle{ultimate}
\bibliography{lib}

\begin{thebibliography}{10}

\bibitem{asoodeh2021local}
S.~Asoodeh, M.~Aliakbarpour, and F.~P. Calmon.
\newblock {\em ``Local Differential Privacy Is Equivalent to Contraction of an
  $ f $-Divergence''}.
\newblock In {\em 2021 IEEE International Symposium on Information Theory
  (ISIT)}, pages 545--550, (2021).

\bibitem{audenaert2007sharp}
K.~M. Audenaert.
\newblock {\em ``A sharp continuity estimate for the von Neumann entropy''}.
\newblock Journal of Physics A: Mathematical and Theoretical {\bf
  40}(28):\,8127, (2007).

\bibitem{audenaert2012comparisons}
K.~M. Audenaert.
\newblock {\em ``Comparisons between quantum state distinguishability
  measures''}.
\newblock Preprint, \href{http://arxiv.org/abs/1207.1197}{arXiv:\,1207.1197}
  (2012).

\bibitem{audenaert2014quantum}
K.~M. Audenaert.
\newblock {\em ``Quantum skew divergence''}.
\newblock Journal of Mathematical Physics {\bf 55}(11):\,112202, (2014).

\bibitem{audenaert2007discriminating}
K.~M. Audenaert, J.~Calsamiglia, R.~Munoz-Tapia, E.~Bagan, L.~Masanes, A.~Acin,
  and F.~Verstraete.
\newblock {\em ``Discriminating states: The quantum Chernoff bound''}.
\newblock Physical review letters {\bf 98}(16):\,160501, (2007).

\bibitem{audenaert2005continuity}
K.~M. Audenaert and J.~Eisert.
\newblock {\em ``Continuity bounds on the quantum relative entropy''}.
\newblock Journal of mathematical physics {\bf 46}(10):\,102104, (2005).

\bibitem{audenaert2011continuity}
K.~M. Audenaert and J.~Eisert.
\newblock {\em ``Continuity bounds on the quantum relative entropy—II''}.
\newblock Journal of Mathematical Physics {\bf 52}(11):\,112201, (2011).

\bibitem{audenaert2008asymptotic}
K.~M. Audenaert, M.~Nussbaum, A.~Szko{\l}a, and F.~Verstraete.
\newblock {\em ``Asymptotic error rates in quantum hypothesis testing''}.
\newblock Communications in Mathematical Physics {\bf 279}:\,251--283, (2008).

\bibitem{berta2019stein}
M.~Berta, C.~Hirche, E.~Kaur, and M.~M. Wilde.
\newblock {\em ``Stein’s lemma for classical-quantum channels''}.
\newblock In {\em 2019 IEEE International Symposium on Information Theory
  (ISIT)}, pages 2564--2568, (2019).

\bibitem{binette2019note}
O.~Binette.
\newblock {\em ``A note on reverse Pinsker inequalities''}.
\newblock IEEE Transactions on Information Theory {\bf 65}(7):\,4094--4096,
  (2019).

\bibitem{bluhm2023continuity}
A.~Bluhm, {\'A}.~Capel, P.~Gondolf, and A.~P{\'e}rez-Hern{\'a}ndez.
\newblock {\em ``Continuity of quantum entropic quantities via almost
  convexity''}.
\newblock IEEE Transactions on Information Theory , (2023).

\bibitem{briet2009properties}
J.~Bri{\"e}t and P.~Harremo{\"e}s.
\newblock {\em ``Properties of classical and quantum Jensen-Shannon
  divergence''}.
\newblock Physical review A {\bf 79}(5):\,052311, (2009).

\bibitem{choi1994equivalence}
M.-D. Choi, M.~B. Ruskai, and E.~Seneta.
\newblock {\em ``Equivalence of certain entropy contraction coefficients''}.
\newblock Linear algebra and its applications {\bf 208}:\,29--36, (1994).

\bibitem{cohen1998comparisons}
J.~Cohen, J.~H. Kempermann, and G.~Zbaganu.
\newblock {\em Comparisons of stochastic matrices with applications in
  information theory, statistics, economics and population}.
\newblock Springer Science \& Business Media (1998).

\bibitem{cohen1993relative}
J.~E. Cohen, Y.~Iwasa, G.~Rautu, M.~B. Ruskai, E.~Seneta, and G.~Zbaganu.
\newblock {\em ``Relative entropy under mappings by stochastic matrices''}.
\newblock Linear algebra and its applications {\bf 179}:\,211--235, (1993).

\bibitem{csiszar63}
I.~Csisz{\'{a}}r.
\newblock {\em ``{Eine informationstheoretische Ungleichung und ihre Anwendung
  auf den Beweis der Ergodizitat von Markoffschen Ketten}''}.
\newblock Magyar. Tud. Akad. Mat. Kutato Int. Kozl {\bf 8}:\,85--108, (1963).

\bibitem{csiszar2011information}
I.~Csisz{\'a}r and J.~K{\"o}rner.
\newblock {\em Information theory: coding theorems for discrete memoryless
  systems}.
\newblock Cambridge University Press (2011).

\bibitem{csiszar2006context}
I.~Csisz{\'a}r and Z.~Talata.
\newblock {\em ``Context tree estimation for not necessarily finite memory
  processes, via BIC and MDL''}.
\newblock IEEE Transactions on Information theory {\bf 52}(3):\,1007--1016,
  (2006).

\bibitem{degroot1962uncertainty}
M.~H. DeGroot.
\newblock {\em ``Uncertainty, information, and sequential experiments''}.
\newblock The Annals of Mathematical Statistics {\bf 33}(2):\,404--419, (1962).

\bibitem{frenkel2022integral}
P.~E. Frenkel.
\newblock {\em ``Integral formula for quantum relative entropy implies data
  processing inequality''}.
\newblock Preprint, \href{http://arxiv.org/abs/2208.12194}{arXiv:\,2208.12194}
  (2022).

\bibitem{gour2020optimal}
G.~Gour and M.~Tomamichel.
\newblock {\em ``Optimal extensions of resource measures and their
  applications''}.
\newblock Physical Review A {\bf 102}(6):\,062401, (2020).

\bibitem{gyorfi2001class}
L.~Gy{\"o}rfi and I.~Vajda.
\newblock {\em ``A class of modified Pearson and Neyman statistics''}.
\newblock Statistics \& Risk Modeling {\bf 19}(3):\,239--252, (2001).

\bibitem{hiai2011quantum}
F.~Hiai, M.~Mosonyi, D.~Petz, and C.~B{\'e}ny.
\newblock {\em ``Quantum f-divergences and error correction''}.
\newblock Reviews in Mathematical Physics {\bf 23}(07):\,691--747, (2011).

\bibitem{hiai1991proper}
F.~Hiai and D.~Petz.
\newblock {\em ``The proper formula for relative entropy and its asymptotics in
  quantum probability''}.
\newblock Communications in mathematical physics {\bf 143}:\,99--114, (1991).

\bibitem{hiai2016contraction}
F.~Hiai and M.~B. Ruskai.
\newblock {\em ``Contraction coefficients for noisy quantum channels''}.
\newblock Journal of Mathematical Physics {\bf 57}(1):\,015211, (2016).

\bibitem{hirche2022bounding}
C.~Hirche and F.~Leditzky.
\newblock {\em ``Bounding Quantum Capacities via Partial Orders and
  Complementarity''}.
\newblock IEEE Transactions on Information Theory {\bf 69}(1):\,283--297,
  (2022).

\bibitem{hirche2022contraction}
C.~Hirche, C.~Rouz{\'e}, and D.~S. Fran{\c{c}}a.
\newblock {\em ``On contraction coefficients, partial orders and approximation
  of capacities for quantum channels''}.
\newblock Quantum {\bf 6}:\,862, (2022).

\bibitem{hirche2022quantum}
C.~Hirche, C.~Rouz{\'e}, and D.~S. Fran{\c{c}}a.
\newblock {\em ``Quantum differential privacy: An information theory
  perspective''}.
\newblock Preprint, \href{http://arxiv.org/abs/2202.10717}{arXiv:\,2202.10717}
  (2022).

\bibitem{jenvcova2023recoverability}
A.~Jen{\v{c}}ov{\'a}.
\newblock {\em ``Recoverability of quantum channels via hypothesis testing''}.
\newblock Preprint, \href{http://arxiv.org/abs/2303.11707}{arXiv:\,2303.11707}
  (2023).

\bibitem{kastoryano2013quantum}
M.~J. Kastoryano and K.~Temme.
\newblock {\em ``Quantum logarithmic Sobolev inequalities and rapid mixing''}.
\newblock Journal of Mathematical Physics {\bf 54}(5):\,052202, (2013).

\bibitem{le2012asymptotic}
L.~Le~Cam.
\newblock {\em Asymptotic methods in statistical decision theory}.
\newblock Springer Science \& Business Media (2012).

\bibitem{lee1999measures}
L.~Lee.
\newblock {\em ``Measures of distributional similarity''}.
\newblock In {\em Proceedings of the 37th annual meeting of the Association for
  Computational Linguistics on Computational Linguistics}, pages 25--32,
  (1999).

\bibitem{lesniewski1999monotone}
A.~Lesniewski and M.~B. Ruskai.
\newblock {\em ``Monotone Riemannian metrics and relative entropy on
  noncommutative probability spaces''}.
\newblock Journal of Mathematical Physics {\bf 40}(11):\,5702--5724, (1999).

\bibitem{liyao23}
K.~Li, Y.~Yao, and M.~Hayashi.
\newblock {\em ``{Tight Exponential Analysis for Smoothing the Max-Relative
  Entropy and for Quantum Privacy Amplification}''}.
\newblock \href{http://dx.doi.org/10.1109/TIT.2022.3217671}{IEEE Transactions
  on Information Theory {\bf 69}(3):\,1680--1694} (2023).

\bibitem{liese1987convex}
F.~Liese and I.~Vajda.
\newblock {\em Convex Statistical Distances}.
\newblock Teubner (1987).

\bibitem{liese2006divergences}
F.~Liese and I.~Vajda.
\newblock {\em ``On divergences and informations in statistics and information
  theory''}.
\newblock IEEE Transactions on Information Theory {\bf 52}(10):\,4394--4412,
  (2006).

\bibitem{lintomamichel14}
S.~M. Lin and M.~Tomamichel.
\newblock {\em ``{Investigating Properties of a Family of Quantum R{\'{e}}nyi
  Divergences}''}.
\newblock \href{http://dx.doi.org/10.1007/s11128-015-0935-y}{Quantum
  Information Processing {\bf 14}(4):\,1501--1512} (2015).

\bibitem{makur2015linear}
A.~Makur and L.~Zheng.
\newblock {\em ``Linear Bounds between Contraction Coefficients for $ f
  $-Divergences''}.
\newblock Preprint, \href{http://arxiv.org/abs/1510.01844}{arXiv:\,1510.01844}
  (2015).

\bibitem{matsumoto2018new}
K.~Matsumoto.
\newblock {\em ``A new quantum version of f-divergence''}.
\newblock In {\em Reality and Measurement in Algebraic Quantum Theory: NWW
  2015, Nagoya, Japan, March 9-13}, pages 229--273. Springer (2018).

\bibitem{mosonyi2015quantum}
M.~Mosonyi and T.~Ogawa.
\newblock {\em ``Quantum hypothesis testing and the operational interpretation
  of the quantum R{\'e}nyi relative entropies''}.
\newblock Communications in Mathematical Physics {\bf 334}:\,1617--1648,
  (2015).

\bibitem{muller2018sandwiched}
A.~M{\"u}ller-Hermes and D.~S. Franca.
\newblock {\em ``Sandwiched R{\'e}nyi convergence for quantum evolutions''}.
\newblock Quantum {\bf 2}:\,55, (2018).

\bibitem{muller2016relative}
A.~M{\"u}ller-Hermes, D.~Stilck~Fran{\c{c}}a, and M.~M. Wolf.
\newblock {\em ``Relative entropy convergence for depolarizing channels''}.
\newblock Journal of Mathematical Physics {\bf 57}(2):\,022202, (2016).

\bibitem{muller2013quantum}
M.~M{\"u}ller-Lennert, F.~Dupuis, O.~Szehr, S.~Fehr, and M.~Tomamichel.
\newblock {\em ``On quantum R{\'e}nyi entropies: A new generalization and some
  properties''}.
\newblock Journal of Mathematical Physics {\bf 54}(12):\,122203, (2013).

\bibitem{nussbaum2009chernoff}
M.~Nussbaum and A.~Szko{\l}a.
\newblock {\em ``The Chernoff Lower Bound for Symmetric Quantum Hypothesis
  Testing''}.
\newblock The Annals of Statistics pages 1040--1057, (2009).

\bibitem{ogawa2000strong}
T.~Ogawa and H.~Nagaoka.
\newblock {\em ``Strong converse and Stein's lemma in quantum hypothesis
  testing''}.
\newblock IEEE Transactions on Information Theory {\bf 46}(7):\,2428--2433,
  (2000).

\bibitem{petz1985quasi}
D.~Petz.
\newblock {\em ``Quasi-entropies for states of a von Neumann algebra''}.
\newblock Publications of the Research Institute for Mathematical Sciences {\bf
  21}(4):\,787--800, (1985).

\bibitem{petz1986quasi}
D.~Petz.
\newblock {\em ``Quasi-entropies for finite quantum systems''}.
\newblock Reports on mathematical physics {\bf 23}(1):\,57--65, (1986).

\bibitem{petz2007quantum}
D.~Petz.
\newblock {\em Quantum information theory and quantum statistics}.
\newblock Springer Science \& Business Media (2007).

\bibitem{petz1998contraction}
D.~Petz and M.~B. Ruskai.
\newblock {\em ``Contraction of generalized relative entropy under stochastic
  mappings on matrices''}.
\newblock Infinite Dimensional Analysis, Quantum Probability and Related Topics
  {\bf 1}(01):\,83--89, (1998).

\bibitem{polyanskiy2010channel}
Y.~Polyanskiy, H.~V. Poor, and S.~Verd{\'u}.
\newblock {\em ``Channel coding rate in the finite blocklength regime''}.
\newblock IEEE Transactions on Information Theory {\bf 56}(5):\,2307--2359,
  (2010).

\bibitem{raginsky2016strong}
M.~Raginsky.
\newblock {\em ``Strong data processing inequalities and $\Phi$-Sobolev
  inequalities for discrete channels''}.
\newblock IEEE Transactions on Information Theory {\bf 62}(6):\,3355--3389,
  (2016).

\bibitem{regula2022tight}
B.~Regula.
\newblock {\em ``Tight constraints on probabilistic convertibility of quantum
  states''}.
\newblock Quantum {\bf 6}:\,817, (2022).

\bibitem{regula2022postselected}
B.~Regula, L.~Lami, and M.~M. Wilde.
\newblock {\em ``Postselected quantum hypothesis testing''}.
\newblock arXiv preprint arXiv:2209.10550 , (2022).

\bibitem{renyi61}
A.~R{\'{e}}nyi.
\newblock {\em ``{On Measures of Information and Entropy}''}.
\newblock In {\em Proc. 4th Berkeley Symposium on Mathematical Statistics and
  Probability}, volume~1, pages 547--561, Berkeley, California, USA(1961).

\bibitem{ruskai1994beyond}
M.~B. Ruskai.
\newblock {\em ``Beyond strong subadditivity? Improved bounds on the
  contraction of generalized relative entropy''}.
\newblock Reviews in Mathematical Physics {\bf 6}(05a):\,1147--1161, (1994).

\bibitem{sason2015reverse}
I.~Sason.
\newblock {\em ``On reverse Pinsker inequalities''}.
\newblock Preprint, \href{http://arxiv.org/abs/1503.07118}{arXiv:\,1503.07118}
  (2015).

\bibitem{sason2015upper}
I.~Sason and S.~Verd{\'u}.
\newblock {\em ``Upper bounds on the relative entropy and R{\'e}nyi divergence
  as a function of total variation distance for finite alphabets''}.
\newblock In {\em 2015 IEEE Information Theory Workshop-Fall (ITW)}, pages
  214--218, (2015).

\bibitem{sason2016f}
I.~Sason and S.~Verd{\'u}.
\newblock {\em ``$ f $-divergence Inequalities''}.
\newblock IEEE Transactions on Information Theory {\bf 62}(11):\,5973--6006,
  (2016).

\bibitem{sharma2012strong}
N.~Sharma and N.~A. Warsi.
\newblock {\em ``On the strong converses for the quantum channel capacity
  theorems''}.
\newblock Preprint, \href{http://arxiv.org/abs/1205.1712}{arXiv:\,1205.1712}
  (2012).

\bibitem{temme2010chi}
K.~Temme, M.~J. Kastoryano, M.~B. Ruskai, M.~M. Wolf, and F.~Verstraete.
\newblock {\em ``The $\chi$ 2-divergence and mixing times of quantum Markov
  processes''}.
\newblock Journal of Mathematical Physics {\bf 51}(12):\,122201, (2010).

\bibitem{thompson1963certain}
A.~C. Thompson.
\newblock {\em ``On certain contraction mappings in a partially ordered vector
  space.''}.
\newblock Proceedings of the American Mathematical Society {\bf
  14}(3):\,438--443, (1963).

\bibitem{tomamichel2015quantum}
M.~Tomamichel.
\newblock {\em Quantum information processing with finite resources:
  mathematical foundations}.
\newblock volume~5, Springer (2015).

\bibitem{vershynina2019upper}
A.~Vershynina.
\newblock {\em ``Upper continuity bound on the quantum quasi-relative
  entropy''}.
\newblock \href{http://dx.doi.org/10.1063/1.5114619}{Journal of Mathematical
  Physics {\bf 60}(10):\,102201} (2019).

\bibitem{watanabe2012private}
S.~Watanabe.
\newblock {\em ``Private and quantum capacities of more capable and less noisy
  quantum channels''}.
\newblock Physical Review A {\bf 85}(1):\,012326, (2012).

\bibitem{wilde2018optimized}
M.~M. Wilde.
\newblock {\em ``Optimized quantum f-divergences and data processing''}.
\newblock \href{http://dx.doi.org/10.1088/1751-8121/aad5a1}{Journal of Physics
  A: Mathematical and Theoretical {\bf 51}(37):\,374002} (2018).

\bibitem{wilde2020amortized}
M.~M. Wilde, M.~Berta, C.~Hirche, and E.~Kaur.
\newblock {\em ``Amortized channel divergence for asymptotic quantum channel
  discrimination''}.
\newblock \href{http://dx.doi.org/10.1007/s11005-020-01297-7}{Letters in
  Mathematical Physics {\bf 110}:\,2277--2336} (2020).

\bibitem{wilde2014strong}
M.~M. Wilde, A.~Winter, and D.~Yang.
\newblock {\em ``Strong converse for the classical capacity of
  entanglement-breaking and Hadamard channels via a sandwiched R{\'e}nyi
  relative entropy''}.
\newblock \href{http://dx.doi.org/10.1007/s00220-014-2122-x}{Communications in
  Mathematical Physics {\bf 331}:\,593--622} (2014).

\bibitem{Zamanlooy2023}
B.~Zamanlooy and S.~Asoodeh.
\newblock {\em ``Strong Data Processing Inequalities for Locally Differentially
  Private Mechanisms''}.
\newblock In {\em 2023 IEEE International Symposium on Information Theory
  (ISIT)}, (2023).

\bibitem{zhang2016lower}
L.~Zhang, K.~Bu, and J.~Wu.
\newblock {\em ``A lower bound on the fidelity between two states in terms of
  their trace-distance and max-relative entropy''}.
\newblock \href{http://dx.doi.org/10.1080/03081087.2015.1057098}{Linear and
  Multilinear Algebra {\bf 64}(5):\,801--806} (2016).

\bibitem{zhou2017differential}
L.~Zhou and M.~Ying.
\newblock {\em ``Differential privacy in quantum computation''}.
\newblock In {\em 2017 IEEE 30th Computer Security Foundations Symposium
  (CSF)}, pages 249--262, (2017).

\end{thebibliography}

\end{document}